\pdfoutput=1
\documentclass[11pt]{article}

\usepackage[lined,boxed,ruled,norelsize,algo2e,linesnumbered]{algorithm2e}
\usepackage{thm-restate}
\usepackage{titlesec}

\titleformat{\part}[display]
  {\normalfont\bfseries}
  {\Large Part \Roman{part}}
  {0pt}                 % vertical gap
  {\LARGE}

\usepackage{fullpage}
\usepackage[english]{babel}
\usepackage[utf8x]{inputenc}
\usepackage[T1]{fontenc}
\usepackage{amsmath}
\usepackage{amssymb}
\usepackage{amsthm}
\usepackage{bbm}
\usepackage{mdframed}
\usepackage{bm}
\usepackage{parskip}
\usepackage{thm-restate}
\usepackage{booktabs}
\usepackage{array}

\usepackage[lined,boxed,ruled,norelsize,linesnumbered]{algorithm2e}
\SetKwProg{Fn}{Function}{:}{end}
\SetKwFunction{GradientOracle}{GradientOracle}
\SetKwComment{tcp}{$\triangleright$\ }{}

\SetCommentSty{graycomment}
\usepackage{tikz}
\usetikzlibrary{tikzmark,decorations.pathreplacing}

\newcommand{\PhaseBrace}[4][9.2cm]{%
  \begin{tikzpicture}[remember picture,overlay]
    \coordinate (edge) at
      ([xshift=-#1]pic cs:phase-right);
    \coordinate (top) at
      ([yshift=1.5ex]pic cs:#2);
    \coordinate (bottom) at
      ([yshift=-0.5ex]pic cs:#3);

    \draw[black!35, line width=0.9pt]
  (edge |- top) ++(-4pt,0)
  -- (edge |- top)
  -- node[midway,right=7pt,align=left,
          font=\small\ttfamily]{#4}
     (edge |- bottom)
  -- ++(-4pt,0);
  \end{tikzpicture}
}

\usepackage{hyperref}
\hypersetup{
	colorlinks=true,
	linkcolor=blue!70!black,
	citecolor=blue!70!black,
	urlcolor=blue!70!black
}
\usepackage{cleveref}
\usepackage{graphicx}
\usepackage{subcaption}

\usepackage{mathtools}
\usepackage{xcolor}

\newtheorem{lemma}{Lemma}
\newtheorem{fact}{Fact}
\newtheorem{corollary}{Corollary}
\newtheorem{proposition}{Proposition}
\newtheorem{definition}{Definition}

\newtheorem{remark}{Remark}

\newcommand{\defeq}{:=}

\newcommand{\norm}[1]{\left\lVert#1\right\rVert}
\newcommand{\norms}[1]{\lVert#1\rVert}

\newcommand{\normf}[1]{\left\lVert#1\right\rVert_{\textup{F}}}

\newcommand{\normsop}[1]{\lVert#1\rVert_{\textup{op}}}

\newcommand{\inprod}[2]{\left\langle#1, #2\right\rangle}

\newcommand{\normx}[2]{\norm{#1}_{#2}}
\newcommand{\normxd}[2]{\norm{#1}_{#2,*}}
\newcommand{\eps}{\varepsilon}
\renewcommand{\epsilon}{\varepsilon}
\newcommand{\lam}{\lambda}
\newcommand{\sig}{\sigma}

\newcommand{\R}{\mathbb{R}}

\newcommand{\N}{\mathbb{N}}
\newcommand{\diag}[1]{\textbf{\textup{diag}}\left(#1\right)}
\newcommand{\half}{\frac{1}{2}}

\newcommand{\E}{\mathbb{E}}

\newcommand{\Var}{\textup{Var}}

\newcommand{\ball}{\mathbb{B}}
\newcommand{\id}{\mathbf{I}}

\newcommand{\dd}{\textup{d}}
\usepackage{xcolor}
\definecolor{burntorange}{RGB}{191, 87, 0}
\definecolor{skyblue}{RGB}{107, 171, 232}

\newcommand{\tO}{\widetilde{O}}

\newcommand{\poly}{\textup{poly}}
\newcommand{\polylog}{\textup{polylog}}
\newcommand{\Par}[1]{\left(#1\right)}
\newcommand{\Brack}[1]{\left[#1\right]}
\newcommand{\Brace}[1]{\left\{#1\right\}}
\newcommand{\Abs}[1]{\left|#1\right|}

\newcommand{\ind}{\mathbb{I}}
\newcommand{\mzero}{\mathbf{0}}
\newcommand{\Sym}{\mathbb{S}}
\newcommand{\PSD}{\Sym_{\succeq \mzero}}
\newcommand{\simiid}{\sim_{\textup{i.i.d.}}}

\newcommand{\ma}{\mathbf{A}}
\newcommand{\mb}{\mathbf{B}}

\newcommand{\mh}{\mathbf{H}}
\newcommand{\mi}{\mathbf{I}}

\newcommand{\mm}{\mathbf{M}}

\newcommand{\ms}{\mathbf{S}}

\newcommand{\mv}{\mathbf{V}}

\newcommand{\0}{\mathbf{0}}
\newcommand{\1}{\mathbf{1}}
\newcommand{\va}{\mathbf{a}}
\newcommand{\vb}{\mathbf{b}}
\newcommand{\vc}{\mathbf{c}}

\newcommand{\vh}{\mathbf{h}}

\newcommand{\vm}{\mathbf{m}}

\newcommand{\vp}{\mathbf{p}}
\newcommand{\vq}{\mathbf{q}}
\newcommand{\vr}{\mathbf{r}}
\newcommand{\vs}{\mathbf{s}}
\newcommand{\vt}{\mathbf{t}}
\newcommand{\vu}{\mathbf{u}}
\newcommand{\vv}{\mathbf{v}}
\newcommand{\vw}{\mathbf{w}}
\newcommand{\vx}{\mathbf{x}}
\newcommand{\vy}{\mathbf{y}}
\newcommand{\vz}{\mathbf{z}}

\newcommand{\calB}{\mathcal{B}}

\newcommand{\calE}{\mathcal{E}}

\newcommand{\calK}{\mathcal{K}}

\newcommand{\calN}{\mathcal{N}}
\newcommand{\calO}{\mathcal{O}}
\newcommand{\calP}{\mathcal{P}}

\newcommand{\calS}{\mathcal{S}}
\newcommand{\calT}{\mathcal{T}}

\newcommand{\calX}{\mathcal{X}}
\newcommand{\calY}{\mathcal{Y}}
\newcommand{\calZ}{\mathcal{Z}}

\newcommand{\sfB}{\mathsf{B}}

\newcommand{\sfF}{\mathsf{F}}

\newcommand{\sfI}{\mathsf{I}}

\newcommand{\sfK}{\mathsf{K}}

\newcommand{\sfM}{\mathsf{M}}

\newcommand{\sfP}{\mathsf{P}}

\renewcommand{\phi}{\varphi}
\newcommand{\Exp}{\textup{Exp}}
\newcommand{\met}{\textup{m}}
\newcommand{\vxs}{\vx^\star}

\newcommand{\osc}{\mathrm{osc}}
\newcommand{\tpi}{\tilde{\pi}}

\newcommand{\TV}{\mathsf{TV}}

\newcommand{\tmu}{\tilde{\mu}}
\newcommand{\LLTProx}{\mathsf{LLTProx}}
\newcommand{\Backwards}{\mathsf{BackwardsHMC}}
\newcommand{\BackwardsMRW}{\mathsf{BackwardsMRW}}
\newcommand{\Forwards}{\mathsf{Forwards}}

\newcommand{\vxi}{\boldsymbol{\xi}}
\newcommand{\Unif}{\mathrm{Unif}}
\DeclareMathSymbol{\mhyphen}{\mathord}{AMSa}{"39}
\newcommand{\WarmLLTProx}{\mathsf{LLTProxHMC}}
\newcommand{\WarmLLTProxMRW}{\mathsf{LLTProxMRW}}
\newcommand{\FunctionalCooling}{\mathsf{FunctionalCooling}}
\newcommand{\AC}{\mathsf{AnalyticCenter}}
\newcommand{\vmu}{\boldsymbol{\mu}}
\newcommand{\msig}{\boldsymbol{\Sigma}}
\newcommand{\vgamma}{\boldsymbol{\gamma}}
\newcommand{\Law}{\mathrm{Law}}
\newcommand{\vDelta}{\boldsymbol{\Delta}}
\newcommand{\bF}{\bar{F}}
\newcommand{\baf}{\bar{f}}
\newcommand{\bnu}{\bar{\nu}}
\newcommand{\HMC}{\mathsf{HMC}}

\newcommand{\hvx}{\hat{\vx}}

\newcommand{\tsfP}{\widetilde{\sfP}}
\newcommand{\sfk}{\mathsf{K}}
\newcommand{\tsfk}{\widetilde{\sfk}}

\newcommand{\tsfB}{\widetilde{\sfB}}
\newcommand{\bmu}{\bar{\mu}}

\DeclareMathOperator*{\esssup}{ess\,sup}
\DeclareMathOperator*{\essinf}{ess\,inf}
\newcommand{\mXi}{\boldsymbol{\Xi}}
\newcommand{\mDelta}{\boldsymbol{\Delta}}
\newcommand{\PTERRYS}{\mathsf{PTERRYS}}

\title{Polytope Sampling at the Speed of Isoperimetry}
\author{Anming Gu\thanks{University of Texas at Austin, \texttt{anminggu@cs.utexas.edu}.} \and Bobby Shi\thanks{Simons Institute for the Theory of Computing, \texttt{bshi@berkeley.edu}. Work partly completed at UT Austin.} \and Kevin Tian\thanks{University of Texas at Austin, \texttt{kjtian@cs.utexas.edu}.} \and Hubert Yang\thanks{University of Texas at Austin, \texttt{hubert.yang@utexas.edu}.}}
\date{}

\begin{document}
\allowdisplaybreaks
% \pagenumbering{gobble}
\maketitle

\begin{abstract}

We give faster algorithms for sampling from the uniform distribution over a polytope. 
Our main result is an algorithm that, when applied to a polytope with $m$ constraints in $n$ dimensions, uses $m \cdot \poly\log(\frac m \delta)$ iterations, each running in time $O(mn^{\omega - 1})$, to return a sample within total variation distance $\delta$ to the uniform distribution. When $m =O(n)$ and $\delta = \poly(\frac 1 n)$, our runtime improves upon all prior works by an $\Omega(n^{2/3})$ factor. Our sampler is based on the recent \emph{functional stochastic localization} framework of \cite{GuST26}, and achieves an iteration count nearly matching the Poincar\'e constant of the uniform distribution with respect to the logarithmic barrier.
\end{abstract}

\thispagestyle{empty}
\newpage

\tableofcontents
\thispagestyle{empty}
\newpage

\pagenumbering{arabic}

\section{Introduction}\label{sec:intro}

Sampling in high dimensions is a central task in theoretical computer science and statistics, with applications ranging from volume computation, Bayesian estimation, and differential privacy. High-dimensional sampling algorithms have largely been developed along two complementary lines, distinguished by the regularity assumption used about the target distribution $\pi$. The first category, \emph{geometry-unaware sampling}, states the assumption with respect to the Euclidean norm, e.g., a convex set being well-rounded with respect to Euclidean balls, or a distribution having a well-conditioned potential $-\log \pi$. This approach is natural in many settings, and has led to strong quantitative analyses of various simple random walks \cite{DyerFK91, LovaszS93, LovaszV07, KookVZ24}, e.g., the ball walk and Gaussian proximal sampler. For general log-concave distributions, this approach achieves the best known convergence rates under value oracle access \cite{JiaLLV26, LiuPST26}.

The second category, \emph{geometry-aware sampling}, is relatively nascent. Here, the idea is to control the regularity of the target $\pi$ using a problem-specific assumption, with the hope of catering more closely to the intrinsic problem geometry. While less general, this approach offers opportunities to substantially speed up algorithms by exploiting problem-specific information. Interestingly, the geometry-aware approach has a rich counterpart in the adjacent field of optimization, e.g., the \emph{mirror descent} algorithm \cite{NemirovskiY83}, and as is particularly relevant to us, the theory of \emph{self-concordant barriers} in interior-point methods \cite{NesterovN94}.

The focus of this work is the fundamental problem of \emph{uniform sampling from a polytope}. Specifically, our goal is to sample from $\pi$, the uniform distribution over a polytope $\calX \defeq \{\vx \in \R^n: \ma \vx < \vb\}$, induced by a constraint matrix and vector $\ma \in \R^{m \times n}$, $\vb \in \R^m$. Beyond its intrinsic value as a subroutine in diverse applications spanning astronomy \cite{LubiniC12}, systems biology \cite{DeMMP15, HaraldsdottirCTFV17}, and psychology \cite{HeckD19}, uniform sampling from a polytope has inspired the development of many analysis and algorithmic frameworks in sampling \cite{KannanN12, ChenDWWY18, LeeV17, LeeV18, GatmiryKV24, Kook2024Gaussian}. This mirrors the central role of \emph{linear programming}, i.e., linear optimization over $\calX$, which has served as a classical testbed for interior-point methods and has seen substantial developments over several decades \cite{Dikin67, Karmarkar84, Vaidya89a, LeeS14, LeeS15, CohenLS19, BrandLSS20, BrandLLSSSW21}. 

Our main results concern the following question.

\begin{center}
\textit{How fast can we approximately sample from the uniform distribution over a polytope?} 
\end{center}

To provide some background, the literature on sampling algorithms addressing this question has largely followed the geometry-unaware vs.\ geometry-aware divide; we summarize the runtimes of existing algorithms in Table~\ref{tab:polytope-sampling}. For example, geometry-unaware algorithms \cite{MangoubiV19, JiaLLV26} are typically based on fast implementations of generic log-concave sampling algorithms, using polytope-specific strategies such as amortization or fast matrix multiplication. 

On the other hand, patterned off of interior-point methods, geometry-aware samplers have largely been built around ways of exploiting the Hessian structure of self-concordant barriers, such as the \emph{log barrier} or \emph{Lee-Sidford barrier} \cite{LeeS14, LeeS19}. An early work, \cite{KannanN12}, showed that the \emph{Dikin walk} (a log barrier geometry-adaptive sampler with Metropolis-Hastings correction), mixes in $O(mn)$ iterations from a warm start. This cost can be charged to two factors: $m$, an intrinsic isoperimetric constant relating the log barrier to the uniform distribution $\pi$ (formalized e.g., in Lemma~\ref{lem:isoperimetry}), and $\tau \approx n$, where $\tau$ is an (inverse) step size parameter chosen to ensure the Metropolis-Hastings filter does not apply too often. Subsequent work has used more sophisticated barriers \cite{LaddhaLV20, GatmiryKV24, ChenK26} coupled with algorithms for generating warm starts \cite{LovaszV06,CousinsV18,Kook2024Gaussian,KookV26}; focusing on the regime $m = O(n)$ for simplicity, the state-of-the-art iteration counts are $\widetilde{O}(n^2)$ \cite{ChenK26}\footnote{See also the concurrent works \cite{Lu26,SongZ26}.} and $\widetilde{O}(n^{41/16})$ \cite{Kook26} for polytope sampling under warm and cold starts, respectively.\footnote{We use $\widetilde{O}$ in this introduction only to suppress logarithmic factors in problem parameters.} For essentially all discussed algorithms (cf.\ Table~\ref{tab:polytope-sampling}), the per-iteration cost is $\tO(mn^{\omega - 1})$, i.e., $\tO(n^\omega)$ when $m = O(n)$, where $\omega \le 2.372$ \cite{AlmanDWXXZ25} is the runtime exponent of square matrix multiplication.\footnote{For notation used throughout this paper, see Section~\ref{ssec:notation}.}

\subsection{Our results}\label{ssec:results}

We provide two new algorithms for polytope sampling, which improve upon the state-of-the-art across a broad range of parameters. For example, when $m = O(n)$, our basic sampler in Theorem~\ref{thm:main_part1} improves all existing runtimes by a factor of at least $n^{1/6}$ even if the alternative algorithm is granted a warm start, and our final sampler in Theorem~\ref{thm:main_part2} improves upon all existing runtimes by a factor of at least $n^{2/3}$. Both of our results place no warmness assumptions on the starting point.

Our algorithms are built on a common framework: \emph{functional stochastic localization} (FSL), developed recently by \cite{GuST26}. FSL is a geometry-aware variant of the \emph{Gaussian proximal sampler} of \cite{LeeST21}, whose iterations are governed by a primal-dual pair of functions $(\phi, \psi)$, where $\psi$ is the \emph{log-Laplace transform} of $\phi$ (cf.\ Definition~\ref{def:llt}). Crucially, high-accuracy implementation of the ideal FSL iterations bypasses the need for a bias correction step, which bottlenecked e.g., the Metropolis-Hastings adjusted Dikin walk \cite{KannanN12} and related analyses using the Lee-Sidford barrier. As inspiration for our proximal sampling approach, in geometry-unaware settings, this strategy has seen success in sampling from convex bodies through the \emph{In-and-Out} sampler of \cite{KookVZ24}.

Our algorithms leverage FSL through the crucial observation that the standard log barrier of a polytope is a log-Laplace transform of a natural potential induced by a pushforward of exponential random variables. FSL requires implementation of two oracles, the \emph{forwards} and \emph{backwards oracles}, simulating draws from conditional distributions. When $\psi$ is taken to be the log barrier, we give an exact implementation of the forwards oracle in Proposition~\ref{prop:forward}.

After providing this setup, the rest of our paper is divided into two parts. Part~\ref{part:sqrtn} applies FSL with the log barrier, and implements a high-accuracy backwards oracle with a uniform guarantee over dual points $\vy$ drawn from the forwards oracle. It culminates in the following result.

\begin{restatable}{theorem}{restatepartI}\label{thm:main_part1}
Let $\delta \in (0, \half)$. Algorithm~\ref{alg:annealing}, using Algorithm~\ref{alg:warm-start-LLT} as $\calO$ with
\[\tau = \left\lceil C\Par{\sqrt{n\log \frac m \delta} + \log \frac m \delta}\right\rceil\]
for a universal constant $C > 0$, outputs $\vx$ satisfying $\TV(\Law(\vx), \pi) \le \delta$. Moreover, if the algorithm is given $\vx_0 \in \calX$ with $\ball_2(\vx_0, r) \subset \calX \subset \ball_2(\vx_0, R)$ for known $0 < r < R$, it runs in time
\[O\Par{m^2n^{\omega - \half}\log^4\Par{\frac m \delta} + m^{\frac 3 2}n^{\omega - 1}\log\Par{\frac{mR\log(\frac 1 \delta)}{r}}}.\]
\end{restatable}

To briefly explain the two terms in Theorem~\ref{thm:main_part1}'s runtime, the first is the cost of our sampler, given an approximate analytic center (i.e., minimizer to the log barrier). The second term is low-order when the aspect ratio $\frac R r = \poly(m)$, and is the cost of finding such an approximate center, whose cost is often not explicit in prior works on polytope sampling, e.g., \cite{KannanN12, LeeV17, ChenDWWY18, GatmiryKV24, Kook26, ChenK26}. We provide an explicit runtime for analytic center approximation in Proposition~\ref{prop:ac}, that depends logarithmically on a natural geometric quantity, $\frac R r$. Notably, existing linear programming algorithms always incur a dependence on a representation-dependent parameter, e.g., there is no known strongly polynomial-time algorithm for finding a feasible point in a polytope.

The runtime of Theorem~\ref{thm:main_part1} suffers a $\tO(\sqrt{n})$ overhead over the $m$ factor arising from the native isoperimetry, i.e., Poincar\'e constant, of $\pi$ with respect to the log barrier (Lemma~\ref{lem:isoperimetry}). This overhead owes to the choice of the step size parameter $\tau$ that our backwards oracle in Part~\ref{part:sqrtn} tolerates. In Part~\ref{part:tilde1}, we develop a second backwards oracle that accepts a much milder $\tau = O(\log \frac m \delta)$, at the cost of a weaker approximation guarantee. This guarantee holds instancewise depending on the particular conditional distribution, but is strong enough on average to be used with the recent weak Poincar\'e inequality convergence analysis framework of \cite{AndrieuLPW24, AndrieuLPW26}.

\begin{restatable}{theorem}{restatepartII}\label{thm:main_part2}
Let $\delta \in (0, \half)$. $\PTERRYS$ (Algorithm~\ref{alg:main}), or Algorithm~\ref{alg:annealing}, using Algorithm~\ref{alg:warm-start-LLT-MRW} as $\calO$ with
\[\tau = \left\lceil C\log \frac{m}{\delta}\right\rceil\]
for a universal constant $C > 0$, outputs $\vx$ satisfying $\TV(\Law(\vx), \pi) \le \delta$. Moreover, if the algorithm is given $\vx_0 \in \calX$ with $\ball_2(\vx_0, r) \subset \calX \subset \ball_2(\vx_0, R)$ for known $0 < r < R$, it runs in time
\[O\Par{m^2n^{\omega - 1}\log^3\Par{\frac m \delta} + m^{\frac 3 2}n^{\omega - 1}\log\Par{\frac{mR\log(\frac 1 \delta)}{r}}}.\]
\end{restatable}

As discussed, Theorems~\ref{thm:main_part1} and~\ref{thm:main_part2} use the same outer loop framework (Algorithm~\ref{alg:annealing}, our functional cooling procedure for FSL), and differ only in how they implement the approximate backwards oracle. The name of our final algorithm stands for PolyTope isopERimetRY Sampler ($\PTERRYS$),\footnote{We make no comment on any apparent resemblance to a fast food chain whose only locations are within a 3-hour drive of the authors' home institution, or inspiration taken from the naming convention of \cite{KookVZ24}.} reflecting that its iteration count nearly matches the log barrier Poincar\'e constant of $\pi$.

\begin{remark}
In light of our approach, a natural idea is to apply the FSL framework with a different LLT than the log barrier. Unfortunately, the most commonly used $\tO(n)$-self-concordant barrier for a polytope, the Lee-Sidford barrier \cite{LeeS14, LeeS19}, is provably not a log-Laplace transform. Interestingly, several other barrier functions used in linear programming, including the universal barrier \cite{NesterovN94} and hybrid barrier \cite{VaidyaA93, Anstreicher96}, are also LLTs; we leave this direction as open for future work.
\end{remark}

\begin{table}[t!]
\newcommand{\matmul}{mn^{\omega-1}}
\centering
\small
\renewcommand{\arraystretch}{1.15}
\begin{tabular}{@{}lcc@{}}
\toprule
\textbf{Result}
& \textbf{Iterations}
& \textbf{Cost/iteration} \\
\midrule

\multicolumn{3}{l}{\textbf{Geometry-unaware}} \\
\midrule

Hit-and-Run,
% Theorem 2.1, 
\cite{LovaszV06}
    & $n^4$
    & $mn$ \\

Ball walk, \cite{MangoubiV19}  & $n^{9/2}$ & $m$ amortized\\

Ball walk, \cite{JiaLLV26} & $n^{7/2}$ & $mn^{0.186}$ amortized \\

\midrule
\multicolumn{3}{l}{\textbf{Geometry-aware}} \\
\midrule
Dikin + log barrier, 
% Theorem 1,
\cite{KannanN12}
    & $mn^2$
    & $\matmul$ \\

Geodesic walk + log barrier, 
% Theorem 1,
\cite{LeeV17}$^*$
    & $mn^{3/4}$
    & $\matmul$ \\

Dikin + Vaidya barrier, 
% Corollary 4,
\cite{ChenDWWY18}
    & $m^{1/2}n^{5/2}$
    & $\matmul$ \\

Dikin + LS barrier, 
% Theorem 2, 
\cite{ChenDWWY18}$^*$
    & $n^{5/2}$
    & $\matmul$ \\

RHMC + log barrier, 
% Table 2, 
\cite{LeeV18}
    & $mn^{2/3}$
    & $\matmul$ \\

Dikin + LS barrier,
% Theorem 2, 
\cite{Kook2024Gaussian}
    & $n^{23/8}$
    & $\matmul$ \\

RHMC + LS barrier, 
% Corollary 1.1.1, 
\cite{GatmiryKV24}
    & $m^{1/3}n^{4/3}$
    & N/A\\

Dikin + LS barrier, 
% Corollary 1.2, 
\cite{Kook26}
    & $n^{41/16}$
    & $\matmul$ \\

Dikin + LS barrier, 
% Corollary 2, 
\cite{ChenK26}$^*$
    & $n^2$
    & $\matmul$ \\

\midrule

\textcolor{burntorange}{FSL with $\tau = \widetilde{O}(\sqrt{n})$ + log barrier, Theorem~\ref{thm:main_part1}}
    & $m\sqrt n$
    & $\matmul$ \\

\textcolor{burntorange}{FSL with $\tau = \widetilde O(1)$ + log barrier, Theorem~\ref{thm:main_part2}}
    & $m$
    & $\matmul$ \\

\bottomrule
\end{tabular}

\caption{
Algorithms for uniform sampling on polytopes, suppressing all polylogarithmic factors. Here $\omega\le 2.372$ \cite{AlmanDWXXZ25} is the runtime exponent of square matrix multiplication, and $^*$ indicates that a warm start is required. Whenever a paper gives both warm and cold started results, we only state the cold start result. Geometry-unaware bounds include cost of isotropic rounding, e.g., rounding the covariance to $\approx \mi_n$; the runtime of \cite{JiaLLV26} is stated with the rectangular matrix multiplication bound $\omega(1,1,\frac 7 2)\le 4.686$; \cite{GatmiryKV24} requires exact implementation for RHMC. This table omits results that contain gaps, see \cite{Kook26} for a discussion.
}
\label{tab:polytope-sampling}

\end{table}

\subsection{Our techniques}\label{ssec:techniques}

In this section, we provide a technical overview of our algorithms. 

Our overview splits naturally by the sections of the paper. Section~\ref{sec:prelims} introduces the ideal FSL-based proximal sampler of \cite{GuST26} and an exact implementation of the forward oracle corresponding to the log barrier function. Section~\ref{sec:part_1_backwards_step} provides an initial backwards oracle implementation that has a uniform approximation guarantee for all conditional distributions, and applies when $\tau \ge \sqrt{n}\cdot \polylog(\frac m \delta)$. Section~\ref{sec:cooling} then develops a functional cooling-based framework for satisfying the warm start needed by FSL, adapted from results in \cite{Kook2024Gaussian}. Finally, Section~\ref{sec:part_2_backwards_step} develops a second backwards oracle that has weaker convergence guarantees, but is tolerant of smaller step sizes $\tau$; we combine this new implementation with Sections~\ref{sec:prelims} and~\ref{sec:cooling} to prove Theorem~\ref{thm:main_part2}.

\textbf{Functional stochastic localization.} Prior state-of-the-art algorithms (in different parameter regimes) for polytope sampling are based on either the Dikin walk of \cite{KannanN12}, equipped with a suitable self-concordant barrier (Definition~\ref{def:sc}) \cite{Kook26,ChenK26}, or more sophisticated processes such as Riemannian Hamiltonian Monte Carlo (RHMC) \cite{LeeV18, GatmiryKV24}. Notably, algorithms based on the Dikin walk are bottlenecked in step size by a bias correction step (i.e., Metropolis-Hastings adjustment) due to the inexact proposal distributions used. On the other hand, while RHMC strategies bypass this bias, their analyses appear bottlenecked at step sizes incurring an $\tO(n^{2/3})$ overhead, as discussed in Section 1.2 of \cite{LeeV18} and the introduction of \cite{GatmiryKV24}.

Our algorithms are based on a different framework for debiasing polytope sampling: the FSL-based proximal sampler of \cite{GuST26}. This algorithm and its analysis through the FSL process were originally developed to generalize a non-Euclidean proximal sampler proposed by \cite{GopiLLST25}, which uses log-Laplace transforms (LLTs) to capture problem-specific geometry. In particular, FSL-based proximal samplers apply Gibbs (forwards-backwards) sampling over the following joint-density over an extended, primal-dual space $(\vx, \vy) \in \calX \times \calY$:
\begin{equation}\label{eq:joint_density_intro}
\tpi_\tau(\vx, \vy):=\exp\Par{\inprod{\vx}{\vy}-\varphi^{*\tau}(\vy)-\tau\psi(\vx)}\pi(\vx),    
\end{equation}
where $*$ denotes convolution, $\psi$ is the LLT of $\varphi$, and $\tau\in \N$ is a step size parameter. We give a self-contained introduction to this framework in Section~\ref{ssec:fsl}.

Our algorithm starts from the simple observation that the standard log barrier of a polytope $\calX$ (where $\{\va^i\}_{i \in [m]}$ are the rows of the constraint matrix), 
\[\psi(\vx) = -\sum_{i \in [m]} \log(\vb_i - \inprod{\va^i}{\vx}),\]
is in fact an LLT, and thus compatible with the algorithmic framework of \cite{GuST26}. This is shown formally in Lemma~\ref{lem:phi_def}: its proof essentially combines the fact that LLTs are closed under addition through convolution of their dual measures (Fact~\ref{fact:llt_conv}), and the integral identity
\begin{equation}\label{eq:forward_logZ}\int_0^\infty \exp\Par{-t\Par{\vb_i - \inprod{\va^i}{\vx}}}\dd t = \exp\Par{-\log\Par{\vb_i - \inprod{\va^i}{\vx}}},\end{equation}
i.e., a single summand in $\psi$ is the log-normalizing constant of an exponential random variable.

To specialize the \cite{GuST26} framework to the log barrier geometry, we require implementations of the forwards and backwards steps of Gibbs sampling, i.e., draws from both conditional distributions of $\tpi_\tau$. Proposition~\ref{prop:forward} builds on \eqref{eq:forward_logZ} to give a simple exact forwards step implementation. The bulk of the remaining work is to design and analyze approximate backwards oracles for this setup.

\textbf{Backwards oracle with $\tau = \widetilde{O}(\sqrt{n})$.} The main challenge of implementing backwards oracles is to tolerate as small of a step size $\tau$ as possible, as the outer loop complexity of FSL incurs a linear dependence on $\tau$. A na\"ive baseline is to set $\tau \approx n$, at which point a direct calculation shows that most of the backwards oracle mass is well-approximated by a Gaussian around the analytic center (analogous to Section 4.1 of \cite{LeeST21}). Unfortunately, this direct approach does not improve upon the $\tO(mn)$ iteration count of the Dikin walk \cite{KannanN12}.

Our backwards oracle in Part~\ref{part:sqrtn} takes inspiration from \cite{FanYC23}, which centers a Gaussian proposal at the mode of the backwards conditional to enable approximate rejection sampling at smaller step sizes $\tau = \tO(\sqrt{n})$. The key is to argue that the backwards conditional law
\[\pi_\tau^{X \mid Y = \vy}(\vx)\propto \exp\Par{\inprod{\vy}{\vx} - \tau\psi(\vx)}\pi(\vx)\]
concentrates around $\vxs_{\vy, \tau}$, where $\vxs_{\vy, \tau}$ is the mode to the backwards conditional potential $\tau\psi - \inprod{\vy}{\cdot}$. The Hessian of this potential is a multiple of
\[\nabla^2 \psi(\vx) = \ma^\top \ms(\vx)^{-2} \ma, \text{ where } \ms(\vx) \defeq \diag{\vb - \ma \vx}\]
is a diagonal matrix induced by the \emph{slacks} $\vs(\vx) \defeq \vb - \ma \vx$ at a point $\vx \in \calX$. The bulk of our technical work in Section~\ref{sec:part_1_backwards_step} establishes that $\pi_\tau^{X \mid Y = \vy}$ specifically concentrates in a \emph{relative slack box} \eqref{eq:relative_slack_box} around the mode $\vx^\star_{\vy, \tau}$. Namely,
when $\tau \ge \sqrt{n} \cdot \polylog(n)$, we show that 
\[\half \vs(\vxs_{\vy, \tau}) \le \vs(\vx) \le \frac 3 2 \vs(\vxs_{\vy, \tau})\]
holds entrywise with high probability over $\vx \sim \pi_\tau^{X \mid Y = \vy}$.
This implies that $\nabla^2 \psi(\vx)$ approximates $\nabla^2 \psi(\vxs_{\vy, \tau})$, the Hessian at the mode, in a high-probability region of the backwards conditional distribution. After applying an affine transformation to whiten the space, we replace the pushforward potential with a globally-smooth surrogate and apply a standard well-conditioned sampler, e.g., the HMC algorithm of \cite{ZhangAC26}, to obtain our high-accuracy backwards oracle.

To exploit this concentration result, we require a complementary approximate mode-finding subroutine, provided in Corollary~\ref{cor:mode_finding}. Our mode-finding algorithm uses the prior backwards draw $\vx$ as an initial guess for the potential optimization problem, and we show that under an $O(1)$ warm start, $\vx$ is sufficiently close to the mode in a suitable $\ell_4$ norm of the relative slacks. Crucially, this is sufficient for a single step of Newton's method to bring the iterate to a constant-radius Dikin ellipsoid centered at the mode, after which standard local convergence guarantees apply. 

\textbf{Functional cooling.} We take a brief aside to discuss a complementary algorithm provided by our paper: a \emph{functional cooling} framework that provides $\chi^2$-warm starts to our proximal sampler. This is critical in our applications, as a na\"ive $\exp(\Omega(n))$-quality cold (feasible) start incurs additional dimension factors in the runtime of FSL (cf.\ Proposition~\ref{prop:prox-sample-guarantee}). Providing an explicit start without a dimension-dependent overhead was left open by the earlier work of \cite{GuST26}.

In the setting of FSL with the logarithmic barrier, we 
extend an annealing strategy from \cite{Kook2024Gaussian}, which we term functional cooling, to provide explicit warm starts.
This strategy considers a sequence of distributions $\pi_0,\dots, \pi_T$, where $\pi_0$ is explicit or is easy to sample from,
$\pi_T=\pi$, and each consecutive distribution is an $O(1)$ warm start for the subsequent one in the $\chi^2$ divergence. We observe that the techniques in \cite{Kook2024Gaussian} extend naturally to the setting of FSL,\footnote{We remark that $\cite{Kook2024Gaussian}$ consider a wide variety of other self-concordance parameters, along with their applications to general settings, e.g., general exponential distributions over polytopes and truncated positive semidefinite cones.} and its iteration complexity dramatically simplifies. In particular,  the distributions along the cooling schedule are $\pi_t \propto \exp(-t\psi)\pi$ for some cooling parameter $t\in\R_{\ge0}$, which naturally inherits a faster FSL mixing rate through the Brascamp-Lieb inequality (Fact~\ref{fact:bl}). Combined with our backwards oracle in Proposition~\ref{prop:backwards_part1}, Theorem~\ref{thm:main_part1}, our polytope sampler with $\widetilde O(m\sqrt{n})$ iterations, follows.

\textbf{Backwards oracle with $\tau = \widetilde{O}(1)$.} Our final algorithm, Algorithm~\ref{alg:main}, follows by combining Sections~\ref{sec:prelims} and~\ref{sec:cooling} with a different approximate backwards oracle implementation, the focus of Section~\ref{sec:part_2_backwards_step}. Our second backwards oracle implementation also considers stable regions with respect to the inherent Hessian geometry, but rather than measure stability with respect to the potential's mode, it is measured directly with respect to the dual draw from the forwards oracle. 

Our strategy is to first prove an \emph{average-case} bound on the stable region, over a random forwards draw (Lemma~\ref{lem:concentration_part2}), whose guarantee is meaningful even when $\tau = O(\log \frac m\delta)$.
We then show that a simple Metropolized random walk (Algorithm~\ref{alg:backwards_part2}) in appropriate whitened coordinates satisfies a \emph{weak Poincar\'e inequality} (Lemma~\ref{lem:mrw}), compatible with the recent analysis framework of \cite{AndrieuLPW24, AndrieuLPW26}. More precisely, Part~\ref{part:tilde1} parameterizes the dual point $\vy$ as $\vy = \ma^\top \vgamma$ (cf.\ the implementation of the forward process in Proposition~\ref{prop:forward}). We then define a proxy Hessian $\mh_{\vgamma} \defeq \ma^\top \diag{\vgamma}^2 \ma$, intuitively simulating relative slacks $\vs(\vx)$ as being proportional to $\vgamma^{-1}$. Lemma~\ref{lem:concentration_part2}'s key observation is that the conditional measure of the stable region, 
\[\Brace{\vx \in \calX: \frac{\tau}{2\vgamma_i} \le \vs_i(\vx) \le \frac{3\tau}{2\vgamma_i} \text{ for all } i \in [m]},\]
whose definition is conditional on the dual point $\vgamma$, is large on average over the marginal of $\vgamma$. Our analysis ``gives up'' on convergence in the unstable region, but by comparing the average instancewise $\chi^2$ contraction to that of the ideal proximal sampler, we nonetheless obtain rapid convergence. To conclude, we utilize the cooling framework in Section~\ref{sec:cooling} to obtain Theorem~\ref{thm:main_part2}.

\section{Preliminaries}\label{sec:prelims}

In this section we discuss some preliminaries. In Section~\ref{ssec:notation}, we set up notation used throughout the paper. In Section~\ref{ssec:fsl}, we introduce the (abstract) sampling framework, functional stochastic localization, with its specialization to the log barrier on the polytope and sampling the forward step in Section~\ref{ssec:logbarrier}. In Section~\ref{ssec:loghelper}, we establish isoperimetry of the uniform distribution, which facilitates mixing time analysis of our sampler.

\subsection{General notation}\label{ssec:notation}

We work in the real RAM model, where we assume scalar operations can be done in $O(1)$ time. We use $\widetilde{O}(\cdot)$ to hide logarithmic factors in all problem parameters; all logs are natural.
For $n\in\N$ we denote $[n]\defeq \Brace{i\in\N\mid i\le n}$. We use $\omega \ge 2$ to denote the runtime exponent of square matrix multiplication ($\omega < 2.371\ldots$ \cite{AlmanDWXXZ25} is the current best peer-reviewed bound).

We use lowercase and uppercase boldface, e.g., $\va, \ma$ to denote vectors and matrices, respectively. We use $\va \circ \vb$ to denote the elementwise (Hadamard) product of two vectors. When a scalar operation (e.g., $\cdot^{-1}$) is written as applied to a vector  (e.g., $\vv^{-1}$), it always means applying the operation entrywise. We  use $\va_i$ to denote the $i^{\text{th}}$ entry of a vector $\va$. We use $\0_d,\1_d$ to respectively denote the all-zeroes and all-ones vectors in $\R^d$, $\0_{m\times n}$ for the $m \times n$ all-zeroes matrix, and $\mi_d$ to denote the identity matrix in $\R^d$. 
For functions $f,g$, $f\succeq g$ means $f-g$ is convex and for matrices $\ma,\mb$, $\ma\succeq \mb$ means $\ma -\mb$ is positive semidefinite. We denote the set of symmetric $d \times d$ matrices by $\Sym^{d \times d}$ and the positive semidefinite subset by $\PSD^{d \times d}$.

When $\ma \in \PSD^{d \times d}$ and $\vv \in \R^d$, we use $\norm{\vv}_{\ma} \defeq \sqrt{\vv^\top \ma \vv}$ to denote the seminorm induced by $\ma$. We let $\ma^\dagger$ denote the pseudoinverse of $\ma \in \R^{m \times n}$, so that $\ma^\dagger \ma$ projects onto the row span of $\ma$. We denote the Euclidean ball of radius $r > 0$ around a point $\vx \in \R^n$ by $\ball_2(\vx, r)$.

\textbf{Probability.} We use $\calP(\calX)$ to denote the set of probability measures over a Borel space $\calX$. When a measure admits a density with respect to the underlying reference measure (taken to be Lebesgue when $\calX$ is a linear subspace of $\R^d$), we identify it with its density. For $\pi \in \calP(\calX)$ and measurable $S \subseteq \calX$, we let $\pi(S)$ denote $\Pr_{x \sim \pi}[x \in S]$. We let $\E[\cdot]$ and $\Var[\cdot]$ denote expectation and variance. For a measurable $f: \calX \to \calX'$ and a measure $\nu \in \calP(\calX)$, $f_{\#} \nu(\cdot) \in \calP(\calX')$ denotes the pushforward of $\nu$ by $f$, i.e., $\nu(f^{-1}(\cdot))$. 
For two measures $\mu, \nu$ on $\R^d$, $\mu \ast \nu$ denotes their convolution.
For an event $\calE$, we let $\ind_\calE$ denote the corresponding 0-1 indicator random variable. 

We denote Markov kernels and Markov operators in uppercase sans-serif, e.g., $\sfK$. For a Markov kernel $\sfK$, we let the measure $\mu \sfK$ be defined in the natural way, i.e., $(\mu \sfK)(S) \defeq \int \sfK(x,S)\mu(\dd x)$. For a Markov kernel $\sfK$, when the context is clear, we overload notation by also using the same symbol to denote the natural Markov operator associated with $\sfK$. Namely, we may write $\sfK f$ to mean the function defined as $(\sfK f)(x) \defeq \int f(y)\sfK(x, \dd y)$. For a fixed $x$, we write $\sfK(x, \cdot)$ to mean the probability measure induced by the Markov kernel $\sfK$ at $x$. We let $\sfK(\cdot, S)$ denote the probability that the transition from $\cdot$ assigns to a measurable $S$.

We use the following distances between  $\mu,\nu\in\calP(\calX)$ with $\mu\ll \nu$.  The $\chi^2$-divergence is $\chi^2(\mu\|\nu) \defeq \E_\nu[(\frac{\dd \mu}{\dd\nu}-1)^2]$ and their TV distance is $\TV(\mu,\nu) \defeq \frac12 \int_\calX |\dd\mu - \dd \nu|$.

We denote the Dirac measure with respect to a point $\vx \in \R^d$ by $\delta_{\vx}$. For a discrete set $S$, we denote the uniform distribution over it by $\Unif(S)$. We denote the multivariate Gaussian with mean $\vmu \in \R^d$ and covariance $\msig \in \PSD^{d \times d}$ by $\calN(\vmu, \msig)$. For any $\lambda > 0$, we denote the exponential distribution on $\R_{> 0}$ with density $\lam\exp(-\lambda \cdot)$ by $\Exp(\lam)$.\footnote{This changes the standard definition by a measure-zero subset (i.e., $\{0\}$), for compatibility with our conventions.} We assume sample access to one-dimensional Gaussian and exponential distributions in $O(1)$ time.

\textbf{Log-Laplace transforms.} Our algorithm and analysis heavily rely on the following construction.

\begin{definition}\label{def:llt}
    The log-Laplace transform (LLT) of $\varphi:\R^d\to \R \cup \{\infty\}$ is defined as \[
    \varphi^\sharp(\vx) \defeq \log\Par{\int \exp\Par{\inprod{\vx}{\vy}-\varphi(\vy)}\dd \vy}.
    \]
    When $\phi$ is finite only on a lower-dimensional linear subspace $L \subseteq \R^d$, we interpret $\exp(-\phi)$ as a density with respect to Lebesgue measure on $L$, denoted by $\dd \vy$ in an abuse of notation.
\end{definition}

For a measure $\nu$, we define the tilted probability density induced by a \emph{linear tilt} $\vy \in \R^d$ by
\[
        \mathcal{T}_\vy \nu(\vx)\defeq \frac{\exp(\inprod{\vy}{\vx}) \nu(\vx)}{\int \exp(\inprod{\vy}{\vz}) \nu(\vz)\,\dd\vz},
\]
whenever the denominator is finite.
In this way, if $\psi \defeq \phi^\sharp$,
\begin{equation}\label{eq:llt_density}
    \calT_\vx\exp(-\phi(\vy)) = \exp\Par{\inprod{\vx}{\vy}-\varphi(\vy)-\psi(\vx)},
\end{equation}
which is indeed a probability density as the normalization constant is exactly given by $\exp(-\psi)$. More generally, if $f$ is any integrable nonnegative function over $\R^d$, we let $\calT_{\vx} f$ be the density that is $\propto \exp(\inprod{\vx}{\cdot}) f$, where we use $\propto$ to indicate proportionality.

We also make extensive use of the following fact regarding the effect of convolution on LLTs. 

\begin{fact}[Lemma 1, \cite{GuST26}]\label{fact:llt_conv}
    Let $\phi: \R^d \to \R \cup \{\infty\}$ and $\xi: \R^d \to \R \cup \{\infty\}$ be such that $\mu = \exp(-\phi)$, $\nu = \exp(-\xi)$ are densities on $\R^d$. Then if $\chi \defeq -\log(\mu * \nu)$,  $\chi^\sharp = \phi^\sharp + \xi^\sharp$. 
\end{fact}

As a consequence of Fact~\ref{fact:llt_conv}, linear tilts commute with convolution: if $\mu = \exp(-\phi)$ and $\nu = \exp(-\xi)$,
\begin{equation}\label{eq:tilt_convolve}
\calT_{\vx}(\mu \ast \nu)(\vy) = \frac{\exp\Par{\inprod{\vx}{\vy}}(\mu \ast \nu)(\vy)}{\exp\Par{\phi^\sharp(\vx) + \xi^\sharp(\vx)}} = (\calT_{\vx}(\mu) \ast \calT_{\vx}(\nu))(\vy).
\end{equation}

When $\phi$ is a potential (i.e., negative log-density rather than density) clear from context, and $\tau \in \N$, we overload notation and let $\phi^{\ast \tau}$ refer to the potential (negative log-density) of the $\tau$-fold convolution of the density $\exp(-\phi)$. With this notation, Fact~\ref{fact:llt_conv} and \eqref{eq:llt_density} imply that $(\phi^{\ast \tau})^\sharp = \tau\phi^\sharp$.

\subsection{Functional stochastic localization}\label{ssec:fsl}

Our algorithm is based on a framework called \emph{functional stochastic localization} (FSL), introduced by 
\cite{GuST26} to analyze a non-Euclidean proximal sampler proposed in \cite[Algorithm~\ref{alg:prox-sample}]{GopiLLST25}.

We summarize the LLT-based proximal sampler in Algorithm~\ref{alg:prox-sample} and briefly describe it here. Let $\pi\in \calP(\R^d)$ be a target distribution. The algorithm considers Gibbs sampling over the joint density
\begin{equation}\label{eq:joint-density}
   \tpi_\tau(\vx, \vy):=\exp\Par{\inprod{\vx}{\vy}-\varphi^{*\tau}(\vy)-\tau\psi(\vx)}\pi(\vx),
\end{equation}
over $(\vx, \vy) \in \R^d \times \R^d$, where $\tau\in\N$ is a regularization parameter (akin to an inverse step-size), and $\psi = \phi^\sharp$. As a consequence of \eqref{eq:llt_density}, $\pi^X$, the $\vx$-marginal of $\tilde{\pi}_\tau$ is simply $\pi$. We use $\pi_{\tau}^{X|Y=\vy}$ and $\pi_{\tau}^{Y|X=\vx}$ to denote the conditional densities. In general, we reserve the use of the tilde notation to denote that $\tmu$ is a joint density on $(X, Y)$, and refer to its marginals by $\mu^X$ and $\mu^Y$.

The main result of \cite{GuST26} is a mixing time guarantee on Algorithm~\ref{alg:prox-sample} for sampling from the target distribution $\pi$, assuming that it satisfies an appropriate \emph{functional Poincar\'e inequality}.

\begin{definition}[$\psi$-PI, Definition 3, \cite{GuST26}]\label{def:pi}
Let $\pi \in \calP(\calX)$ for an open $\calX \subseteq \R^d$, and let $\psi: \calX \to \R$ be strictly convex and twice-differentiable. We say that $\pi$ satisfies a \emph{$\psi$-Poincar\'e inequality with constant $\alpha$} (abbreviated as \emph{$\alpha$-$\psi$-PI}) if, for all smooth functions $f: \calX \to \R$,
\[\Var_{\pi}[f] \le \frac 1 \alpha \E_\pi\Brack{\norm{\nabla f}_{(\nabla^2 \psi)^{-1}}^2}.\]
\end{definition}

\begin{proposition}[Theorem 1, \cite{GuST26}]\label{prop:prox-sample-guarantee}
Assume $\pi$ satisfies an $\alpha$-$\psi$-PI and $\vx_0 \sim \mu_0$ with $\mu_0\ll \pi$. Then, for every $k\in\N$, $\vx_k$ in  Algorithm~\ref{alg:prox-sample} has a density $\mu_k$ that satisfies \[
    \chi^2(\mu_k\|\pi)\le \frac{1}{\Par{1+\frac\alpha\tau}^{2k}}\chi^2(\mu_0\|\pi).
    \] 
\end{proposition}

As a sufficient condition, we recall that relative strong convexity implies a functional PI.

\begin{fact}[Theorem 4.1, \cite{BrascampL76}]\label{fact:bl}
    If $-\log\pi\succeq \alpha\psi$, then $\pi$ satisfies an $\alpha$-$\psi$-PI.
\end{fact}

\begin{algorithm2e}[t!]
\caption{$\LLTProx(\vx_0, \tau, K)$}
\label{alg:prox-sample}
\DontPrintSemicolon
\textbf{Input: }initial point $\vx_0$, step size $\tau \in \N$, number of iterations $K$ \;
\For{$k=0, 1, \dots, K-1$}{
    $\vy_k\sim \pi^{Y\mid X=\vx_k}= \mathcal{T}_{\vx_k}\exp\Par{-\varphi^{*\tau}}$ \tcp*{\emph{Forwards} step}\label{line:prox-sample-forward}
    $\vx_{k+1} \sim \pi^{X \mid Y = \vy_k} = \calT_{\vy_k}\exp(-\tau\psi)\pi$\tcp*{\emph{Backwards} step}\label{line:prox-sample-backward}
}
\textbf{Return: }$\vx_K$
\end{algorithm2e}

\subsection{Specialization to logarithmic barrier}\label{ssec:logbarrier}

We now specialize the FSL framework to the setting of the logarithmic barrier (log barrier) over a nonempty, full-dimensional, bounded polytope. The rest of the paper fixes the following notation.

Let $\ma \in \R^{m \times n}$ be a full-rank matrix with rows $\{\va^i\}_{i \in [m]}$, and let $\vb \in \R^m$. Define the polytope
\begin{equation}\label{eq:polytope}\calX \defeq \Brace{\vx \in \R^n \mid \ma \vx < \vb}.\end{equation}
We reserve the notation $\pi$ for the uniform measure over $\calX$.

At a point $\vx \in \calX$, we define its vector of \emph{slacks} $\vs(\vx) \defeq \vb - \ma \vx \in \R^m_{> 0}$, and the corresponding diagonal matrix $\ms(\vx) \defeq \diag{\vs(\vx)}$. The log barrier is defined as
\begin{equation}\label{eq:barrier}\psi(\vx) \defeq -\sum_{i \in [m]} \log\Par{\vs_i(\vx)}.\end{equation}
% Recall that $\pi$ is the uniform measure on $\calX$.

Observe that $\nabla \psi(\vx) = \ma^\top \vs(\vx)^{-1}$, where $\vs(\vx)^{-1}$ is applied entrywise. 
We also define:
\begin{equation}\label{eq:hess_norm}\mh(\vx) \defeq \nabla^2 \psi(\vx) = \ma^\top \ms(\vx)^{-2}\ma,\quad \normx{\vv}{\vx} \defeq \norm{\vv}_{\mh(\vx)},\quad \normxd{\vv}{\vx} \defeq \norm{\vv}_{\mh(\vx)^{-1}}. \end{equation}
Finally, for our application of FSL, it is convenient to define a dual cone:
\begin{equation}\label{eq:dualspace}\calY \defeq \Brace{\vy = \ma^\top \vz: \vz \in \R^m_{> 0}}.\end{equation}
We give a summary of the connection between the primal and dual sets \eqref{eq:polytope}, \eqref{eq:dualspace}.

\begin{lemma}\label{lem:bijection}
The mapping $\nabla \psi: \calX \to \calY$ is a bijection, and for $\vy = \ma^\top \vz \in \calY$ and $\tau > 0$,
\[\vx \mapsto -\inprod{\vy}{\vx} + \tau\psi(\vx)\]
is a strictly convex function, with unique minimizer $\vx$ satisfying $\vy = \tau \ma^\top \vs(\vx)^{-1}$.
\end{lemma}
\begin{proof}
Let $\calY_0 \defeq \{\vy = \ma^\top \vz: \vz \in \R^m_{\ge 0}\}$, so that $\calY$ is the interior of $\calY_0$. We claim that the domain of $\psi^*$, the Fenchel conjugate of $\psi$, is exactly $\calY$. To see this, if $\vy = \ma^\top \vz \in \calY$, then
\begin{align*}\inprod{\vy}{\vx} - \psi(\vx) &= \inprod{\vz}{\ma \vx} + \sum_{i \in [m]} \log(\vs_i(\vx)) \\
&= \inprod{\vz}{\vb} + \sum_{i \in [m]} \Par{\log(\vs_i(\vx)) - \vz_i \vs_i(\vx)} \le \inprod{\vz}{\vb} - \sum_{i \in [m]} \log \vz_i < \infty,
\end{align*}
where the only inequality used $\log s - zs \le -\log z$ for $s, z > 0$. 

Conversely, if $\vy \not\in \calY$, then there is a nonzero separating hyperplane $\vh$ such that $\inprod{\vh}{\vy} \ge 0$, and $\ma \vh \le \0_m$ entrywise. Take any $\vx_0 \in \calX$, and let $\vx_t \defeq \vx_0 + t\vh \in \calX$ for all $t \ge 0$. Then, we claim
\[\inprod{\vy}{\vx_t } - \psi(\vx_t) = \inprod{\vy}{\vx_0} + t\inprod{\vy}{\vh} + \sum_{i \in [m]} \log(\vs_i(\vx_0) - t\inprod{\va^i}{\vh})\]
diverges as $t \to \infty$. To see this, because $\ma$ is full-rank, at least one entry $i \in [m]$ of $\ma\vh$ is not $0$, and therefore $\vs_i(\vx_0) - t\inprod{\va^i}{\vh}$ diverges as $t \to \infty$. In conclusion, the domain of $\psi^*$ is $\calY$.

Now, by the definition of $\calX$ and our assumption that $\ma$ is full-rank, we have $\mh(\vx) \succ \0_{n \times n}$, so $\vx \mapsto -\inprod{\vy}{\vx} + \tau\psi(\vx)$ is strictly convex. Moreover, $\psi$ diverges as $\vx$ approaches the boundary of $\calX$. Thus, $\psi$ is Legendre and the remaining conclusions follow from Theorem 26.5 of \cite{Roc70}.
\end{proof}

Next, to leverage our FSL framework in the context of the log barrier, we explicitly prove that $\psi$ is a LLT, i.e., that there exists $\phi$ such that $\psi=\varphi^\sharp$, as in Definition~\ref{def:llt}.

\begin{lemma}\label{lem:phi_def}
For each $i \in [m]$, define $\mathcal{R}_i\defeq \{\vy\in \R^n\mid \vy=t\va^i, t > 0\}$.  Let $\exp(-\varphi_i)$ be the density on $\mathrm{span}(\va^i)$, supported on $\mathcal{R}_i$, induced by the pushforward of $\exp(-t\vb_i)\ind_{t > 0}\dd t$ by $t\mapsto t\va^i$.  Define 
\[\exp(-\varphi)\defeq \exp(-\varphi_1)*\dots * \exp(-\varphi_m).\]
Then $\psi=\varphi^\sharp$ for $\psi$ defined in \eqref{eq:barrier}.
\end{lemma}
\begin{proof}
It suffices to establish that $\psi_i \defeq -\log(\vs_i(\vx))$ is the LLT of $\phi_i$, and then apply Fact~\ref{fact:llt_conv}. This follows from the observation
\begin{align*}
\int \exp\Par{\inprod{\vx}{\vy} - \phi_i(\vy)} \dd \vy &= \int_0^\infty \exp\Par{-t(\vb_i - \inprod{\va^i}{\vx})} \dd t = \frac{1}{\vb_i - \inprod{\va^i}{\vx}} = \exp\Par{-\log\Par{\vs_i(\vx)}}.\end{align*}
\end{proof}

Algorithm~\ref{alg:prox-sample} requires implementations of both the forward step (Line~\ref{line:prox-sample-forward}) and the backward step (Line~\ref{line:prox-sample-backward}). We conclude this section with an implementation of the forward step. For a fixed $\vx \in \calX$, the forward step samples from the measure on $\vy \in \R^n$ that is
\begin{equation}\label{eq:forward_measure}
 \pi_{\tau}^{Y|X=\vx} =  \calT_{\vx}\exp\Par{-\phi^{\ast \tau}} \propto \exp\Par{\inprod{\vx}{\cdot} - \phi^{\ast \tau}(\cdot)}.\end{equation}
Define for a vector $\vr\in\R^m_{>0}$ the product measure $\Exp(\vr) \defeq \bigotimes_{i\in[m]} \Exp(\vr_i)$ over $\R^m_{>0}$.

\begin{algorithm2e}[h!]
\caption{$\Forwards(\vx, \tau)$}
\label{alg:forwards}
\DontPrintSemicolon
\textbf{Input: }current iterate $\vx\in \calX$, regularization parameter $\tau\in \N$\;
$\vz^{(i)} \simiid \Exp(\vs(\vx))$ for all $i\in[\tau]$\;
$\vgamma \leftarrow \sum_{i\in[\tau]} \vz^{(i)}$\;
$\vy\leftarrow \ma^\top \vgamma $\;
\textbf{Return: }$\vy$
\end{algorithm2e}

\begin{proposition}\label{prop:forward}
Let $\tau \in \N$ and $\vx \in \calX$.
Let $\vz^{(i)} \simiid \Exp(\vs(\vx))$ for all $i\in[\tau]$, and let $\vgamma \defeq \sum_{i\in[\tau]}\vz^{(i)}$. Then, $\vy\defeq \ma^\top \vgamma$ is distributed as~\eqref{eq:forward_measure} and can be computed in time $O(m(n+\tau))$.
\end{proposition}
\begin{proof}
Let $\vy^{(i)} \defeq \ma^\top \vz^{(i)}$ for all $i \in [\tau]$. We first claim that $\vy^{(i)} \sim \calT_{\vx}\exp(-\phi)$, following the notation of Lemma~\ref{lem:phi_def}. We use the following elementary fact: if $T: \Omega \to \Omega'$ is measurable, $\mu$ is a measure on $\Omega$, and $\dd\mu'(\omega) = g(T(\omega))\dd\mu(\omega)$ for a nonnegative function $g$ on $\Omega'$, then
\begin{equation}\label{eq:pushforward_measure}\dd(T_\sharp\mu')(\omega') = g(\omega') \dd (T_\sharp\mu)(\omega').\end{equation}
This follows by expanding definitions: for every measurable $S \subseteq \Omega'$,
\[(T_\sharp\mu')(S) = \int_{T^{-1}(S)} g(T(\omega)) \dd \mu(\omega) = \int_S g(\omega') \dd(T_\sharp\mu)(\omega').\]
Now by Lemma~\ref{lem:phi_def}, $\exp(-\phi) = T_\sharp\mu$ where $T(\vz) = \ma^\top \vz$, and $\mu$ is the measure on $\Omega \defeq \R^m_{>0}$ with density $\mu(\vz) = \exp(-\inprod{\vb}{\vz})$. On the other hand, $\vz^{(i)} \sim \Exp(\vs(\vx))$ has the exact density at $\vz \in \R^m_{>0}$:
\begin{align*}\prod_{j \in [m]} \vs_j(\vx) \exp\Par{-\vs_j(\vx)\vz_j} &= \exp\Par{-\psi(\vx) - \inprod{\vb - \ma\vx}{\vz}} = \exp\Par{\inprod{\vx}{\ma^\top \vz} - \psi(\vx)} \mu(\vz).\end{align*}
Thus, if $\mu'$ denotes the law of $\vz^{(i)}$, then
\[\dd\mu'(\vz) = g(T(\vz)) \dd\mu(\vz),\quad g(\vy) \defeq \exp\Par{\inprod{\vx}{\vy} - \psi(\vx)},\]
Applying \eqref{eq:pushforward_measure} then yields the claim that $\vy = \ma^\top\vz$, for $\vz \sim \mu'$, has law $g(\vy) T_\sharp\mu(\vy) = (\calT_{\vx} \exp(-\phi))(\vy)$.

From \eqref{eq:llt_density}, Fact~\ref{fact:llt_conv}, and \eqref{eq:tilt_convolve}, it follows that the law of $\ma^\top \vgamma = \sum_{i \in [\tau]}\vy^{(i)}$ is $\calT_{\vx}\exp(-\phi^{\ast \tau})$, as claimed. 
We can compute $\vs(\vx)$ in $O(mn)$ time, and each of the $\tau$ draws $\vz^{(i)}$ takes $O(m)$ time. Forming $\vgamma$ then takes $O(m\tau)$ time, and computing $\ma^\top \vgamma$ takes $O(mn)$ time.
\end{proof}

The bulk of the remaining paper concerns (approximately) implementing the backwards step, Line~\ref{line:prox-sample-backward} in Algorithm \ref{alg:prox-sample}.

\subsection{Isoperimetry}\label{ssec:loghelper}

In this section, we recall some facts on the log barrier, and 
prove that $\pi$, the uniform measure on the polytope $\calX$, satisfies a $\psi$-PI. We begin by recalling the definition of self-concordance.

\begin{definition}[Self-concordance]\label{def:sc}
Let $\calX \subseteq \R^d$ be open and convex, and let $\nu > 0$.
We call thrice-differentiable, strictly convex $\psi: \calX \to \R$ a \emph{$\nu$-self-concordant barrier} if $\psi(\vx)\to \infty$ as $\vx$ approaches the boundary of $\calX$, and for all $\vx \in \calX$ and $\vv \in \R^d$, the following inequalities hold:
\[\Abs{\nabla^3 \psi(\vx)[\vv,\vv,\vv]} \le 2\norm{\vv}_{\nabla^2 \psi(\vx)}^3,\quad \norm{\nabla \psi(\vx)}_{(\nabla^2 \psi(\vx))^{-1}}^2 \le \nu.\]
\end{definition}

It is well-known (e.g., Section 2.3, \cite{NesterovN94}) that the log barrier satisfies Definition~\ref{def:sc} with $\nu = m$.

\begin{fact}\label{fact:sc_log}
    The log barrier $\psi$ in~\eqref{eq:barrier} is an $m$-self-concordant barrier.
\end{fact}

Our isoperimetry proof leverages a known comparison between a Riemannian metric induced by the Hessian geometry of $\psi$, and an alternative metric on polytopes called the Hilbert metric. To state this comparison, we require several definitions.

\begin{definition}[Riemannian metric]
For $\vx, \vx' \in \calX$, we define the Riemannian metric
\[\met_\psi(\vx, \vx') \defeq \inf_{\substack{\gamma: [0, 1] \to \calX \\ \gamma(0) = \vx, \gamma(1) = \vx'}} \int_0^1 \norm{\frac{\dd}{\dd s}\gamma(s)\mid_{s=t}}_{\gamma(t)} \dd t.\]
\end{definition}

\begin{definition}[Hilbert metric]
For $\vx, \vx' \in \calX$, we define the Hilbert metric
\[\met_H(\vx, \vx') \defeq \log\Par{1 + \frac{\norm{\vx - \vx'}_2\norm{\vp - \vq}_2}{\norm{\vp - \vx}_2\norm{\vx' - \vq}_2}} \]
where $(\vp, \vx, \vx', \vq)$ are points on the extended line between $(\vx, \vx')$, appearing in that order, and $(\vp, \vq)$ are the intersections of the line with $\calX$'s boundary.
\end{definition}

\begin{fact}[Lemma 3, Item 3, \cite{Narayanan16}]\label{fact:met_compare}
For all $\vx, \vx' \in \calX$, we have $\met_\psi(\vx, \vx') \le \sqrt m \cdot \met_H(\vx, \vx')$.
\end{fact}

To leverage Fact~\ref{fact:met_compare}, we require a result from \cite{Lovasz99} on the isoperimetry of $\pi$ in $\met_H$, as well as Cheeger's inequality in $\met_\psi$, which relates the following two properties of a measure.

\begin{definition}[Cheeger constant]\label{def:cheeger}
For a probability measure $\mu \in \calP(\calX)$, define its Cheeger constants
\[h_\psi(\mu) \defeq \inf_{S \subseteq \calX: 0 < \mu(S) \le \half }\frac{\mu^+_\psi(S)}{\mu(S)},\quad h_H(\mu) \defeq \inf_{S \subseteq \calX: 0 < \mu(S) \le \half }\frac{\mu^+_H(S)}{\mu(S)},\]
where we define the Minkowski contents (boundary measures) of $S \subseteq \calX$ by
\begin{align*}\mu^+_\psi(S) \defeq \lim\inf_{r \to 0} \frac{\mu(S^r_\psi) - \mu(S)}{r},\quad S_\psi^r \defeq \Brace{\vx \in \calX: \met_\psi(\vx, S) < r}, \\
\mu^+_H(S) \defeq \lim\inf_{r \to 0} \frac{\mu(S^r_H) - \mu(S)}{r},\quad S_H^r \defeq \Brace{\vx \in \calX: \met_H(\vx, S) < r}.
\end{align*}
\end{definition}

\begin{definition}[Spectral gap]\label{def:gap}
For a measure $\pi$, we define its spectral gap w.r.t.\ $\met_\psi$ by
\[\lam_\psi(\pi) \defeq \inf_{\Var_\pi[f] > 0} \frac{\calE_\psi(f, f)}{\Var_\pi[f]},\quad \calE_\psi(f, f) \defeq \int \normxd{\nabla f(\vx)}{\vx}^2 \pi(\vx) \dd \vx. \]
\end{definition}

A $\psi$-PI (Definition~\ref{def:pi}) then amounts to lower bounding the spectral gap $\lam_\psi(\pi)$. Cheeger's inequality states that the spectral gap and Cheeger constant obey $\lam_\psi(\pi) \ge \frac 1 4  h_\psi(\pi)^2$, and we can compare $h_\psi$ to $h_H$ via Fact~\ref{fact:met_compare}. This lets us establish a $\psi$-PI for $\pi$.

\begin{lemma}\label{lem:isoperimetry}
$\pi$ satisfies a $\frac 1 {16m}$-$\psi$-PI (Definition~\ref{def:pi}).
\end{lemma}
\begin{proof}
First, by Definition~\ref{def:cheeger} and Fact~\ref{fact:met_compare}, we have $S^{r\sqrt{m}}_\psi \supseteq S^r_H$ for any $S \subseteq \calX$, $r > 0$, which implies $\pi(S^{r\sqrt{m}}_\psi) \ge \pi(S^r_H)$. Thus, Definition~\ref{def:cheeger} shows that $\sqrt{m}h_\psi(\pi) \ge h_H(\pi)$.
The claim then follows from
\[\lam_\psi(\pi) \ge \frac 1 4 h_\psi(\pi)^2 \ge \frac 1 {4m} h_H(\pi)^2 \ge \frac 1 {16m},\]
where the first inequality was Cheeger's inequality $\lam_\psi(\pi) \ge \frac 1 4 h_\psi(\pi)^2$ (see e.g., Theorem 3.1 and Remark 3.6, \cite{BobkovH97}), and the last used $h_H(\pi) \ge \half$ as established by Theorem 1, \cite{Lovasz99}.\footnote{This is not exactly how Theorem 1, \cite{Lovasz99} is stated, but plugging in $S_1 \gets S$ and $S_2 \gets \calX \setminus S_H^r$, taking an infimum over $S$ with $\pi(S) < \half$, and sending $r \to 0$, shows our claim $h_H(\pi) \ge \half$.}
\end{proof}

\begin{remark}\label{rem:slc}
The conclusion of Lemma~\ref{lem:isoperimetry} extends to any log-concave measure $\pi$ supported on $\calX$. To see this, Theorem 2.5 of \cite{LovaszV07} extends Theorem 1, \cite{Lovasz99} as cited in the proof of Lemma~\ref{lem:isoperimetry} to arbitrary log-concave $\pi$, which is the only step that actually used the specific measure. This is relevant for our cooling scheme in Section~\ref{sec:cooling}, which multiplies $\pi$ by a log-concave function.
\end{remark}

\newpage
\part{Polytope Sampling in $\widetilde{O}(m\sqrt n)$ Iterations}\label{part:sqrtn}

\section{Backwards Oracle Implementation via Smoothing}\label{sec:part_1_backwards_step}

The goal of this section is to implement the backwards step of Algorithm~\ref{alg:prox-sample} up to small error in total variation distance, in the setup of Section~\ref{ssec:logbarrier}, from a fixed $\vy \in \calY$. Our analysis requires taking
\[\tau = \Omega\Par{\sqrt{n\log \frac m \delta} + \log \frac m \delta},\]
and assumes that $\vy$ is drawn from the forwards step (Proposition~\ref{prop:forward}) from an initial $\vx_0$ with
\begin{equation}\label{eq:chis_warm}\chi^2\Par{\Law(\vx_0) \| \pi} \le 3.\end{equation}
In Section~\ref{sec:cooling}, we develop a cooling strategy to ensure that \eqref{eq:chis_warm} holds. In Part~\ref{part:tilde1}, we improve our backwards oracle implementation to hold at $\tau = \log\frac{m}{\delta}$, at the expense of a weaker mixing guarantee, later compensated for by an approximation-tolerant variant of Proposition~\ref{prop:prox-sample-guarantee}.

In Section~\ref{ssec:part_1_preliminaries}, we first set up some convenient notation for the rest of the section. In Section~\ref{ssec:concentration}, we then prove that most of the mass of $\pi_\tau^{X \mid Y = \vy}$ lies in a relative slack box around the mode to its potential, complemented by an approximate mode-finding subroutine in Section~\ref{ssec:mode}. In Section~\ref{ssec:smoothing}, we show how to smooth the backwards step potential in the high-probability region, and we put the pieces together in Section~\ref{ssec:main_theorem} to give our backwards step implementation.

\subsection{Setup}\label{ssec:part_1_preliminaries}

Our goal throughout Section~\ref{sec:part_1_backwards_step} is to sample from the backwards conditional of \eqref{eq:joint-density}:
\begin{equation}\label{eq:piy_def}
\pi_{\tau}^{X|Y=\vy} \defeq \calT_{\vy}\exp(-\tau\psi)\pi.
\end{equation}
We define the mode of this density's potential (which exists and is unique, see Lemma~\ref{lem:bijection}):
\begin{equation}
\label{eq:mode_def}\vxs_{\vy, \tau} \defeq \arg\min_{\vx \in \calX}\Brace{\tau\psi(\vx) - \inprod{\vy}{\vx}}.
\end{equation}
% Following \eqref{eq:foo_backwards}, observe that
% $\vy = \tau\ma^\top \vs(\vxs)^{-1}$, so that we have for the relative slack,
We take $\vy$ and $\tau$ to be fixed throughout Section~\ref{sec:part_1_backwards_step}; hence, we drop their dependence in notation for simplicity and write $\pi^{X|Y}:= \pi_{\tau}^{X|Y=\vy}$ and $\vxs:=\vxs_{\vy,\tau}$. Furthermore, we define the following objects relative to a fixed $\vw\in\calX$; when $\vw$ is omitted, $\vw = \vx^\star$ always by default. Our applications take $\vw$ to be an approximate mode of $\tau\psi - \inprod{\vy}{\cdot}$, computed in Section~\ref{ssec:mode}.
\begin{enumerate}
    \item {\emph{Relative slacks.} For all $\vx \in \calX$, applying division coordinatewise, we define:
    \begin{equation}\label{eq:relative_slacks}
         \vr(\vx; \vw) \defeq \frac{\vs(\vw) - \vs(\vx)}{\vs(\vw)} = \frac{\ma(\vx - \vw)}{\vs(\vw)}.
    \end{equation}
    }
    \item {\emph{Whitened coordinates.} For all $\vx \in \calX$, we define:
    \begin{equation}\label{eq:whitened_coords}
        \vz(\vx; \vw) \defeq \mh(\vw)^{\half}(\vx - \vw).
    \end{equation}
    }
    \item {\emph{Isotropic frame.} For all $i \in [m]$, we define:
    \begin{equation}\label{eq:vdef}
        \vv^i(\vw) \defeq \mh(\vw)^{-\half} \frac{\va^i}{\vs_i(\vw)},
    \end{equation}
    and let $\mv(\vw)\in\R^{m\times n}$ be the matrix with rows $\vv^i(\vw)^\top$.  We repeatedly use the facts
    \begin{equation}\label{eq:frame_facts}
    \mv(\vw)^\top \mv(\vw) = \id_n,\quad \vv^i(\vw)^\top \vz(\vx; \vw) = \vr_i(\vx; \vw)\text{ for all } i \in [m].
    \end{equation}
    }
    \item {\emph{Relative slack box.} For all $\rho \ge 0$, we define:
    \begin{equation}\label{eq:relative_slack_box}
        \calB(\vw, \rho):=\Brace{\vx\in\calX : \norm{\vr(\vx; \vw)}_\infty \leq \rho}.
    \end{equation}
    }
    \item {\emph{Relative gradient.} We define:
    \begin{equation}\label{eq:relative_gradient}\vh(\vw) \defeq \mh(\vw)^{-\half}\Par{\nabla \psi(\vw) - \frac{\vy}{\tau}}.\end{equation}
    Recall from Lemma~\ref{lem:bijection} that $\vh(\vxs) = \0_n$.
    }
\end{enumerate}

With this notation in hand, we next reparameterize the potential to $\pi^{X \mid Y}$ in a convenient way.

\begin{lemma}\label{lem:pushforward_reparam}
Fix $\vw \in \calX$, and let $\nu_\vw$ be the pushforward of $\pi^{X \mid Y}$ under the map $\vx \mapsto \vz(\vx; \vw)$. Then the density of $\nu_\vw$ relative to Lebesgue measure is
\[\nu_\vw(\dd \vz) \propto \exp\Par{-\tau\Par{\sum_{i \in [m]} f\Par{\inprod{\vv^i(\vw)}{\vz}} + \inprod{\vh(\vw)}{\vz}}} \dd \vz,\]
for all $\vz$ in the image of $\calX$ under the stated map, where for all $u < 1$, $f(u) \defeq -\log(1 - u) - u$.
\end{lemma}
\begin{proof}
First, our definition of $\psi$ in \eqref{eq:barrier} implies
\[\psi(\vx) - \psi(\vw) = -\sum_{i \in [m]} \log \frac{\vs_i(\vx)}{\vs_i(\vw)} = -\sum_{i \in [m]} \log\Par{1 - \vr_i(\vx; \vw)} = \sum_{i \in [m]} f\Par{\vr_i(\vx;\vw)} + \vr_i(\vx;\vw).\]
On the other hand,
\[\sum_{i \in [m]} \vr_i(\vx;\vw) = \inprod{\ma^\top\vs(\vw)^{-1}}{\vx - \vw} = \inprod{\nabla \psi(\vw)}{\vx - \vw}.\]
Combining the above displays with \eqref{eq:frame_facts} then shows
\begin{align*}
\inprod{\vy}{\vx - \vw} - \tau\Par{\psi(\vx) - \psi(\vw)} &= -\tau\Par{\inprod{\nabla \psi(\vw) - \frac{\vy}{\tau}}{\vx - \vw} + \sum_{i \in [m]} f\Par{\vr_i(\vx; \vw)}} \\
&= -\tau\Par{\sum_{i \in [m]} f\Par{\inprod{\vv^i(\vw)}{\vz(\vx; \vw)}} + \inprod{\vh(\vw)}{\vz(\vx; \vw)}},
\end{align*}
and performing an affine change of variables (with constant Jacobian $\mh(\vw)^{-\half}$) gives the result.
\end{proof}

\subsection{Concentration}\label{ssec:concentration}

In this section, we prove that the backwards conditional 
$\pi^{X \mid Y}$ concentrates in a constant-width relative slack box around its mode. The following result is our main export.

\begin{lemma}\label{lem:main_infty_conc}
There is a universal constant $C > 0$ such that, for any $\delta \in (0, \half)$, if
\[\tau \ge C\Par{\sqrt{n} + \log \frac m \delta},\]
then 
\[\pi^{X \mid Y}\Par{\calB\Par{\vxs, \frac 1 8}} \ge 1 - \delta.\]
\end{lemma}

Throughout this section, we primarily work in the whitened coordinates \eqref{eq:whitened_coords}, and denote $\nu \defeq \nu_{\vxs}$ as shorthand for the pushforward distribution of $\pi^{X \mid Y}$, consistently with Lemma~\ref{lem:pushforward_reparam}. Following the notation of Section~\ref{ssec:part_1_preliminaries}, we also abbreviate
\[\mv \defeq \mv(\vxs),\quad \vv^i \defeq \vv^i(\vxs),\quad \vr(\vz) \defeq \mv \vz,\]
where the last definition is consistent with \eqref{eq:frame_facts}. Since $\vh(\vxs) = \0_n$, Lemma~\ref{lem:pushforward_reparam} gives
\begin{equation}\label{eq:nu_potential}\nu(\dd \vz) \propto \exp(-F(\vz)),\quad F(\vz) \defeq \tau \sum_{i \in [m]} f(\vr_i(\vz)),\quad f(u) = -\log(1 - u) - u.\end{equation}
Motivated by this decomposition, we define a helpful nonlinear transformation of $\vz$:
\begin{equation}\label{eq:transform_slacks}\vt_i(\vz) \defeq -\log(1 - \vr_i(\vz)) = -\log\Par{1 - \inprod{\vv^i}{\vz}},\text{ for all } i \in [m].\end{equation}
On the support of $\nu$, we have
\[\vr_i(\vz) = 1 - \exp(-\vt_i(\vz)),\quad f'(\vr_i(\vz)) = \exp(\vt_i(\vz)) - 1,\quad f''(\vr_i(\vz)) = \exp(2\vt_i(\vz)), \]
and hence,
\begin{equation}\label{eq:bigF_deriv}\nabla F(\vz) = \tau\mv^\top\Par{\exp(\vt(\vz)) - \1_m},\quad \nabla^2 F(\vz) = \tau\mv^\top \diag{\exp(2\vt(\vz))} \mv.\end{equation}
Our proof of Lemma~\ref{lem:main_infty_conc} proceeds in two parts. Lemma~\ref{lem:infty_dev} first shows that the relative slack vector $\vr(\vz)$ concentrates around its mean in the $\ell_\infty$ norm, via a general argument based on bounding the moment generating function (MGF) of each $\vt_i(\vz)$. We then carefully apply the \emph{score} and \emph{virial} identities (Fact~\ref{fact:score_virial}) to directly bound the $\ell_\infty$ norm of $\E_{\vz \sim \nu}[\vr(\vz)]$ in Lemma~\ref{lem:infty_expect}.

\textbf{Concentration of $\vr(\vz)$ in $\ell_\infty$.} We begin with a general strategy for establishing concentration of a nonlinear transformation of a variable, via controlling its MGF.

\begin{lemma}\label{lem:main_concentration}
Let $\Omega\subseteq \R^n$ be open and convex, let $F: \Omega \to \R$, $G: \Omega \to \R$ be twice continuously differentiable, and let $L > 0$. For all $\lambda\in [-L, L]$, define the tilted probability density
\[
\mu_\lambda(\dd \vx) \propto \exp(-F(\vx) + \lam G(\vx)) \dd \vx,
\]
and let $X_\lam \sim \mu_\lam$. If there exists $\sigma>0$ such that $\Var[G(X_\lambda)] \leq \sigma^2$ for all $\lambda \in [-L, L]$, then

\[
\Pr \Brack{\Abs{G(X_0)-\E \Brack{G(X_0)}}\ge t}\le 2\exp\Par{-\frac{t^2}{2\sigma^2}} \text{ for all } t\in[0, \sigma^2 L].
\]

\end{lemma}
\begin{proof}
Define the log-MGF by $\Lambda(\lambda)\defeq \log\E\Brack{\exp(\lambda(G(X_0)-\E [G(X_0)]))}$. Direct calculation gives
    \[\Lambda'(\lambda)=\E[G(X_\lambda)]-\E[G(X_0)], \quad \Lambda''(\lambda)=\operatorname{Var}[G(X_\lambda)] \le \sigma^2.\]
    Since $\Lambda(0) = \Lambda'(0) = 0$, integrating the last inequality twice gives
    \begin{equation}\label{eq:mgf_bound}\Lambda(\lam) \le \frac{\sigma^2 \lam^2}{2},\text{ for all } \lambda \in [-L, L].\end{equation}
    For $t \le \sigma^2 L$, taking $\lambda = \frac{t}{\sigma^2}$ above and using Markov's inequality implies 
    \begin{align*}
    \Pr\Brack{G(X_0) - \E[G(X_0)] \ge t} &= \Pr\Brack{\exp(\lambda(G(X_0) - \E[G(X_0)])) \ge \exp(\lambda t)} \\
    &\le \exp(-\lambda t) \E\Brack{\exp\Par{\lambda(G(X_0) - \E[G(X_0)])}}\\
    &\le \exp\Par{-\lambda t + \frac{\sigma^2 \lam^2}{2}} = \exp\Par{-\frac{t^2}{2\sigma^2}}.
    \end{align*}
    Applying the same argument to $-G$ and taking a union bound gives the claim.
\end{proof}

We now derive our concentration bound on the relative slacks by applying Lemma~\ref{lem:main_concentration}.

\begin{lemma}\label{lem:infty_dev}
Let $Z \sim \nu$. There is a universal constant $C > 0$ such that, for any $\delta \in (0, 1)$, if
\[\tau \ge C\log \frac{m}{\delta},\]
then with probability $\ge 1 - \delta$,
\[\norm{\vr(Z) - \E[\vr(Z)]}_\infty \le 6\Par{1 + \norm{\E \Brack{\vr(Z)}}_\infty} \sqrt{\frac{\log \frac {2m} \delta}{\tau}}.\]
\end{lemma}
\begin{proof}
We first establish concentration for the transformed slacks $\vt(Z)$ in \eqref{eq:transform_slacks}. Fix $i \in [m]$, and for $|\lam| \le \frac \tau 2$ define the tilted potential $F_\lam(\vz) \defeq F(\vz) - \lam \vt_i(\vz)$. Using \eqref{eq:bigF_deriv}, we compute
\[\nabla \vt_i(\vz) = \exp(\vt_i(\vz)) \vv^i,\quad \nabla^2 F_\lam(\vz) = (\tau-\lam) \exp(2\vt_i(\vz))\vv^i(\vv^i)^\top + \tau \sum_{j \in [m]: j \neq i} \exp(2\vt_j(\vz))\vv^j(\vv^j)^\top.\]
Therefore, for every $\vu \in \R^n$ and $\vz$ in the support of $\nu$,
\[\inprod{\nabla \vt_i(\vz)}{\vu}^2 = \exp\Par{2\vt_i(\vz)}\inprod{\vv^i}{\vu}^2 \le \frac 1 {\tau - \lam} \inprod{\vu}{\nabla^2 F_\lam(\vz)\vu}.\]
Taking $\vu = (\nabla^2 F_\lam(\vz))^{-1}\nabla \vt_i(\vz)$ then implies
\[\inprod{\nabla \vt_i(\vz)}{\Par{\nabla^2 F_\lam(\vz)}^{-1} \nabla \vt_i(\vz)} \le \frac 1 {\tau - \lam} \le \frac 2 \tau.\]
Now the Brascamp-Lieb inequality (Theorem 4.1, \cite{BrascampL76}) shows that for $Z_\lam$ distributed $\propto \exp(-F_\lam)$,
\begin{equation}\label{eq:t_brascamp_lieb}
    \Var[\vt_i(Z_{\lambda})] \leq \E\Brack{\inprod{\nabla \vt_i(Z_\lam)}{\Par{\nabla^2 F_\lam(Z_\lam)}^{-1} \nabla \vt_i(Z_\lam)}} \leq \frac{2}{\tau}.
\end{equation}
Then Lemma~\ref{lem:main_concentration} with $\sigma^2 \gets \frac 2 \tau$, $L \gets \frac \tau 2$, and $t \gets 2\sqrt{\log(\frac{2m}{\delta})/\tau} \le \sigma^2 L = 1$ implies that 
\begin{equation}\label{eq:good_event_t}\norm{\vt(Z) - \E \Brack{\vt(Z)}}_\infty \le 2\sqrt{\frac{\log\frac {2m} \delta}{\tau}},\end{equation}
with probability $\ge 1 - \delta$,
by taking a union bound over all $i \in [m]$. Condition on \eqref{eq:good_event_t} henceforth. It remains to transfer this concentration bound from $\vt$ to the relative slacks $\vr$. To do so, for all $i \in [m]$, define $\Lambda_i \defeq \log \E [\exp(\E [\vt_i(Z)] - \vt_i(Z))] = \Lambda(-1)$ in the notation of Lemma~\ref{lem:main_concentration}, so that
\[1 - \E\Brack{\vr_i(Z)} = \E\Brack{\exp\Par{-\vt_i(Z)}} = \exp\Par{-\E\Brack{\vt_i(Z)} + \Lambda_i},\]
and hence
\begin{equation}\label{eq:main_transfer_bound}
\begin{aligned}
\Abs{\vr_i(Z) - \E\Brack{\vr_i(Z)}} &= \Abs{\exp\Par{-\E[\vt_i(Z)] + \Lambda_i} - \exp\Par{-\vt_i(Z)}} \\
&= \exp\Par{-\E\Brack{\vt_i(Z)} + \Lambda_i}\Abs{1 - \exp\Par{-\vt_i(Z) + \E\Brack{\vt_i(Z)} - \Lambda_i}} \\
&= \Par{1 - \E\Brack{\vr_i(Z)} }\Abs{1 - \exp\Par{-\vt_i(Z) + \E\Brack{\vt_i(Z)} - \Lambda_i}}.
\end{aligned}
\end{equation}
By definition, $1 - \E[\vr_i(Z)] \le 1 + \norm{\E[\vr(Z)]}_\infty$, and we also have
\[\Abs{-\vt_i(Z) + \E\Brack{\vt_i(Z)} - \Lambda_i} \le 2\sqrt{\frac{\log\frac {2m} \delta}{\tau}} + \frac 1 \tau \le 3\sqrt{\frac{\log\frac {2m} \delta}{\tau}} \le 1.\]
The last two inequalities were simplified using our choice of $\tau$, and the first applied \eqref{eq:good_event_t} and 
\begin{equation}\label{eq:lam_infty_bound}
0 \le \Lambda_i \le \frac{\sigma^2}{2} = \frac 1 \tau,
\end{equation}
where the upper bound follows by taking $\lambda = -1$ in \eqref{eq:mgf_bound}, and the lower bound is Jensen's inequality. Plugging this into \eqref{eq:main_transfer_bound} gives the claim, upon 
applying $|1 - \exp(-a)| \le 2|a|$ for $|a| \le 1$.
\end{proof}

\textbf{Bound on $\E[\vr(\vz)]$ in $\ell_\infty$.} To complete our proof of Lemma~\ref{lem:main_infty_conc}, it remains to combine Lemma~\ref{lem:infty_dev} with a bound on $\norm{\E[\vr(Z)]}_\infty$ for $Z \sim \nu$. We require the following helper facts.

\begin{fact}\label{fact:score_virial}
Following the notation \eqref{eq:nu_potential},\footnote{These are standard identities that follow from appropriate integration by parts. In both cases, the boundary terms vanish in our application because the density tends to zero at the boundary of the support.}
\[\E_{Z \sim \nu}\Brack{\nabla F(Z)} = \0_n,\quad \E_{Z \sim \nu}\Brack{\inprod{Z}{\nabla F(Z)}} = n.\]
\end{fact}

\begin{lemma}\label{lem:auxiliary_ineq}
    For $x \in \R$ and $\eta\in [0, 1]$, we have
    \begin{equation*}
        (\exp(x) - 1)(1 - \exp(-x + \eta)) \geq \begin{cases}
            x^2, \ &x < 0\\
            x^2 - 2x\eta, \ &x \geq 0. 
        \end{cases}
    \end{equation*}
\end{lemma}
\begin{proof}
Denote the left and right hand sides above by $g(x)$ and $h(x)$ respectively.
    Observe that $g''(x)=\exp(x)+\exp(-x+\eta)\geq 2=h''(x)$ everywhere, because $\eta \ge 0$. 
    
    When $x<0$, $g'(x) - h'(x) = \exp(x)-\exp(-x+\eta)-2x \le 0$. Since $g(0) = h(0) = 0$, taking a double integral shows $g(x) \ge h(x)$ when $x < 0$.  
    Similarly, $\exp(x)-\exp(-x+\eta)-2x+2\eta \ge 0$ for all $x\geq 0$, and taking another double integral completes the proof.
\end{proof}

We proceed to our bound on the $\ell_\infty$ norm of the expected relative slack vector. We in fact prove a stronger result,
namely a bound on its $\ell_2$ norm. To this end, recalling $\vt_i$ defined in~\eqref{eq:transform_slacks}, we introduce the quantity
\begin{equation}\label{eq:delta_definition}
    \vDelta_i := \log \left(\E_{Z\sim \nu}[\exp(\vt_i(Z))]\E_{Z\sim \nu}[\exp(-\vt_i(Z))]\right),\text{ for all } i \in [m].
\end{equation}
The intuition for introducing $\vDelta_i$, which is exactly $\Lambda(1) + \Lambda(-1)$ in the notation of Lemma~\ref{lem:main_concentration} when $G \gets \vt_i$, is that the mean $\E[\vt_i(Z)]$ exactly cancels, allowing us to transfer bounds on the fluctuations of $\vt_i(Z)$ (e.g., from Lemma~\ref{lem:main_concentration}) to bounds on $\vDelta$. We now provide such a helper bound.

\begin{lemma}\label{lem:delta_bounds}
    For $\vDelta$ defined as in \eqref{eq:delta_definition} and $\tau \ge 2$, we have $\norm{\vDelta}_2 \leq \frac{\sqrt {2n}}{\tau}$.
\end{lemma}
\begin{proof}
    We first upper bound $\norm{\vDelta}_1$. A direct calculation using \eqref{eq:bigF_deriv} shows that
    \begin{align*}
        \inprod{\vz}{\nabla F(\vz)}
        &=\tau \sum_{i\in[m]} f'(\vr_i(\vz))\inprod{\vv^i}{\vz}
        =\tau \sum_{i\in[m]} \left(\frac{1}{1-\vr_i(\vz)} - 1\right) \vr_i(\vz) \\
        &=\tau \sum_{i\in[m]} \exp(\vt_i(\vz)) + \exp(-\vt_i(\vz)) - 2.
    \end{align*}
    Taking expectations on both sides over $Z \sim \nu$ and applying Fact~\ref{fact:score_virial} on the left hand side gives
    \begin{align*}
        \frac{n}{\tau} &= \sum_{i\in[m]} \E[\exp(\vt_i(Z))] + \E[\exp(-\vt_i(Z))] - 2 \\
        &\geq \sum_{i\in [m]} \log \E[\exp(\vt_i(Z))] + \log \E [\exp (-\vt_i(Z))]
        = \sum_{i\in[m]} \vDelta_i = \norm{\vDelta}_1,
    \end{align*}
    where the second line uses the inequality $x - 1 \geq \log x$, and that $\vDelta_i \ge 0$ by Jensen's inequality. 
    
    Next, combining \eqref{eq:mgf_bound} and \eqref{eq:t_brascamp_lieb}, i.e., taking $G \gets \vt_i$, $\sigma^2\leftarrow \frac2\tau$, and $\lambda\in\{\pm 1\}$ in Lemma~\ref{lem:main_concentration}, yields
    \begin{equation}\label{eq:Delta_inf_bound}
        \vDelta_i = \log \E[\exp(\vt_i(Z) - \E[\vt_i(Z)])] + \log \E[\exp(\E[\vt_i(Z)] - \vt_i(Z))] \leq \frac{1}{\tau} + \frac{1}{\tau} = \frac{2}{\tau},
    \end{equation}
    which implies $\norm{\vDelta}_\infty \leq \frac2\tau$. Finally, the claim follows from
    \begin{equation*}
        \norm{\vDelta}_2^2 \leq \norm{\vDelta}_1 \norm{\vDelta}_\infty \leq \frac{2n}{\tau^2}\implies \norm{\vDelta}_2 \leq \frac{\sqrt{2n}}{\tau}.
    \end{equation*}
\end{proof}

Next, we use Fact~\ref{fact:score_virial} to convert our $\norm{\vDelta}_2$ bound from Lemma~\ref{lem:delta_bounds} into a bound on $\norm{\E[\vr(Z)]}_2$.

\begin{lemma}\label{lem:infty_expect}
Let $Z \sim \nu$ for $\nu$ defined in \eqref{eq:nu_potential}. If $\tau \ge 3 \sqrt{2n}$, then
\[\norm{\E\Brack{\vr(Z)}}_\infty \le \norm{\E\Brack{\vr(Z)}}_2 \le \frac{6\sqrt{2n}}{\tau}.\]
\end{lemma}
\begin{proof}
For convenience, let $\vp \defeq \log \E[\exp(\vt(Z))]$, where all vector operations are coordinatewise. By \eqref{eq:bigF_deriv} and the score identity in Fact~\ref{fact:score_virial},
\[\mv^\top\Par{\exp(\vp) - \1_m} = \0_n.\]
On the other hand, by the definition of $\vDelta$ in \eqref{eq:delta_definition},
\[\E\Brack{\exp(-\vt(Z))} = \exp\Par{\vDelta - \vp},\]
and hence using $\vr(\vz) = \mv\vz$ and the relationship between $\vt$ and $\vr$,
\[\1_m - \E[\exp(-\vt(Z))] = \1_m - \exp(\vDelta - \vp) = \E[\vr(Z)] = \mv\E[Z].\]
Thus, $\1_m - \exp(\vDelta - \vp)$ is in the span of $\mv$ and $\exp(\vp) - \1_m$ is in the kernel of $\mv^\top$. Combining,
\begin{align*}0 &= \inprod{\exp(\vp) - \1_m}{\1_m - \exp(\vDelta - \vp)} = \sum_{i \in [m]} \Par{\exp(\vp_i) - 1}\Par{1 - \exp(-\vp_i + \vDelta_i)} \\
&\ge \sum_{i \in [m]} \vp_i^2 - \sum_{i \in [m]: \vp_i \ge 0} 2\vp_i \vDelta_i \ge \norm{\vp}_2^2 - 2\norm{\vp}_2\norm{\vDelta}_2.
\end{align*}
The second line used Lemma~\ref{lem:auxiliary_ineq} and $\norm{\vDelta}_\infty \le 1$ from \eqref{eq:Delta_inf_bound}. Rearranging shows $\norm{\vp}_2 \le 2\norm{\vDelta}_2$, so
\[\norm{\vDelta - \vp}_2 \le \norm{\vDelta}_2 + \norm{\vp}_2 \le 3\norm{\vDelta}_2 \le \frac{3\sqrt{2n}}{\tau} \le 1.\]
Finally, recall that
\[\E[\vr(Z)] = \1_m - \E[\exp(-\vt(Z))] = \1_m - \exp\Par{\vDelta - \vp}.\]
Combining the above displays with $|1 - \exp(a)| \le 2|a|$ for $|a| \le 1$ yields the claim.
\end{proof}

We are now ready to combine Lemmas~\ref{lem:infty_dev} and~\ref{lem:infty_expect} to prove Lemma~\ref{lem:main_infty_conc}.

\begin{proof}[Proof of Lemma~\ref{lem:main_infty_conc}]
This is immediate from Lemmas~\ref{lem:infty_dev} and~\ref{lem:infty_expect}, after changing coordinates from $\vz$ to $\vx$, because under the definition of $\nu$ as the pushforward of $\pi^{X \mid Y}$ in Lemma~\ref{lem:pushforward_reparam},
\[\norm{\vr(\vz(\vx))}_\infty \le \frac 1 8 \iff \vx \in \calB\Par{\vxs, \frac 1 8}.\]
\end{proof}

\textbf{Extensions to $\ell_2$.} For our later use, we conclude with an $\ell_2$ analog of the high-probability deviation bound in Lemma~\ref{lem:infty_dev}. Importantly, it scales with the ambient dimension $n$, rather than $m$.

\begin{lemma}\label{lem:l2_bound_r}
Let $Z \sim \nu$ for $\nu$ defined in \eqref{eq:nu_potential}. There is a universal constant $C > 0$ such that, for any $\delta \in (0, \half)$, if
\[\tau \ge C\Par{\sqrt n + \log \frac m \delta},\]
then with probability $\ge 1 - \delta$,
\[\norm{\vr(Z) - \E[\vr(Z)]}_2 \le C\sqrt{\frac{n + \log \frac m \delta}{\tau}}.\]
\end{lemma}
\begin{proof}
We first apply Lemma~\ref{lem:main_infty_conc} with failure probability $\frac \delta 2$, so that adjusting constants appropriately, with probability $\ge 1 - \frac \delta 2$, we have $\norm{\vr(Z)}_\infty \le \frac 1 8$. We also claim that with probability $\ge 1 - \frac \delta 2$,
\begin{equation}\label{eq:f_bound}F(Z) \le 2\Par{n\log 2 + \log \frac 2 \delta},\end{equation}
following the notation \eqref{eq:nu_potential}.
To see this, let $\Omega$ denote the support of $\nu$, i.e., the pushforward of $\calX$ under the whitening operation \eqref{eq:whitened_coords}. Since $F(\0_n) = 0$ and $\0_n \in \Omega$, applying convexity of $F$ gives $F(\half\vz) \le \half F(\vz)$ for all $\vz \in \Omega$, and hence
\[\E\Brack{\exp\Par{\frac{F(Z)}{2}}} = \frac{\int_\Omega \exp(-\frac{F(\vz)}{2}) \dd \vz}{\int_\Omega \exp(-F(\vz))\dd \vz} \le 2^n \frac{\int_{\half \Omega} \exp(-F(\vz)) \dd \vz}{\int_\Omega \exp(-F(\vz)) \dd \vz} \le 2^n.\]
Applying Markov's inequality and taking logarithms then shows \eqref{eq:f_bound}. Condition on these two events in the rest of the proof. Since $f(0) = f'(0) = 0$, and
\[f''(u) = \frac 1 {(1 - u)^2} \ge \frac 2 3 \text{ for } |u| \le \frac 1 8,\]
we have $f(u) \ge \frac{u^2}{3}$ on the interval $u \in [-\frac 1 8, \frac 1 8]$. Therefore, by the definition of $F$ in \eqref{eq:nu_potential},
\[F(Z) = \tau\sum_{i \in [m]} f(\vr_i(Z)) \ge \frac \tau 3 \norm{\vr(Z)}_2^2 \implies \norm{\vr(Z)}_2 \le \sqrt{\frac{6(n\log 2 + \log \frac 2 \delta)}{\tau}}.\]
Finally, combining with Lemma~\ref{lem:infty_expect} and adjusting constants appropriately gives the claim.
\end{proof}

\subsection{Mode finding}\label{ssec:mode}

In this section, we give a self-contained presentation on computing the mode of the backwards oracle potential \eqref{eq:mode_def} from $\vy \in \calY$ to high accuracy, i.e., approximating
\[\vxs_{\vy, \tau} \defeq \arg\min_{\vx \in \calX}\Brace{\tau\psi(\vx) - \inprod{\vy}{\vx}}.\]

Our presentation has two parts. We first repurpose our existing concentration bounds to show that, with high probability over $X \sim \mu$, a $\chi^2$ warm start to $\pi$, the relative slack vector $\vr(X; \vx^\star_{Y, \tau})$ defined in \eqref{eq:relative_slacks} already has bounded $\ell_4$ norm, after drawing $Y$ from the forwards step from $X$. We then analyze the convergence of Newton's method under an initial $\ell_4$ bound on relative slacks.

\textbf{High-probability relative slack bound.} We begin by showing concentration of relative slacks.

\begin{lemma}\label{lem:l4_slacks}
Let $\chi^2(\mu \| \pi) \le 3$ and let $\delta \in (0, \frac 1 3)$. If 
\[\tau \ge C\Par{\sqrt{n\log \frac m \delta} + \log \frac m \delta},\]
where $C > 0$ is a universal constant,
then following the notation \eqref{eq:forward_measure}, \eqref{eq:mode_def},
\[\Pr_{\substack{X \sim \mu \\ Y \sim \calT_{X}\exp(-\varphi^{*\tau})}}\Brack{\norm{\vr\Par{X; \vxs_{Y, \tau}}}_4 > \frac 1 4} \le \delta.\]
\end{lemma}
\begin{proof}
Denote by $\calE$ the event that $\norms{\vr(X; \vxs_{Y, \tau})}_4 > \frac 1 4$. Also, let $\tpi$ be the joint density on $(X, Y)$ where $X \sim \pi$ and $Y \sim \calT_{X}\exp(-\varphi^{*\tau})$, and similarly denote the same joint density with $X \sim \mu$ by $\tmu$. Under this notation, our goal is to show $\Pr_{\tmu}[\calE] \le \delta$.
We claim it suffices to show that 
\begin{equation}\label{eq:true_pi_l4}\Pr_{\tpi}[\calE] \le \frac{\delta^2}{4}, \end{equation}
i.e., that the failure probability is $\le \frac{\delta^2}{4}$ when $X \sim \pi$ instead of $X \sim \mu$. To see this, a direct calculation shows $\chi^2(\tmu \| \tpi) = \chi^2(\mu \| \pi) \le  3$, and thus by the Cauchy-Schwarz inequality,
\begin{align*}
\Pr_{\tmu}[\calE] = \E_{\tmu}[\ind_{\calE}] \le \Par{\E_{\tpi}[\ind_{\calE}]}^{\half}\Par{\E_{\tpi}\Brack{\Par{\frac{\tmu}{\tpi}}^2}}^{\half} \le \sqrt{\frac{\delta^2}{4}\Par{\chi^2\Par{\tmu \| \tpi} + 1}} \le \delta.
\end{align*}
We now prove \eqref{eq:true_pi_l4}. Note that $(X, Y) \sim \tpi$ can equivalently be generated by first drawing $Y$ from its marginal under $\tpi$, and then drawing $X \sim \pi_\tau^{X \mid Y} = \calT_Y\exp(-\tau\psi)$. We instead prove that for \emph{any} $\vy \in \calY$, the bound $\Pr_{X \sim \calT_{\vy}\exp(-\tau\psi)}[\calE] \le \frac{\delta^2}{4}$ holds, at which point averaging over $\vy$ yields \eqref{eq:true_pi_l4}.

Now fix $\vy \in \calY$. We claim that the following hold with probability $\ge 1 - \frac{\delta^2} 4$ over $X \sim \calT_{\vy}\exp(-\tau\psi)$:
\begin{equation}\label{eq:three_bounds}
\norm{\E[\vr(X)]}_2 \le \frac 1 8,\quad \norm{\vr(X) - \E[\vr(X)]}_\infty \le C'\sqrt{\frac{\log \frac m \delta}{\tau}},\quad \norm{\vr(X) - \E[\vr(X)]}_2 \le C'\sqrt{\frac{n + \log \frac m \delta}{\tau}},
\end{equation}
where $C' > 0$ is a universal constant.
Under \eqref{eq:three_bounds}, we indeed have for a large enough $C$ in the definition of $\tau$ that $\calE$ does not occur:
\begin{align*}\norm{\vr(X)}_4 &\le \norm{\E[\vr(X)]}_4 + \norm{\vr(X) - \E[\vr(X)]}_4 \\
&\le \norm{\E[\vr(X)]}_2 + \sqrt{\norm{\vr(X) - \E[\vr(X)]}_2 \norm{\vr(X) - \E[\vr(X)]}_\infty} \le \frac 1 4,\end{align*}
where we used H\"older's inequality to conclude that $\norm{\vv}_4^2 \le \norm{\vv}_2\norm{\vv}_\infty$ for all vectors $\vv$.

We now establish \eqref{eq:three_bounds}. The first bound is Lemma~\ref{lem:infty_expect} with a large enough $C$, the second is Lemma~\ref{lem:infty_dev} with $\delta \gets \frac {\delta^2} 8$, and the third is Lemma~\ref{lem:l2_bound_r} with $\delta \gets \frac {\delta^2} 8$, adjusting constants appropriately.
\end{proof}

\textbf{Newton's method under $\ell_4$-bounded relative slacks.} We next show that one step of Newton's method on $\tau\psi - \inprod{\vy}{\cdot}$, starting from $\vx$ with small $\|\vr(\vx)\|_4$, enters the Dikin ellipse of the mode.

\begin{lemma}\label{lem:l4_newton}
Fix $\vy \in \calY$, and let $\vx \in \calX$ satisfy $\norms{\vr(\vx; \vx^\star_{\vy, \tau})}_4 \le \frac 1 4$. Then if
\begin{equation}\label{eq:newton}\vx' \gets \vx - \mh(\vx)^{-1}\Par{\nabla \psi(\vx) - \frac 1 \tau \vy},\end{equation}
we have $\norms{\mh(\vx^\star_{\vy, \tau})^{\half} (\vx' - \vxs_{\vy, \tau})}_2 \le \frac 1 5$.
\end{lemma}
\begin{proof}
Throughout this proof we denote $\vx^\star \defeq \vx^\star_{\vy, \tau}$, $\vr \defeq \vr(\vx; \vxs)$, and $\vw \defeq \mh(\vxs)^{\half}(\vx - \vxs)$. We also let $\mv \in \R^{m \times n}$ be defined as in \eqref{eq:vdef}, so that $\vr = \mv \vw$. We first have
\begin{align*}\mh(\vx) &= \sum_{i \in [m]} \frac{\va^i(\va^i)^\top}{\vs_i(\vx)^2} =\sum_{i \in [m]} \frac{\va^i(\va^i)^\top}{\vs_i(\vxs)^2(1 - \vr_i)^2} \\
&= \sum_{i \in [m]} \frac 1 {(1 - \vr_i)^2} \mh(\vxs)^{\half} \vv^i (\vv^i)^\top \mh(\vxs)^{\half}.  
\end{align*}
Thus, letting $\mm \defeq \mv^\top \diag{\1_m - \vr}^{-2}\mv$, we have shown
\begin{equation}\label{eq:relative_h}\mh(\vx) = \mh(\vxs)^{\half} \mm \mh(\vxs)^{\half} \implies \mh(\vxs)^{\half}\mh(\vx)^{-1} = \mm^{-1}\mh(\vxs)^{-\half}.\end{equation}
Next, first-order optimality of $\vxs$ gives $\tau \nabla \psi(\vxs) = \vy$, so we have
\[\nabla \psi(\vx) - \frac 1 \tau \vy = \nabla \psi(\vx) - \nabla \psi(\vxs) = \sum_{i \in [m]} \frac{\va^i}{\vs_i(\vxs)}\Par{\frac 1 {1 - \vr_i} - 1} = \sum_{i \in [m]} \frac{\va^i}{\vs_i(\vxs)} \cdot \frac{\vr_i}{1 - \vr_i}. \]
Therefore, 
\begin{equation}\label{eq:hessian_update}\mh(\vxs)^{-\half}\Par{\nabla \psi(\vx) - \frac 1 \tau \vy} = \mv^\top \Par{\frac{\vr}{\1_m - \vr}},\end{equation}
where the vector division above is entrywise. Rewriting the Newton update in terms of $\vw$,
\begin{align*}
\vw' &\defeq \mh(\vxs)^{\half}(\vx' - \vxs) \\
&= \mh(\vxs)^{\half}(\vx - \vxs) - \mh(\vxs)^{\half}\mh(\vx)^{-1}\Par{\nabla \psi(\vx) - \frac 1 \tau \vy} \\
&= \vw - \mm^{-1} \mv^\top \Par{\frac{\vr}{\1_m - \vr}},
\end{align*}
where we used \eqref{eq:relative_h} and \eqref{eq:hessian_update} in the last line. Further, since $\mv \vw = \vr$,
\[\mm \vw = \mv^\top \diag{\1_m - \vr}^{-2} \vr = \mv^\top \Par{\frac{\vr}{(\1_m - \vr)^2}} \]
where again, the vector division and squaring is entrywise. Combining the two displays above gives
\[\vw' = \mm^{-1}\mv^\top\Par{\frac{\vr^2}{(\1_m - \vr)^2}}.\]
Finally, $\norm{\vr}_\infty \le \norm{\vr}_4 \le \frac 1 4$ implies $\vs(\vx) \le \frac 5 4 \vs(\vxs)$ entrywise, so $\mh(\vx) \succeq \frac {16} {25} \mh(\vxs)$, and $\normsop{\mm^{-1}} \le \frac {25}{16}$ by using $\mm^{-1} = \mh(\vxs)^{\half}\mh(\vx)^{-1}\mh(\vxs)^{\half}$. We also have $\normsop{\mv^\top} = 1$, so the claim follows from
\[\norm{\vw'}_2 \le \frac{25}{16}\sqrt{\sum_{i \in [m]} \frac{\vr_i^4}{(1 - \vr_i)^4}} \le \frac{25}{9} \norm{\vr}_4^2 \le \frac{25}{144} < \frac 1 5. \]
\end{proof}

By applying standard convergence results of Newton's method within a Dikin ellipse of the objective minimizer, we derive the following high-accuracy mode approximation result.

\begin{corollary}\label{cor:mode_finding}
Fix $\vy \in \calY$ and $\Delta \in (0, \half)$, and assume we are given $\vx \in \calX$ with $\norms{\vr(\vx; \vxs_{\vy, \tau})}_4 \le \frac 1 4$. There is an algorithm that in time $O(mn^{\omega - 1}\log\log \frac 1 \Delta)$ outputs a point $\vw$ with 
\[\norm{\vw - \vxs_{\vy, \tau}}_{\vw} \le \Delta,\quad \norm{\vh(\vw)}_2 \le \Delta.\]
\end{corollary}
\begin{proof}
We simply run $O(\log \log \frac 1 \Delta)$ steps of Newton's method, namely, the iteration \eqref{eq:newton}. Lemma~\ref{lem:l4_newton} shows that the first step enters the Dikin ellipse of radius $\frac 1 5$ around $\vxs \defeq \vxs_{\vy, \tau}$. Then, standard results on the convergence of Newton's method in this region (Theorem 2.2.3, \cite{NesterovN94}) and comparisons between the Newton decrement and induced distance in $\mh(\vxs)$ (Theorem 2.2.5, \cite{NesterovN94}) shows that the output point $\vw$ satisfies the stated bounds.
\end{proof}

\subsection{Smoothing}\label{ssec:smoothing}

The goal of this section is to construct a smoothed potential $\bar{F}$ that enjoys global regularity properties (i.e., being well-conditioned and having a Lipschitz Hessian), and whose induced measure has small TV distance to the true backwards conditional. The former property makes our smoothed measure amenable to black-box high-accuracy (polylogarithmic dependence on inverse accuracy) samplers for sufficiently well-behaved densities, e.g., Hamiltonian Monte Carlo \cite{ZhangAC26}.\footnote{We began this project when \cite{ZhangAC26} was the only high-accuracy algorithm with a $o(n^{\omega - 2})$ iteration complexity under second-order smoothness. Since then, many other algorithms have been shown to have this property under a well-conditioned assumption \cite{ChenCRZ26,ChenCLZ26a,ChenCLZ26b}, and any such algorithm also suffices for our sampler.} The main export of this section is the following.

\begin{lemma}\label{lem:smoothed_overlap}
Let $\delta \in (0, \half)$ and $\vw \in \calX$. There is an absolute constant $C > 0$ such that if
\[\tau\geq C\Par{\sqrt n+ \log\frac{m}{\delta}},\quad \norm{\vw - \vxs}_{\vw} \le \frac 1 {16},\quad \norm{\vh(\vw)}_2 \le \frac 1 {300},\]
then there is a potential $\bF: \R^n \to \R$ and measure $\bnu \propto \exp(-\bF)$
such that defining $\nu_{\vw}$ as in Lemma~\ref{lem:pushforward_reparam}, we have $\TV(\nu_{\vw}, \bnu) \le \delta$, and for all $\vz, \vz' \in \R^n$,
\[\frac \tau 3 \id_n \preceq \nabla^2 \bF(\vz) \preceq 3\tau\id_n,\quad \normf{\nabla^2 \bF\Par{\vz} - \nabla^2 \bF(\vz')} \le 5\tau\norm{\vz - \vz'}_2.\]
Further, given $\mv(\vw)$ and $\vh(\vw)$, $\nabla \bF$ can be evaluated in $O(mn)$ time.
\end{lemma}

\textbf{Smoothed potential.} We begin with a helper lemma that forms the basis of our smoothing.

\begin{lemma}\label{lem:smoothed-potential}
Let $f(u) = -\log(1 - u) - u$ for $u < 1$ as in~\eqref{eq:nu_potential}. There exists $\bar{f} : \R \to \R$ such that 
\begin{equation}\label{eq:regularity}\bar{f}(u) = f(u) \text{ for all } u \in \Brack{-\frac 1 4, \frac 1 4}, \quad \bar{f}'' \in \Brack{\frac 1 3, 3}, \quad \bar{f}''' \in [0, 5],\end{equation}
where the last two statements are on all of $\R$. Moreover, $\bar{f}'$ can be evaluated in $O(1)$ time.
\end{lemma}
\begin{proof}
Let $\rho \defeq \frac 1 4$ for shorthand.
    Define $\eta:\R\to\R$ as
    \[\eta(u)=\begin{cases}
    \frac{1-\rho}{(1+\rho)^3}u-\frac{\rho^2(13+3\rho)}{3(1+\rho)^3}, \quad &u\in (-\infty, -3\rho]\\
    \frac{1}{6\rho(1+\rho)^3}u^3+\frac{3}{2(1+\rho)^3}u^2+\frac{2+7\rho}{2(1+\rho)^3}u+\frac{\rho^2(1-6\rho)}{6(1+\rho)^3}, \quad &u\in (-3\rho, -\rho),\\
        \frac{u}{1-u},\quad & u\in [-\rho, \rho] ,\\
        -\frac{1}{6\rho(1-\rho)^3}u^3+\frac{3}{2(1-\rho)^3}u^2+\frac{2-7\rho}{2(1-\rho)^3}u+\frac{\rho^2(1+6\rho)}{6(1-\rho)^3}, \quad &u\in (\rho, 3\rho),\\
         \frac{1+\rho}{(1-\rho)^3}u-\frac{\rho^2(13-3\rho)}{3(1-\rho)^3}, \quad & u\in [3\rho, +\infty).
    \end{cases}\]
    Then define $\bar{f}(u)=\int_0^u \eta(v)\dd v$.  It is straightforward to check that $\bar{f}$ is thrice continuously differentiable. Also, evaluating $\bar{f}'$ amounts to evaluating $\eta$, which takes $O(1)$ time as it is explicit.

    We now verify the conditions \eqref{eq:regularity}. On $[-\frac 1 4, \frac 1 4]$, $\eta = f'$, and $\bar{f}(0) = f(0) = 0$, so $\bar{f} = f$ on this region. Moreover, outside $(-3\rho, 3\rho)$, $\eta$ is affine, so $\eta''$ is nonzero only on $(-3\rho, 3\rho)$, where
    \[\eta''(u) = \begin{cases} \frac{u}{\rho(1 + \rho)^3} + \frac{3}{(1 + \rho)^3} & u \in (-3\rho, -\rho) \\
    \frac{2}{(1-u)^3}& u \in [-\rho, \rho] \\
    -\frac{u}{\rho(1-\rho)^3} + \frac 3{(1 - \rho)^3}& u \in (\rho, 3\rho)
    \end{cases},\]
    i.e., it is two linear functions with a monotone function in between. By checking endpoints, we have that $\bar{f}''' = \eta'' \in [0, \frac{128}{27}] \subset [0, 5]$ and $\bar{f}'' = \eta' \in [\frac{48}{125}, \frac{80}{27}] \subset [\frac 1 3, 3]$, as claimed.
\end{proof}

We next define the main smoothed measure we consider in this section. 

\begin{lemma}\label{lem:smoothed_measure}
Let $\bar{f}: \R \to \R$ be as in Lemma~\ref{lem:smoothed-potential}, and let $\vw \in \calX$. Following \eqref{eq:vdef}, \eqref{eq:relative_gradient}, define
\begin{equation}\label{eq:barF}\bF(\vz) \defeq \tau\Par{\sum_{i \in [m]} \bar{f}(\inprod{\vv^i(\vw)}{\vz}) + \inprod{\vh(\vw)}{\vz} }.\end{equation}
Then for all $\vz, \vz' \in \R^n$,
\begin{equation}\label{eq:hess_regularity}\frac \tau 3 \id_n \preceq \nabla^2 \bF(\vz) \preceq 3\tau\id_n,\quad \normf{\nabla^2 \bF\Par{\vz} - \nabla^2 \bF(\vz')} \le 5\tau\norm{\vz - \vz'}_2.\end{equation}
Moreover, given $\mv(\vw)$ and $\vh(\vw)$, $\nabla \bF$ can be evaluated in $O(mn)$ time.
\end{lemma}
\begin{proof}
For shorthand in this proof, let $\mv \defeq \mv(\vw)$ and $\vv^i \defeq \vv^i(\vw)$. The chain rule shows that
\[\nabla \bF(\vz) = \tau \Par{\mv^\top \baf'(\mv\vz) + \vh(\vw)},\quad \nabla^2 \bF(\vz) = \tau\mv^\top \diag{\baf''(\mv \vz)}\mv,\]
where $\baf'$ and $\baf''$ are applied entrywise. Also, Lemma~\ref{lem:smoothed-potential} shows that
\[\frac 1 3 \id_m \preceq \diag{\baf''(\mv \vz)} \preceq 3\id_m.\]
The first bound in \eqref{eq:hess_regularity} then follows from orthonormality of $\mv$. The second similarly follows from
\begin{align*}\normf{\mv^\top \diag{\baf''(\mv\vz) - \baf''(\mv\vz')}\mv} &\le \normf{\diag{\baf''(\mv\vz) - \baf''(\mv\vz')}} \\
&= \norm{\baf''(\mv\vz) - \baf''(\mv\vz')}_2 \le 5\norm{\vz - \vz'}_2, \end{align*}
where we again used Lemma~\ref{lem:smoothed-potential}. Finally, the runtime cost for computing $\nabla \bF(\vz)$ is dominated by matrix multiplications through $\mv$, since $\baf'$ can be evaluated in $O(1)$ time.
\end{proof}

\textbf{Concentration.} We next establish concentration of the relative slacks under the smoothed measure, so that it approximates the true conditional distribution. We begin with a mean bound.

\begin{lemma}\label{lem:mean_bound}
Let $\bF$ be as in \eqref{eq:barF} for some $\vw \in \calX$, let $\bnu \propto \exp(-\bF)$, and let $Z \sim \bnu$. Then
\[\norm{\E\Brack{\vr(Z; \vw)}}_2 \le 25\Par{\frac{\sqrt n}{\tau} + \norm{\vh(\vw)}_2}.\]
\end{lemma}
\begin{proof}
For brevity we drop the dependences of $\mv$, $\vv^i$, $\vh$, and $\vr$ on $\vw$ throughout this proof, as it is fixed; all relative slacks $\vr$ are measured with respect to $\vw$. Recall from \eqref{eq:frame_facts} that
\[\vm \defeq \E\Brack{\vr(Z)} = \mv\E[Z].\]
To control $\norm{\vm}_2$, our strategy is to use our lower bound on $\baf''$ from Lemma~\ref{lem:smoothed-potential}, which implies 
\begin{equation}\label{eq:sc_lower}\frac 1 3 \norm{\vm}_2^2 \le \inprod{\vm - \0_m}{\baf'(\vm) - \baf'(\0_m)} = \inprod{\vm}{\baf'(\vm)},\end{equation}
where the first inequality used the fundamental theorem of calculus and $\baf'' \ge \frac 1 3$. 

We next bound the right-hand side. Our strategy is to first bound $\norms{\baf'(\vm) - \E[\baf'(\vr(Z))]}_2$, and then characterize the mean term explicitly via the score identity in Fact~\ref{fact:score_virial} incurring a dependence on $\norm{\vh}_2$. First, $\bar{f}''' \le 5$ (Lemma~\ref{lem:smoothed-potential}) and Taylor's theorem imply that 
\[\Abs{\baf'(\vr_i(Z)) - \baf'(\vm_i) - (\vr_i(Z) - \vm_i)\baf''(\vm_i)} \le \frac 5 2(\vm_i - \vr_i(Z))^2,\]
for all $i \in [m]$.
Taking expectations of both sides and using that $\vm_i = \E[\vr_i(Z)]$ gives
\[\Abs{\E\Brack{\baf'(\vr_i(Z))} - \baf'(\vm_i)} \le \frac 5 2 \Var[\vr_i(Z)].\]
On the other hand, $\bF$ is $\frac\tau 3$-strongly convex (Lemma~\ref{lem:smoothed_measure}), so Fact~\ref{fact:bl} yields
\begin{equation}\label{eq:bl_one_slack}\Var\Brack{\vr_i(Z)} = \Var\Brack{\inprod{\vv^i}{Z}} \le \frac{3\norm{\vv_i}_2^2}{\tau}.\end{equation}
Combining the above two displays shows that
\begin{equation}\label{eq:combine_mean_1}
\norm{\baf'(\vm) - \E[\baf'(\vr(Z))]}_2 \le \frac{15}{2\tau}\sqrt{\sum_{i \in [m]}\norm{\vv^i}_2^4} \le \frac{15}{2\tau}\sqrt{\sum_{i \in [m]}\norm{\vv^i}_2^2} = \frac{15\sqrt{n}}{2\tau},
\end{equation}
using orthonormality of $\mv$. We account for the expectation term using Fact~\ref{fact:score_virial}, which shows
\[\0_n = \E\Brack{\nabla \bF(\vz)} = \tau\Par{\mv^\top \E\Brack{\baf'(\vr(Z))} + \vh},\]
i.e., that $-\vh = \mv^\top \E\Brack{\baf'(\vr(Z))}$. Finally, combining with \eqref{eq:sc_lower} and \eqref{eq:combine_mean_1}, we have
\begin{align*}
\frac 1 3\norm{\vm}_2^2 &\le \inprod{\vm}{\baf'(\vm) - \E[\baf'(\vr(Z))]} + \inprod{\vm}{\E[\baf'(\vr(Z))]} \\
&\le \norm{\vm}_2 \norm{\baf'(\vm) - \E[\baf'(\vr(Z))]}_2 + \inprod{\E[Z]}{\mv^\top \E\Brack{\baf'(\vr(Z))}} \\
&\le \norm{\vm}_2\Par{\frac{15\sqrt{n}}{2\tau} + \norm{\vh}_2},
\end{align*}
since $\norm{\E[Z]}_2 = \norm{\mv \E[Z]}_2 = \norm{\vm}_2$.
Dividing through by $\frac 1 3 \norm{\vm}_2$ gives the claim.
\end{proof}

This lets us show that $\bnu$ assigns substantial mass to (the pushforward of) $\calB(\vw, \frac 1 4)$.

\begin{lemma}\label{lem:bnu_concentration}
    Let $\delta\in (0, \half)$, let $\vw\in \calX$, and let $\calZ(\vw, \frac 1 4)$ be the image of $\calB(\vw, \frac{1}{4})$ under the map $\vx\mapsto \vz(\vx; \vw)$. There is an absolute constant $C> 0$ such that if
    \begin{equation*}
        \tau\geq C\Par{\sqrt n+ \log\frac{m}{\delta}}\quad \text{and}\quad \norm{\vh(\vw)}_2 \leq \frac{1}{300},
    \end{equation*}
    then defining $\bnu$ as in Lemma~\ref{lem:mean_bound}, we have $\bar{\nu}(\calZ(\vw, \frac 1 4)) \geq 1 - \delta$.
\end{lemma}
\begin{proof}
As in Lemma~\ref{lem:mean_bound} we do not list dependences on $\vw$. For $\lam \in \R$, let $Z_\lam$ have density
\[\bnu_\lam(\vz) \propto \exp(-\bF(\vz)+ \lam\vr_i(\vz)).\]
Moreover, the variance bound \eqref{eq:bl_one_slack} still holds with $Z_\lam$ in place of $Z$, because the potential $\bF - \lam \vr_i$ has the same Hessian as $\bF$ (as $\vz \mapsto \vr_i(\vz)$ is linear).
Hence, Lemma~\ref{lem:main_concentration} with $G \gets \vr_i$, $\sig^2 \gets \frac 3 \tau$, and $t \gets \frac 1 8$ implies
\[\Pr\Brack{\Abs{\vr_i(Z) - \E[\vr_i(Z)]} \ge \frac 1 8} \le 2\exp\Par{-\frac \tau {384}} \le \frac \delta m,\]
for our choice of $\tau$. A union bound gives that with probability $\ge 1 - \delta$, we have $\norm{\vr(Z) - \E[\vr(Z)]}_\infty \le \frac 1 8$. The conclusion follows because Lemma~\ref{lem:mean_bound} and the assumed bounds imply
\[\norm{\E[\vr(Z)]}_\infty \le 25\Par{\frac{\sqrt n}{\tau} + \norm{\vh(\vw)}_2} \le \frac 1 8.\]
\end{proof}

\textbf{Agreement of smoothed and original measures.} We conclude by using Lemma~\ref{lem:bnu_concentration} and our earlier results from Section~\ref{ssec:concentration} to show that sampling from the smoothed measure is close to the true backwards conditional distribution in total variation distance.

\begin{proof}[Proof of Lemma~\ref{lem:smoothed_overlap}]
The potential is $\bF$ as constructed in Lemma~\ref{lem:smoothed_measure}, so \eqref{eq:hess_regularity} verifies the stated regularity assumptions, and Lemma~\ref{lem:smoothed_measure} also gives the runtime cost of evaluating $\nabla \bF$. 

We next claim that $\nu_{\vw}(\calZ(\vw, \frac 1 4)) \ge 1 - \delta$. From Lemma~\ref{lem:main_infty_conc}, it is enough to show that $\calB(\vxs, \frac 1 8) \subseteq \calB(\vw, \frac 1 4)$,
because $\calZ(\vw, \frac 1 4)$ is the image of $\calB(\vw, \frac 1 4)$ under the map $\vx \mapsto \vz(\vx; \vw)$. 

To see this, we first note that for all $i \in [m]$,
\[\Abs{\vr_i(\vxs; \vw) }= \Abs{\inprod{\vv^i(\vw)}{\mh(\vw)^{\half}(\vxs - \vw} } \le \norm{\vw - \vxs}_{\vw} \le \frac 1 {16}. \]
Now for any $\vx \in \calB(\vxs, \frac 1 8)$, we have
\begin{align*}\vr_i(\vx; \vw) &= \frac{\vs_i(\vw) - \vs_i(\vx)}{\vs_i(\vw)} = 1 - \frac{\vs_i(\vx)}{\vs_i(\vxs)} \cdot \frac{\vs_i(\vxs)}{\vs_i(\vw)} \\
&= 1 - \Par{1 - \vr_i(\vx; \vxs)}\Par{1 - \vr_i(\vxs; \vw)} \\
&=  \vr_i(\vx; \vxs) + \vr_i(\vxs; \vw) - \vr_i(\vx; \vxs)\vr_i(\vxs; \vw),
\end{align*}
and substituting $\norm{\vr(\vx; \vxs)}_\infty \le \frac 1 8$ and $\norm{\vr(\vxs; \vw)}_\infty \le \frac 1 {16}$ proves
$\vx \in \calB(\vw, \frac 1 4)$. 

Thus, we have shown $\nu_{\vw}(\calZ(\vw, \frac 1 4)) \ge 1 - \delta$ and $\bnu(\calZ(\vw, \frac 1 4)) \ge 1 - \delta$ (Lemma~\ref{lem:bnu_concentration}). By Lemma~\ref{lem:smoothed-potential}, the densities $\nu_{\vw}$ and $\bnu$ are proportional to each other on $\calZ(\vw, \frac 1 4)$, and hence the conditional distributions of $\nu_{\vw}$ and $\bar{\nu}$ on $\calZ(\vw, \frac 1 4)$ coincide. Define $\rho$ to be their common conditional distribution on $\calZ(\vw, \frac 1 4)$; note that its density outside is $0$. Then for $\vz \in \calZ(\vw, \frac 1 4)$,
\[\min\Brace{\nu_{\vw}(\vz), \bnu(\vz)} \ge \min\Brace{\nu_{\vw}\Par{\calZ\Par{\vw, \frac 1 4}}, \bnu\Par{\calZ\Par{\vw, \frac 1 4}}}\rho(\vz) \ge (1 - \delta)\rho(\vz).\]
Therefore,
\[\TV\Par{\nu_{\vw}, \bnu} = 1 - \int \min\Brace{\nu_{\vw}(\vz), \bnu(\vz)} \dd \vz \le 1 - (1 - \delta) \int \rho(\vz) \dd \vz = \delta.\]
\end{proof}

\subsection{Putting things together}\label{ssec:main_theorem}

In this section, we give our main backwards oracle implementation. Our oracle applies the sampler of \cite{ZhangAC26}, whose properties we recall here, to the smoothed measure in Lemma~\ref{lem:smoothed_overlap}.

\begin{proposition}[Theorem 3.4, \cite{ZhangAC26}]\label{prop:hmc}
Let $\delta \in (0, \half)$, $\kappa > 1$, $\beta, L > 0$, and let $\nu \in \calP(\R^n)$ have $\nu \propto \exp(-V)$ for a potential $V: \R^n \to \R$ that satisfies, for all $\vx, \vx' \in \R^n$,
\[\frac \beta \kappa \id_n \preceq \nabla^2 V(\vx) \preceq \beta \id_n,\quad \normf{\nabla^2 V(\vx) - \nabla^2 V(\vx')} \le L\norm{\vx - \vx'}_2 .\]
There is an algorithm, $\HMC(V, \kappa, \beta, L, \delta)$, that uses
\[O\Par{n^{\frac 1 4}\log^2\Par{\frac n \delta}\cdot \poly\Par{\kappa, \frac{L}{(\beta/\kappa)^{\frac 3 2}}}}\]
queries to $\nabla V$, and produces a sample whose law has total variation distance at most $\delta$ to $\nu$.\footnote{The statement of Theorem 3.4, \cite{ZhangAC26} formally assumes that we have access to the mode of $V$. However, their analysis goes through even if we only have a point $\vw'$ where $\|\nabla V(\vw')\|_2 = \poly(\frac \delta {n})$, up to condition number-dependent factors (see discussion in \cite{ZhangAC26} before the theorem statement). In our setting, the initial distance bound from $\vw$ to the mode yielded by Lemma~\ref{lem:smoothed_overlap} is small enough that running well-conditioned gradient descent to compute $\vw'$ does not dominate the stated sampling runtime. We omit further details for brevity.}
\end{proposition}

We now state our backwards oracle implementation in Algorithm~\ref{alg:backwards}.

\begin{algorithm2e}[h!]
\caption{$\Backwards(\vx_0, \vy, \tau, N, \delta)$}
\label{alg:backwards}
\DontPrintSemicolon
\textbf{Input: }previous iterate $\vx_0\in \calX$, forward draw $\vy\in\R^n$, regularization parameter $\tau\in \R_{>0}$, number of Newton iterations $N \in \N$, target accuracy $\delta \in (0, \half)$\;
\For{$i=0,\ldots, N - 1$}{\label{line:newton_start}
$\vx_{i + 1} \leftarrow \vx_i - \mh(\vx_i)^{-1}\Par{\nabla \psi(\vx_i) - \frac{1}{\tau}\vy}$\tcp*{Newton iteration, $\mh$ defined in \eqref{eq:hess_norm}.}
$\vx_{i + 1} \gets \vx_0$ if $\vx_{i + 1} \not\in \calX$\;
}\label{line:newton_end}
$\vw \gets \vx_N$\;
$\vz \gets \HMC(\bF, 9, 3\tau, 5\tau, \frac \delta 2)$\tcp*{$\HMC$ defined in Proposition~\ref{prop:hmc}, $\bF$ defined in Lemma~\ref{lem:smoothed_measure}.}\label{line:hmc}
$\vx \gets \vw + \mh(\vw)^{-\half}\vz$\label{line:undo_pushforward} if $\vw + \mh(\vw)^{-\half}\vz \in \calX$, else $\vx \gets \vx_0$ \;
\textbf{Return: } $\vx$
\end{algorithm2e}

\begin{proposition}\label{prop:backwards_part1}
Let $\delta \in (0, \half)$, let $\vx_0 \sim \mu$ for $\mu \in \calP(\calX)$ satisfying $\chi^2(\mu \| \pi) \le 3$, and let $\vy \sim \pi_\tau^{Y \mid X = \vx_0}$. There are universal constants $C > 0$ and $N \in \N$ such that, if
\[\tau \ge C\Par{\sqrt{n\log \frac m \delta} + \log \frac m \delta},\]
then with probability $\ge 1 - \delta$ over the joint randomness of $(\vx_0, \vy)$, Algorithm~\ref{alg:backwards} outputs $\vx$ with
\[\TV\Par{\Law(\vx \,| \,\vx_0, \vy), \pi_\tau^{X \mid Y = \vy}} \le \delta.\]
Moreover, Algorithm~\ref{alg:backwards} can be implemented in time
\[O\Par{mn^{\omega - 1} + mn^{\frac 5 4}\log^2\Par{\frac n \delta}}.\]
\end{proposition}
\begin{proof}
For the correctness, we first condition on the success of Lemma~\ref{lem:l4_slacks}, which gives the failure probability over $\vx_0$. Under the event in Lemma~\ref{lem:l4_slacks}, Corollary~\ref{cor:mode_finding} shows that if $N$ is a sufficiently large constant (i.e., so that $\Delta \le \frac 1 {300}$), Lines~\ref{line:newton_start} to~\ref{line:newton_end} successfully return a point $\vw$ such that the precondition in Lemma~\ref{lem:smoothed_overlap} holds, with $\delta \gets \frac \delta 2$, so that $\TV(\nu_{\vw}, \bnu) \le \frac \delta 2$. Further, Lemma~\ref{lem:smoothed_overlap} shows that the potential $\bF$ defining $\bnu$ satisfies the conditions of Proposition~\ref{prop:hmc}.
By Proposition~\ref{prop:hmc}, Line~\ref{line:hmc} outputs a sample whose law has TV distance $\frac \delta 2$ to $\bnu$, so it is within TV distance $\delta$ to $\nu_{\vw}$. The conclusion follows because Line~\ref{line:undo_pushforward} undoes the operation that pushes forward $\pi_\tau^{X\mid Y = \vy}$ to $\nu_{\vw}$.

For the runtime, Corollary~\ref{cor:mode_finding} falls within the allotted budget when $\frac 1 \Delta = O(1)$, and computing $\mv(\vw)$ and $\vh(\vw)$ as required by Lemma~\ref{lem:smoothed_overlap} takes $O(mn^{\omega - 1})$. The remaining cost is due to combining Proposition~\ref{prop:hmc} with the $O(mn)$-time implementation of $\nabla \bF$ from Lemma~\ref{lem:smoothed_overlap}.
\end{proof}

\section{Functional Cooling}\label{sec:cooling}

In this section, we give a functional cooling framework, which provides the $\chi^2$ warm starts required by Proposition~\ref{prop:prox-sample-guarantee} and Section~\ref{sec:part_1_backwards_step} at a polylogarithmic overhead to our overall iteration complexity. Our framework is inspired by techniques in previous works \cite{LovaszV06, CousinsV18}, that reduce the cost of sampling convex bodies in the membership oracle setting from a cold start. The high-level approach is to utilize a cooling schedule, by considering a sequence of distributions $\pi_1, \ldots, \pi_T$ with $\pi_T$ being the target distribution. This schedule should maintain two properties: that each $\pi_i$ is a sufficient warm start to $\pi_{i + 1}$, and that the overall cost of sampling from all of the $\{\pi_i\}_{i \in [T]}$ is not much more than the cost of sampling from $\pi_T$ directly. Prior works chose each $\pi_i$ to be $\pi_T$ multiplied by an appropriate Gaussian, which for a carefully chosen cooling schedule, maintains both properties. In particular, early cooling iterations are cheaper due to improved isoperimetry from large Gaussian components, which helps bound the cost of sampling from all cooled measures.

%One such construction is to multiply the density of $\pi_T$ with a highly concentrated gaussian to obtain $\pi_0$ (which makes it easy to sample from via, e.g., rejection sampling), and then successively flatten the gaussian by increasing its variance to obtain $\pi_1, \pi_2, \ldots$ until the distribution is close to uniform. For a carefully chosen cooling schedule, one can then chain together warm starts by using an approximate sample from $\pi_t$ as an initial warm start for sampling from $\pi_{t+1}$. \kjtian{I don't think we need to go too much into the implementation details yet, just communicate the main idea}

Our framework builds on the recent work of \cite{Kook2024Gaussian}, who extended the Gaussian cooling scheme of \cite{CousinsV18} to use self-concordant barriers. This is particularly relevant to sampling from the uniform distribution over a polytope, which is compatible with the geometry of the log barrier (Lemma~\ref{lem:isoperimetry}). To keep our paper self-contained, we briefly rederive the main facts used from \cite{Kook2024Gaussian}, but note that our framework was already present in that work (see their Section 3.2). Our main contribution is to show that this strategy yields only polylogarithmic overhead when coupled with the FSL-based implementation of each cooling stage, and to streamline the presentation.

In Section~\ref{ssec:cstotv}, we abstract the requirement of each cooling stage in terms of a $\chi^2$-to-TV oracle (Definition~\ref{def:cs_to_tv_oracle}), which we implement by combining Propositions~\ref{prop:prox-sample-guarantee},~\ref{prop:forward}, and~\ref{prop:backwards_part1}. In Section~\ref{ssec:cooling}, we then present the cooling framework of \cite{Kook2024Gaussian} and apply it to polytope sampling.

\subsection{$\chi^2$-to-TV oracle implementation}\label{ssec:cstotv}

Throughout this section, we follow the notation of Section~\ref{ssec:logbarrier}. Our framework works with cooled variants of the uniform measure $\pi$, i.e., where a multiple of the log barrier $\psi$ is added to the potential. We denote the multiple by $t \ge 0$, termed an \emph{inverse temperature} in the statistical physics literature, and define the corresponding cooled measure and its normalizing constant by
\begin{equation}\label{eq:cooling-path}
    \pi_t \propto \exp(-t\psi)\pi, \quad Z(t) \defeq \int \exp(-t\psi(\vx))\pi(\vx)\dd \vx.
\end{equation}
To keep our notation consistent, $\tau \in \N$ always denotes a regularization parameter for use with FSL (Proposition~\ref{prop:prox-sample-guarantee}), whereas $t \in \R_{\ge 0}$ denotes an inverse temperature. Analogously to \eqref{eq:joint-density}, we set
\begin{equation}\label{eq:cooling-path-joint}
    \tpi_{t, \tau}(\vx, \vy) \defeq \exp\Par{\inprod{\vx}{\vy} - \phi^{\ast \tau}(\vy) - \tau\psi(\vx)}\pi_t(\vx)
\end{equation}
to be the corresponding joint density, and $\pi_{t,\tau}^{X \mid Y = \vy}$ and $\pi_{t, \tau}^{Y \mid X = \vx}$ denote its conditional distributions, where we drop the parameter $t$ from subscripts if $t = 0$.

Our sampler for each $\pi_t$ requires the following simple extension of Proposition~\ref{prop:backwards_part1}.

\begin{corollary}\label{cor:backwards_cooled}
Let $t \ge 0$, let $\delta \in (0, \half)$, let $\vx_0 \sim \mu$ for $\mu \in \calP(\calX)$ satisfying $\chi^2(\mu \| \pi_t) \le 3$, and let $\vy \sim \pi_{t, \tau}^{Y \mid X = \vx_0}$. There is a universal constant $C > 0$ such that, if 
\[\tau \ge C\Par{\sqrt{n\log \frac m \delta} + \log \frac m \delta},\]
then with probability $\ge 1 - \delta$ over the joint randomness of $(\vx_0, \vy)$, Algorithm~\ref{alg:backwards} with $\tau \gets \tau + t$ outputs $\vx$ with
\[\TV\Par{\Law\Par{\vx \mid \vx_0, \vy},\pi_{t,\tau}^{X \mid Y = \vy}} \le \delta.\]
\end{corollary}
\begin{proof}
The backwards conditional distribution $\pi_{t, \tau}^{X \mid Y = \vy}$ is precisely
\[\calT_{\vy}\exp\Par{-(\tau + t)\psi}\pi.\]
Moreover, the only place in Section~\ref{sec:part_1_backwards_step} where we specifically used the base measure was in the approximate mode finding in Section~\ref{ssec:mode} (specifically, Lemma~\ref{lem:l4_slacks}). It is straightforward to verify that the argument in Lemma~\ref{lem:l4_slacks} only uses concentration properties of $\vx \sim \pi_{t,\tau}^{X \mid Y= \vy}$ through a lower bound on $\tau$, which clearly still applies here because $\tau + t \ge \tau$.
Thus, the entirety of Section~\ref{sec:part_1_backwards_step} goes through unchanged (other than setting $\tau \gets \tau + t$) and the claim follows.
\end{proof}

We next leverage Propositions~\ref{prop:prox-sample-guarantee} and~\ref{prop:forward} with Corollary~\ref{cor:backwards_cooled} to show that we can sample from each cooled measure $\pi_t$ under a suitable warm start. 
Our mixing guarantee is formalized as follows:
\begin{definition}\label{def:cs_to_tv_oracle}
Let $\mu \in \calP(\calX)$ and let $\calO: \calX \to \calX$ be randomized. We call $\calO$ a \emph{$(\beta, \delta)$-$\chi^2$-to-TV oracle for $\mu$} if, given $\vx_0 \in \calX$ with $\chi^2(\Law(\vx_0) \| \mu) \le \beta$, it outputs $\vx \in \calX$ with $\TV(\Law(\vx), \mu) \le \delta$.
\end{definition}

Our main result in this section is a $(3,\delta)$-$\chi^2$-to-TV oracle for each $\pi_t$. We state our implementation in Algorithm~\ref{alg:warm-start-LLT}. Algorithm~\ref{alg:warm-start-LLT} almost directly implements FSL as required by Proposition~\ref{prop:prox-sample-guarantee}, but it uses an approximate backwards step.
To reason about the TV approximation error incurred by the calls to $\Backwards$ in Line~\ref{line:backwards}, we require one helper fact.

\begin{algorithm2e}[ht!]
\caption{$\WarmLLTProx$($\vx_0, \tau, t, \delta$) }
\label{alg:warm-start-LLT}
\DontPrintSemicolon
\textbf{Input: }initial $\vx_0\in \calX$, regularization parameter $\tau\in \N$, cooling parameter $t \ge 0$,  target error $\delta \in (0, \half)$ \;
$K \gets \lceil 32 \min\{m\tau, 1 + \frac \tau t\}\log \frac 1 \delta\rceil $\;
\For {$i = 0, \ldots, K - 1$}{
$\vy_i\leftarrow \Forwards(\vx_i, \tau)$\label{line:forwards_warm} \tcp*[r]{See Algorithm~\ref{alg:forwards}.}
$\vx_{i+1}\leftarrow \Backwards(\vx_i, \vy_i, \tau + t, N, \frac{\delta}{4K}$) \;\tcp*[r]{See Algorithm~\ref{alg:backwards}. $N$ is the constant from Proposition~\ref{prop:backwards_part1}.}\label{line:backwards}
}
\textbf{Return: }$\vx_K$
\end{algorithm2e}

\begin{lemma}\label{lem:remove_bad_set}
    Let $(S,\calS)$ and $(T, \calT)$ be measurable spaces, and let $\mathsf{K}, \widetilde{\mathsf{K}}: S \times \calT\rightarrow [0,1]$ be Markov kernels. Let $E\in\calS$ be a measurable set such that
    \begin{equation*}
        \mathsf{TV}\Par{\sfk(s, \cdot ), \tsfk(s, \cdot)} \leq \delta  \text{ for all } s\in E.
    \end{equation*}
    Then, for any probability measure $\mu$ on $S$ such that $\mu(E) \geq 1 - \delta$, we have
    \begin{equation*}
        \mathsf{TV}\Par{\mu\mathsf{K}, \mu\widetilde{\mathsf{K}}} \leq 2\delta.
    \end{equation*}
\end{lemma}
\begin{proof}
    For any measurable $B\in\calT$, we have
    \begin{align*}
        \Abs{ \mu\mathsf{K}(B) - \mu\widetilde{\mathsf{K}}(B) }
        =\Abs{\int \Par{\sfk(s, B) - \tsfk(s, B)} \mu(\dd s)}
        \leq\Abs{\int_E \delta \mu(\dd s)}
        +\Abs{\int_{E^c} \mu(\dd s)}
        \leq 2\delta.
    \end{align*}
    Taking the supremum over $B \in \calT$ completes the proof.
\end{proof}

We briefly contextualize Lemma~\ref{lem:remove_bad_set}, which states that if kernels $\sfk$ and $\tsfk$ agree on a high-probability event $E \in \calS$ under $\mu$, then applying $\widetilde{\sfK}$ instead of $\sfK$ to $\mu$ incurs small TV error. In our application, $\sfK$ and $\widetilde{\sfK}$ are respectively the exact and approximate backwards steps, $S = \calX \times \calY$ corresponds to the pair $(\vx_0, \vy)$ passed as input, and $T = \calX$ corresponds to the outcome space of the backwards step. This lets us set $E \in \calS$ to the high-probability event guaranteed by Corollary~\ref{cor:backwards_cooled}. 

We conclude with our main $\chi^2$-to-TV oracle guarantee.

\begin{lemma}\label{lem:warm-llt-prox-guarantee}
    Let $t \ge 0$, $\delta \in (0, \half)$, let $\pi_t$ be defined as in \eqref{eq:cooling-path}.
    Further, suppose $\tau \in \N$ satisfies, for an absolute constant $C > 0$, and $K \defeq \lceil 32\min\{m\tau, 1 + \frac \tau t\}\log \frac 1 \delta\rceil$,
    \begin{equation*}
    \tau \ge C\Par{\sqrt{n\log \frac {mK}\delta} + \log \frac {mK} \delta}.
    \end{equation*}
    Then
    Algorithm~\ref{alg:warm-start-LLT}, as a function of its input $\vx_0$, is a $(3,\delta)$-$\chi^2$-to-TV oracle for $\pi_t$ that runs in time
    \[O\Par{\Par{m\tau + mn^{\omega - 1} + mn^{\frac 5 4}\log^2\Par{\frac {nK} \delta}} \cdot K}.\]
\end{lemma}
\begin{proof} 
The runtime is immediate from Propositions~\ref{prop:forward} and~\ref{prop:backwards_part1}, and the iteration count in Algorithm~\ref{alg:warm-start-LLT}. 

It remains to prove correctness. 
    Let $\Brace{\vx_i}_{i\in[K]}$ be the sequence of iterates produced in Algorithm~\ref{alg:warm-start-LLT}. Let $\sfP$ denote the Markov kernel corresponding to one step of the ideal $\LLTProx$ algorithm targeting $\pi_t$ (Lines~\ref{line:prox-sample-forward} and~\ref{line:prox-sample-backward} of Algorithm~\ref{alg:prox-sample}, where we take $\tau\leftarrow \tau+t$ in the backwards step), and $\tsfP$ be the Markov kernel corresponding to one step of the realized algorithm (Lines~\ref{line:forwards_warm} and~\ref{line:backwards} of Algorithm~\ref{alg:warm-start-LLT}). Define
    \begin{equation*}
        \mu_0 \defeq \Law(\vx_0),\quad \tmu_i \defeq \mu_0\tsfP^i = \Law(\vx_i), \quad \mu_i \defeq \mu_0\sfP^i.
    \end{equation*}
    Since the guarantee of Proposition~\ref{prop:backwards_part1} is given by the probability over the joint randomness of its input $(\vx, \vy)$, this motivates us to formulate the Markov kernels of the forwards and backwards steps in the following way. Let $S \defeq \calX \times \calY$, $T \defeq \calX$, and let $\calS \defeq \calB(S)$, $\calT \defeq \calB(T)$ be the associated Borel $\sigma$-algebras. Let $\sfF: \calX \times \calS \rightarrow [0, 1]$ and $\sfB, \widetilde{\sfB} : S \times \calT \rightarrow \R$ be Markov kernels corresponding to the forwards step, ideal backwards step, and implemented backwards step (recall that the forward step admits an exact implementation), respectively. Formally,
    \begin{align*}
        \sfF(\vx,\cdot) := \delta_\vx \otimes \pi_{t, \tau}^{Y|X=\vx},\quad \sfB((\vx, \vy), \cdot) \defeq \pi_{t, \tau}^{X|Y=\vy},\\
        \widetilde{\sfB}((\vx, \vy), \cdot) \defeq \Law\Par{\Backwards\Par{\vx, \vy, \tau + t, N, \frac \delta {4K}}}.
    \end{align*}
    Define $E$ to be the set of all $(\vx, \vy)\in S$ such that
    \begin{equation*}
        \TV\Par{\sfB((\vx, \vy), \cdot ), \widetilde{\sfB}((\vx, \vy), \cdot )} \leq \frac \delta {4K}.
    \end{equation*}
    Note that by construction,
    \begin{equation*}
        \mu_i\sfP  = \mu_i\sfF\sfB,\quad \mu_i\tsfP = \mu_i \sfF\widetilde{\sfB}.
    \end{equation*}
    Furthermore,
    \begin{equation*}
        \chi^2(\mu_i\Vert \pi_t) = \chi^2(\mu_0 \sfP^i\Vert \pi_t) \leq \chi^2(\mu_0 \sfP^i\Vert \pi_t\sfP^i) \leq \chi^2(\mu_0 \Vert \pi_t) \leq 3,
    \end{equation*}
    where we used that $\pi_t$ is stationary under the ideal $\LLTProx$ kernel $\sfP$.
    Thus, we apply Corollary~\ref{cor:backwards_cooled} to obtain $(\mu_i\sfF)(E) \geq 1 - \frac{\delta}{4K}$, and so by Lemma~\ref{lem:remove_bad_set},\begin{equation}\label{eq:ideal_vs_approx_helper}
        \TV(\mu_i\sfP, \mu_i\tsfP) = \TV(\mu_i \sfF\sfB, \mu_i\sfF\widetilde{\sfB}) \leq \frac{\delta}{2K}. 
    \end{equation}
    We may now bound
    \begin{align*}
        \TV(\tmu_i, \mu_i)
        &= \TV(\tmu_{i-1}\tsfP, \mu_{i-1}\sfP) \\
        &\le \TV(\tmu_{i-1}\tsfP, \mu_{i-1}\tsfP) + \TV(\mu_{i-1}\tsfP, \mu_{i-1}\sfP)\\
        &\leq \TV(\tmu_{i-1}, \mu_{i-1}) + \frac{\delta}{2K},
    \end{align*}
    where the last line uses \eqref{eq:ideal_vs_approx_helper}. Iterating the above display gives $\TV(\tmu_K, \mu_K) \leq \frac \delta 2$. Finally,
    \begin{equation*}
        \TV(\Law(\vx_K), \pi_t) \leq \TV(\tmu_K, \mu_K) + \TV(\mu_K, \pi_t) \leq \frac \delta 2 + \frac{1}{2} \sqrt{\chi^2(\mu_K\Vert \pi_t)} \leq \delta.
    \end{equation*}
    The last inequality followed because $\pi_t$ satisfies an $\alpha_t$-$\psi$-PI, for $\alpha_t \defeq \max\{t, \frac 1 {16m}\}$, where the first term in the maximum comes from Fact~\ref{fact:bl}, and the second comes from Lemma~\ref{lem:isoperimetry} (via Remark~\ref{rem:slc}). Taking $K$ as stated and applying Proposition~\ref{prop:prox-sample-guarantee} and $\sqrt{3} \cdot (1 + \frac{\alpha_t}{\tau})^{-K}\le \delta$ gives the claim.
\end{proof}
\subsection{Cooling framework}\label{ssec:cooling}

We now present our main cooling framework in Algorithm~\ref{alg:annealing}, which builds on the $\chi^2$-to-TV oracle implementation in Lemma~\ref{lem:warm-llt-prox-guarantee}. Algorithm~\ref{alg:annealing} proceeds in four phases, the first of which (Lines~\ref{line:phase_1_start} to~\ref{line:phase_1_end}) initializes at a draw from an appropriate Gaussian. The remaining phases simply iteratively call our oracle from Lemma~\ref{lem:warm-llt-prox-guarantee} at different inverse temperatures $t$.

\begin{algorithm2e}[h!]
\caption{$\FunctionalCooling(\delta, \{\calO_t\}_{t \ge 0})$}
\label{alg:annealing}
\DontPrintSemicolon
\makebox[0pt][l]{\hspace*{\hsize}\tikzmark{phase-right}}%
\textbf{Input:} target error $\delta \in (0, \half)$, oracles $\{\calO_t\}_{t \ge 0}$ where $\calO_t$ is a 
$(3, \frac \delta {2T})$-$\chi^2$-to-TV oracle for $\pi_t$, $T \defeq C' (m + \sqrt{n}\log(\frac m \delta))$ for the constant $C'$ in Corollary~\ref{cor:part-1-iteration-count} \;
\BlankLine
\BlankLine
$D \gets n + \log \frac 8 \delta$, $t\leftarrow t_1 \defeq 10^6D^3$  \;
\BlankLine
\BlankLine
\tikzmark{phase-1-start}$\hvx\leftarrow \AC(\calX, \frac 1 {100\sqrt{tD}})$\tcp*[r]{See Proposition~\ref{prop:ac}.}\label{line:phase_1_start}
$\vx \sim \calN(\hvx, \frac 1 t\mh(\hvx)^{-1})$\;
\If{$\norm{\vx - \hvx}_{\hvx} > 4\sqrt{\frac D t}$}{\label{line:bigrad}
$\vx \gets \hvx$
}(\tikzmark{phase-1-end})\label{line:phase_1_end}
\BlankLine
\BlankLine
\tikzmark{phase-2-start}\While{$t\geq \frac n m$}{\label{line:phase_2_start}
$\vx\leftarrow \calO_t(\vx)$ \;
$t \leftarrow t\cdot (1+\frac 1 {\sqrt{n}})^{-1}$\;
}(\tikzmark{phase-2-end})\label{line:phase_2_end}
\BlankLine
\BlankLine
\tikzmark{phase-3-start}\While{$t\geq \frac 1 m$}{\label{line:phase_3_start}
$\vx\leftarrow \calO_t(\vx)$ \;
$t \leftarrow t\cdot (1+\frac 1 {\sqrt{mt}})^{-1}$\;\label{line:phase_3_update}
}(\tikzmark{phase-3-end})\label{line:phase_3_end}
\BlankLine
\BlankLine
\tikzmark{phase-4-start}$\vx\leftarrow \calO_t(\vx)$ \;\label{line:phase_4_start}
\tikzmark{phase-4-end}$\vx\leftarrow \calO_0(\vx)$\;\label{line:phase_4_end}
\textbf{Return: }$\vx$
\PhaseBrace{phase-1-start}{phase-1-end}{Phase 1}
\PhaseBrace{phase-2-start}{phase-2-end}{Phase 2}
\PhaseBrace{phase-3-start}{phase-3-end}{Phase 3}
\PhaseBrace{phase-4-start}{phase-4-end}{Phase 4}
\end{algorithm2e}

We separate our analysis of Algorithm~\ref{alg:annealing} into a runtime bound and a correctness analysis, the latter of which shows that our successive $\calO_t$ calls maintain a $\chi^2$ warm start.

\textbf{Bounding the runtime.} Our main runtime claim on Algorithm~\ref{alg:annealing} is provided in Corollary~\ref{cor:part-1-iteration-count}, which follows by bounding the runtime of each of the algorithm's phases.
We first discuss the runtime of Phase 1 of Algorithm~\ref{alg:annealing}. All steps in this phase clearly take $O(mn^{\omega - 1})$ time, except for Line~\ref{line:phase_1_start}. We bound the cost of this initialization by appealing to standard results in linear programming.

\begin{proposition}\label{prop:ac}
Let $\calX$ be a bounded polytope induced by $(\ma, \vb)$, let $0 < r < R$, and let $\Delta \in (0, \half)$. There is an algorithm, $\AC(\calX, \Delta)$, that if given $\vx_0 \in \calX$ with $\ball_2(\vx_0, r) \subset \calX \subset \ball_2(\vx_0, R)$, outputs a point $\hvx$ with
\[\norm{\nabla \psi(\hvx)}_{\hvx, *} \le \Delta,\]
following the notation of Section~\ref{ssec:logbarrier}.
Moreover, the algorithm runs in time
\[O\Par{mn^{\omega - 1}\Par{\sqrt{m}\log\frac{mR}{r} + \log\log \frac 1 \Delta}}.\]
\end{proposition}
\begin{proof}
This follows from Section 10.6.2 of \cite{Vishnoi21}, combined with Newton iteration. The strategy is the standard one in the literature, where we first run the central path initialized at $\vx_0$ (which works for a particular choice of linear objective given by $\vc$), and then switch the linear objective to $\0_n$ to minimize $\psi$ directly.
Let $\vc \defeq -\nabla \psi(\vx_0)$, and consider the central path minimization problems \[\min_{\vx \in \calX} F_\eta(\vx) \defeq \eta \inprod{\vc}{\vx} + \psi(\vx).\]
First-order optimality and the definition of $\vc$ show that $\vx_0$ minimizes $F_1(\vx)$, and so initializing the short-step path-following method of Section 10.6.2, \cite{Vishnoi21} (see also \cite{NesterovN94}) gives $\vx \in \calX$ with
\begin{equation}\label{eq:central_path_c}\norm{\nabla F_\eta(\vx)}_{\vx, *} \le \frac 1 {12},\quad \eta = \frac{r}{24mR}, \end{equation}
in $O(\sqrt{m} \log \frac 1 \eta)$ Newton steps. Each Newton step takes time $O(mn^{\omega - 1})$ which falls within our runtime budget.
Next, observe that the containment assumption $\ball_2(\vx_0, r) \subset \calX$ implies
\[\vs_i(\vx_0) = \vb_i - \inprod{\va^i}{\vx_0} \ge r\norm{\va^i}_2 \text{ for all } i \in [m],\]
and hence
\begin{equation}\label{eq:cbound}
\norm{\vc}_2 \le \sum_{i \in [m]} \frac{\norm{\va^i}_2}{\vs_i(\vx_0)} \le \frac m r.
\end{equation}
Further, Section 2.3 of \cite{NesterovN94} shows that the Dikin ellipse around $\vx$ (ball of radius $1$ in the $\mh(\vx)$ norm) is contained in $\calX$. Thus, for any $\vh \in \R^n$ with $\norm{\vh}_{\vx} \le 1$,
\[\norm{\vh}_2 \le \norm{\vx - \vx_0}_2 + \norm{\vx_0 - (\vx + \vh)}_2 \le 2R.\]
This implies the norm comparison $\norm{\cdot}_{\vx, *} \le 2R\norm{\cdot}_2$, and combining with \eqref{eq:cbound},
\[\norm{\vc}_{\vx, *} \le 2R\norm{\vc}_2 \le \frac{2mR}{r}.\]
Finally, applying our choice of $\eta$ in our guarantee \eqref{eq:central_path_c} gives
\[\norm{\nabla \psi(\vx)}_{\vx, *} \le \eta\norm{\vc}_{\vx, *} + \norm{\nabla F_\eta(\vx)}_{\vx, *} \le \frac 1 6. \]
This is now a sufficiently good initialization for Newton's method (cf.\ Theorems 2.2.3 and 2.2.5, \cite{NesterovN94}), which  reduces the Newton decrement from $\frac 1 6$ to $\Delta$ in $O(\log\log \frac 1 \Delta)$ additional iterations, each costing time $O(mn^{\omega-1})$.
\end{proof}

We next bound the runtimes of the remaining phases in Lemmas~\ref{lem:phase-2-iterations},~\ref{lem:phase-3-iterations}, and~\ref{lem:phase-4-iterations} respectively, when using Algorithm~\ref{alg:warm-start-LLT} as $\calO$ (with the runtime bounds in Lemma~\ref{lem:warm-llt-prox-guarantee}).

\begin{lemma}\label{lem:phase-2-iterations}
The runtime of Lines~\ref{line:phase_2_start} to~\ref{line:phase_2_end} of Algorithm~\ref{alg:annealing}, using $\WarmLLTProx(\cdot, \tau, t, \frac \delta {2T})$ (Algorithm~\ref{alg:warm-start-LLT}) as $\calO_t$ for all $t \ge 0$, is 
\begin{equation*}
    O\Par{\Par{m\tau + mn^{\omega - 1} + mn^{\frac 5 4}\log^2\frac {m\tau} \delta}\Par{\sqrt n \log \frac m \delta + \frac{\tau m}{\sqrt n}} \log \frac{m}{\delta}}.
\end{equation*}
\end{lemma}
\begin{proof}
    Let $\Brace{t_k}_{k \geq 1}$ correspond to the sequence of cooling parameters during Phase 2 of Algorithm~\ref{alg:annealing}, i.e., $t_1:= 10^6D^3$ and $t_k := t_1 \cdot (1+n^{-1/2})^{-(k - 1)}$. Recall that $\WarmLLTProx$ performs at most $33(1 + \frac \tau {t_k})\log \frac T \delta$ iterations when called with inputs $(\vx, \tau, t_k, \frac \delta T)$ (where we took $K$ as in Algorithm~\ref{alg:warm-start-LLT} and accounted for the ceiling). Thus, by the stopping rule in Line~\ref{line:phase_2_start} of Algorithm~\ref{alg:annealing}, the total number of iterations performed by $\WarmLLTProx$ during Phase 2 is at most
    \begin{equation*}
        33\sum_{k \in [T_2]} \Par{1 + \frac{\tau}{t_k}}\log \frac{T}{\delta} = 33\log \frac{T}{\delta}\cdot\Par{T_2 + \tau\sum_{k\in[T_2]} \frac{1}{t_k}},
    \end{equation*}
    where $T_2$ is the largest integer such that $t_{T_2} = t_1\cdot (1+n^{-1/2})^{-(T_2 - 1)} \geq \frac n m$. We first bound
    \begin{equation*}
        T_2 = \left\lceil\frac{\log\frac{mt_1}{n}}{\log\Par{1 + \frac 1 {\sqrt n}}}\right\rceil \leq 1+ 2\sqrt n \log \frac{mt_1}{n} = O\Par{\sqrt n \log\frac m \delta}.
    \end{equation*}
    Next, we have
\begin{align*}
\sum_{k\in[T_2]}\frac1{t_k}
=\sum_{k\in[T_2]}\frac{1}{t_1}\Par{1+\frac{1}{\sqrt n}}^{k - 1}
=\frac{1 + \sqrt n}{t_{T_2}}-\frac{\sqrt n}{t_1}
\le\Par{1 + \sqrt n}\frac{m}{n}
\le\frac{2m}{\sqrt n}.
\end{align*}
The total number of iterations is thus bounded by
\begin{equation}\label{eq:iterations_phase2}
    33\Par{1 + \sqrt n \log\frac{mt_1}{n} + \frac{2\tau m}{\sqrt n}}\log \frac{T}{\delta} = O\Par{\Par{\sqrt n \log\frac m \delta + \frac{\tau m}{\sqrt n}}\log\frac{m}{\delta}}.
\end{equation}
The conclusion follows by taking the cost of $\WarmLLTProx$ as a function of $K$ in Lemma~\ref{lem:warm-llt-prox-guarantee}.
\end{proof}

\begin{lemma}\label{lem:phase-3-iterations}
The runtime of Lines~\ref{line:phase_3_start} to~\ref{line:phase_3_end} of Algorithm~\ref{alg:annealing}, using $\WarmLLTProx(\cdot, \tau, t, \frac \delta {2T})$ (Algorithm~\ref{alg:warm-start-LLT}) as $\calO_t$ for all $t \ge 0$, is 
\begin{equation*}
O\Par{\Par{m\tau + mn^{\omega - 1} + mn^{\frac 5 4}\log^2\frac {m\tau} \delta}m\tau \log \frac{m}{\delta}}.
\end{equation*}
\end{lemma}
\begin{proof}
    We proceed similarly to the proof of Lemma \ref{lem:phase-2-iterations}. Let $\Brace{t_k}_{k \geq 1}$ be the sequence of cooling parameters during Phase 3. The total iteration count of $\WarmLLTProx$ in Phase 3 is at most
    \begin{equation*}
        33\log\frac T \delta\cdot \Par{T_3 + \tau\sum_{k\in[T_3]} \frac{1}{t_k}},
    \end{equation*}
    where $T_3$ is the largest integer such that $t_{T_3} \geq \frac 1 m$. We first show that the quantity $u_k := \sqrt{mt_k}$ decreases by a constant during each iteration of the while loop. By the update rule in Line~\ref{line:phase_3_update},
\begin{equation*}
    t_{k+1} = t_k \Par{\frac{1}{1+\frac 1 {\sqrt{mt_k}}}}\implies u_{k+1} = \frac{u_k}{\sqrt{1 + \frac 1 {u_k}}}.
\end{equation*}
Hence, for $k\in[T_3]$, we have
\begin{equation}\label{eq:uk_difference}
    u_k - u_{k+1} = \frac{1}{1 + \frac 1 {u_k} + \sqrt{1 + \frac 1 {u_k}}} \geq \frac{1}{4},
\end{equation}
where we use that $u_k \geq 1$.  This also shows that $T_3 \le 4\sqrt{n}$, because $u_1 \le \sqrt n$.
Thus, combining \eqref{eq:uk_difference} with the fact that $u_{T_3} \geq 1$, we have
\begin{equation*}
    u_k\ge1 + \frac{T_3-k}{4} \ge\frac{T_3 - k + 1}{4}\implies \sum_{k\in[T_3]}\frac1{t_k}
=\sum_{k\in[T_3]}\frac{m}{u_k^2}
\le16m\sum_{k\in[T_3]}\frac1{k^2}
\le32m.
\end{equation*}
Thus,
\begin{equation}\label{eq:iterations_phase3}
33\log\frac T \delta\cdot \Par{T_3 + \tau\sum_{k\in[T_3]} \frac{1}{t_k}}
= O\Par{m\tau\log\frac m \delta},
\end{equation}
gives the runtime claim analogously to Lemma~\ref{lem:phase-2-iterations}.
\end{proof}

\begin{lemma}\label{lem:phase-4-iterations}
The runtime of Lines~\ref{line:phase_4_start} to~\ref{line:phase_4_end} of Algorithm~\ref{alg:annealing}, using $\WarmLLTProx(\cdot, \tau, t, \frac \delta {2T})$ (Algorithm~\ref{alg:warm-start-LLT}) as $\calO_t$ for all $t \ge 0$, is 
\begin{equation*}
O\Par{\Par{m\tau + mn^{\omega - 1} + mn^{\frac 5 4}\log^2\frac {m\tau} \delta} m\tau\log \frac{m}{\delta}}.
\end{equation*}
\end{lemma}
\begin{proof}
This phase uses two calls to $\WarmLLTProx$, whose runtime cost is bounded by Lemma~\ref{lem:warm-llt-prox-guarantee} and the fact that the iteration count $K$ is always at most $O(m\tau\log \frac m \delta)$ for our choice of $T$.
\end{proof}

\begin{corollary}\label{cor:part-1-iteration-count}
The total runtime of Phases 2 to 4 (Lines~\ref{line:phase_2_start} to~\ref{line:phase_4_end}) of Algorithm~\ref{alg:annealing} is
\[O\Par{\Par{m\tau + mn^{\omega - 1} + mn^{\frac 5 4}\log^2\frac {m\tau} \delta}\Par{m\tau + \sqrt{n}\log \frac m \delta}\log\frac m \delta},\]
and for a universal constant $C' > 0$, the total number of calls to $\WarmLLTProx$ is at most
\[T \defeq C'\Par{m + \sqrt{n}\log \frac m \delta}.\]
\end{corollary}
\begin{proof}
The runtime follows by combining Lemmas~\ref{lem:phase-2-iterations},~\ref{lem:phase-3-iterations}, and~\ref{lem:phase-4-iterations}. The bound on $T$ follows from our bounds on $T_2$ and $T_3$ in Lemmas~\ref{lem:phase-2-iterations} and~\ref{lem:phase-3-iterations}, since Phase 4 uses two calls.
\end{proof}

\textbf{Correctness.} We conclude by giving a correctness guarantee of Algorithm~\ref{alg:annealing} as a high-accuracy sampler for $\pi$. We begin by giving a self-contained guarantee on Phase 1.

\begin{lemma}\label{lem:initialization}
Let $\vx$ be the output of Phase 1 of Algorithm~\ref{alg:annealing}. There exists $\mu \in \calP(\calX)$ such that
\[\TV\Par{\Law(\vx), \mu} \le \frac \delta 2,\quad \chi^2\Par{\mu \| \pi_{t_0}} \le 3.\]
\end{lemma}
\begin{proof}
%\hubert{should we consider moving this to an appendix?} 
Let $t \defeq t_0$, let $R \defeq 4\sqrt{\frac D t}$, and let $\mu$ be $\bmu \defeq \calN(\hvx, \frac 1 t \mh(\hvx)^{-1})$ conditioned on Line~\ref{line:bigrad}'s failure:
\[\mu \defeq \bmu(\cdot \mid \calE),\quad \calE \defeq \Brace{\vx: \norm{\vx - \hvx}_{\hvx} \le R}.\]
We first observe that $\calE \subseteq \calX$, i.e., every point in $\mu$'s support is feasible. This is because $R < 1$ so $\calE$ is a subset of the Dikin ellipse of radius $1$ around $\hvx$, contained in $\calX$ by Theorem 2.1.1, \cite{NesterovN94}.

We next prove that $\TV(\Law(\vx), \mu) \le \frac \delta 2$. It suffices to show that Line~\ref{line:bigrad} passes with probability $\le  \frac \delta 2$. To see this, we make an affine change of coordinates to $\vz(\vx; \hvx)$ as in \eqref{eq:whitened_coords}, and let $\bnu$, $\nu$, $\calE'$ denote the pushforwards of $\bmu$, $\mu$, $\calE$ under this coordinate change, so $\bnu = \calN(\0_n, \frac 1 t\id_n)$. Then standard $\chi^2$ concentration bounds (Lemma 1, \cite{LaurentM00}) show that $\bnu(\calE') \ge 1 - \frac \delta 2$ as required.

We finally show the $\chi^2$ bound. The pushforward measure of $\pi_t$ under $\vx \mapsto \vz(\vx; \hvx)$, denoted $\zeta_t$, is
\[\zeta_t \propto \exp(-U),\quad U(\vz) \defeq t\Par{\psi\Par{\hvx + \mh(\hvx)^{-\half}\vz} - \psi(\hvx)}.\]
An identical calculation to Lemma~\ref{lem:pushforward_reparam} with $\vy \gets \0_n$ then gives, for $f(u) = -\log(1 - u) - u$,
\[U(\vz) = t\inprod{\mh(\hvx)^{-\half}\nabla \psi(\hvx)}{\vz} + t\sum_{i \in [m]} f(\inprod{\vv_i(\hvx)}{\vz}).\]
A Taylor expansion shows that $|f(u) - \frac{u^2}{2}| \le \frac{|u|^3}{2}$ for $|u| \le \frac 1 4$,
so for all $\norm{\vz}_2 \le R < \frac 1 4$, 
\begin{equation}\label{eq:pot_compare}
\begin{aligned}
    \Abs{U(\vz) - \frac t 2 \norm{\vz}_2^2} &= \Abs{U(\vz) - \frac t 2 \sum_{i \in [m]} \inprod{\vv_i(\hvx)}{\vz}^2}  \\
    &\le t\norm{\nabla \psi(\hvx)}_{\hvx, *}\norm{\vz}_2 + \frac{t}{2}\sum_{i \in [m]} \Abs{\inprod{\vv_i(\hvx)}{\vz}}^3 \\
    &\le t\norm{\nabla \psi(\hvx)}_{\hvx, *}\norm{\vz}_2 + \frac{t\norm{\vz}_2}{2}\sum_{i \in [m]} \Abs{\inprod{\vv_i(\hvx)}{\vz}}^2 \le \frac 1 {10},
\end{aligned}
\end{equation}
where the last line used that each $|\inprod{\vv_i(\hvx)}{\vz}| \le \norm{\vz}_2$, and plugged in the guarantee of $\AC$ and our choice of $t$. Moreover, because $U$ is convex and $U(\0_n) = 0$, any $\vz \not\in \calE'$ satisfies
\begin{align*}U\Par{\frac{\vz}{2}} \ge \frac{\norm{\vz}_2}{R} \cdot U\Par{\frac{R}{2\norm{\vz}_2} \vz} \ge \frac t 2 \Par{\frac R 2}^2 - \frac 1 {10} = 2D - \frac 1 {10} \ge D,\end{align*}
and combining with $U(\vz) \ge 2U(\frac{\vz}{2})$ gives $U(\vz) \ge U(\frac{\vz}{2}) + D$ for $\vz \not\in \calE'$.
Therefore, letting $(\calE')^c$ be the complement of $\calE'$ on the support of $\zeta_t$,
\begin{align*}\int_{(\calE')^c} \exp(-U(\vz))\dd \vz &\le \exp\Par{-D}\int_{(\calE')^c} \exp\Par{-U\Par{\frac {\vz} 2}} \\
&\le \exp(-D) \cdot 2^n  \int \exp(-U(\vz)) \dd \vz \le \frac Z {10},\end{align*}
where we let $Z \defeq \int \exp(-U(\vz)) \dd \vz$ and used our choice of $D$. This implies
\[Z \le \frac{10}{9}\int_{\calE'} \exp(-U(\vz)) \dd \vz \le \frac{10\exp(\frac 1 {10})}{9}\int \exp\Par{-\frac t 2 \norm{\vz}_2^2} \dd \vz \le 2\Par{\frac{2\pi}{t}}^{\frac n 2}\]
because the bound \eqref{eq:pot_compare} holds on $\calE'$.
Finally, for all $\vz \in \calE'$,
\[\frac{\nu(\vz)}{\zeta_t(\vz)} = \frac{Z}{\bnu(\calE')\Par{\frac{2\pi} t}^{\frac n 2}}\exp\Par{U(\vz) - \frac t 2 \norm{\vz}_2^2} \le \frac{2\exp\Par{\frac 1 {10}}}{1 - \frac \delta 2} \le 3, \]
using our earlier bound on $\bnu(\calE')$. The same conclusion holds for $\frac{\mu}{\pi_t}$ after undoing the affine change of coordinates, and therefore $\chi^2(\mu \| \pi_t) = \E_{\mu}[\frac{\mu}{\pi_t}] - 1 \le 3$ as claimed.
\end{proof}

We next analyze the $\chi^2$ warmness of successive distributions in the cooling path, as $t$ updates throughout Phases 2 through 4 of Algorithm~\ref{alg:annealing}, leveraging the following helper lemma.

\begin{lemma}\label{lem:cooling-helper}
    Let $t>t' > 0$. Following the notation \eqref{eq:cooling-path}, let $F(s)\defeq \log Z(s)$. Then, we have
    \begin{equation*}
        \log\Par{1+\chi^2(\pi_t\|\pi_{t'})} \leq (t - t')^2\sup_{s\in [t', 2t - t']} F''(s).
    \end{equation*}
\end{lemma}
\begin{proof}
    Letting $h \defeq t - t'$, we compute
\begin{align*}
    \mathbb{E}_{\pi_t}\Brack{\frac{\dd\pi_t}{\dd\pi_{t'}}}
    &=\int \frac{\exp(-t\psi(\vx))/Z(t)}{\exp(-t'\psi(\vx))/Z(t')}\cdot \frac{\exp(-t\psi(\vx))}{Z(t)} \pi(\vx)\dd \vx\\
    &=\frac{Z(t')}{Z(t)^2}\int \exp(-(2t-t')\psi(\vx))\pi(\vx)\dd\vx\\ 
    &=\frac{Z(2t-t')Z(t')}{Z(t)^2} =\frac{Z(t + h)Z(t - h)}{Z(t)^2}.
\end{align*}
The left-hand side is $1 + \chi^2(\pi_t \| \pi'_t)$, so the claim follows from Taylor's theorem:
\begin{equation}\label{eq:log_chi_squared_helper}
    \log\Par{1+\chi^2(\pi_t\|\pi_{t'})} = F(t + h) + F(t - h) - 2F(t) \leq h^2 \sup_{s\in [t - h,t + h]} F''(s).
\end{equation}
\end{proof}

As consequences of Lemma~\ref{lem:cooling-helper}, we derive $\chi^2$ warmness bounds for each of Phases 2 to 4.

\begin{lemma}\label{lem:phase2-warmstart}
    Let $t > 0$ and $t' := t\cdot(1 + \frac 1 {\sqrt n})^{-1}$. Following the notation \eqref{eq:cooling-path}, we have
    \begin{equation*}
        \chi^2(\pi_t\Vert \pi_{t'}) \leq e-1.
    \end{equation*}
\end{lemma}
\begin{proof}
    Since $\psi$ is convex, $s\mapsto s^n Z(s)$ is log-concave by Lemma 3.2, \cite{Kalai2006simulated}. Equivalently, $F(s) + n\log s$ is concave, so
    % \begin{equation*}
        $F''(s) \leq \frac{n}{s^2}$.
    % \end{equation*}
    Since $t - t' = \frac{t'}{\sqrt n}$, the claim follows from Lemma~\ref{lem:cooling-helper}:
    \begin{equation*}
        \log\Par{1+\chi^2(\pi_t\|\pi_{t'})} \leq \frac{(t')^2}{n} \cdot \frac{n}{(t')^2} = 1.
    \end{equation*}
\end{proof}
\begin{lemma}\label{lem:phase3-warmstart}
    Let $t > 0$ and $t' := t\cdot(1+\frac 1 {\sqrt{mt}})^{-1}$. Following the notation \eqref{eq:cooling-path}, we have
    \begin{equation*}
        \chi^2(\pi_t\Vert \pi_{t'}) \leq e - 1.
    \end{equation*}
\end{lemma}
\begin{proof}
For all $s \ge 0$, we have
\begin{equation}\label{eq:f''_le_ms}
    F''(s) = \Var_{\pi_s}[\psi] \leq \frac{1}{s}\E_{\pi_s}\Brack{\nabla \psi^\top (\nabla^2\psi)^{-1} \nabla\psi} \leq \frac{m}{s},
\end{equation}
where the first equality is a standard identity on the log-normalizing constants of exponential families (and follows by direct differentiation), the first inequality uses Fact~\ref{fact:bl} with $-\log \pi_s \succeq s\psi$, and the second inequality uses Fact~\ref{fact:sc_log}. Thus, since $t - t' = \frac{t'}{\sqrt{mt}}$, Lemma~\ref{lem:cooling-helper} gives the claim:
\begin{equation*}
    \log\Par{1+\chi^2(\pi_t\|\pi_{t'})} \leq \frac{(t')^2}{mt} \cdot \frac{m}{t'} = \frac{t'}{t}\leq 1.
\end{equation*}
\end{proof}

\begin{lemma}\label{lem:phase4-warmstart}
Let $0 < t \leq \frac{1}{m}$. Following the notation \eqref{eq:cooling-path}, we have
\begin{equation*}
    \chi^2(\pi_t\Vert \pi) \leq 3.
\end{equation*}
\end{lemma}
\begin{proof}
Recalling that $F(0) = \log 1 = 0$, \eqref{eq:log_chi_squared_helper} implies
\[H(s) \defeq \log\Par{1 + \chi^2(\pi_s \| \pi)} = F(2s) - 2F(s)\]
for all $s > 0$. Our goal is to prove that $H(t) \le mt\log 4 \le \log 4$, which proves $\chi^2(\pi_t \| \pi) \le 3$.

To prove this, we bound $H'(s)$ for $s \in [0, t]$. We have
\begin{equation*}
    H'(s) = 2(F'(2s) - F'(s))=2\int_s^{2s} F''(r)\dd r \leq 2m \log 2,
\end{equation*}
where the inequality uses \eqref{eq:f''_le_ms}. Moreover, $H$ is continuous at 0 and $H(0) = 0$, so the claim follows:
\begin{equation*}
    H(s) \leq 2ms \log 2 \implies H(t) \leq \log 4.
\end{equation*}
\end{proof}

We conclude with our correctness guarantee on Algorithm~\ref{alg:annealing}.

\begin{lemma}\label{lem:annealing_correct}
The output $\vx$ of Algorithm~\ref{alg:annealing} satisfies $
        \TV(\Law(\vx), \pi) \leq \delta$.
\end{lemma}
\begin{proof}
Let $\Brace{\vx_i}_{i\geq 0}$ and $\Brace{t_i}_{i\geq 1}$ be the sequences of iterates and the sequences of cooling parameters produced by Algorithm~\ref{alg:annealing}. Namely, we let $\vx_0$ be the output of Phase 1 and $t_1:=10^6 D^3$. Let $\sfP_t$ denote the Markov kernel induced by $\calO_t$. Formally,
    \begin{equation*}
        \sfP_t(\vx, \cdot) \defeq \Law\Par{\calO_t(\vx)}.
    \end{equation*}
    Note that $\calO_t$ is assumed to be a $(3,\frac{\delta}{2T})$-$\chi^2$-to-TV oracle for $\pi_t$ for all $t\geq 0$. In other words,
    \begin{equation}\label{eq:chi_to_tv_helper}
        \TV(\nu \sfP_t, \pi_t) \leq \frac{\delta}{2T} \quad \text{for all } \nu \text{ satisfying} \;\; \chi^2(\nu\Vert \pi_t) \leq 3.     
    \end{equation}
    By Lemma~\ref{lem:initialization}, there exists a measure $\mu$ such that $\TV(\Law(\vx_0), \mu) \leq \frac{\delta}{2}$ and $\chi^2(\mu\Vert \pi_{t_1}) \leq 3$. Thus,
        \begin{align*}
        \TV(\Law(\vx_1), \pi_{t_1})
        &= \TV(\Law(\vx_0)\sfP_{t_1}, \pi_{t_1}) \\
        &\leq \TV(\Law(\vx_0)\sfP_{t_1}, \mu \sfP_{t_1}) + \TV(\mu \sfP_{t_1}, \pi_{t_1})\\
        &\leq \TV(\Law(\vx_0), \mu) + \frac{\delta}{2T}\leq \frac{\delta}{2} + \frac{\delta}{2T}.
        \end{align*}
    where the last line uses \eqref{eq:chi_to_tv_helper}. For $i\geq 2$, we may similarly argue
    \begin{align*}
        \TV(\Law(\vx_i), \pi_{t_i})
        &= \TV(\Law(\vx_{i - 1})\sfP_{t_i}, \pi_{t_i})\\
        &\leq \TV(\Law(\vx_{i - 1})\sfP_{t_i}, \pi_{t_{i - 1}} \sfP_{t_i}) + \TV(\pi_{t_{i - 1}} \sfP_{t_i}, \pi_{t_i})\\
        &\leq \TV(\Law(\vx_{i-1}), \pi_{t_{i - 1}}) + \frac{\delta}{2T},
    \end{align*}
    where the last line uses Lemmas~\ref{lem:phase2-warmstart}, \ref{lem:phase3-warmstart}, \ref{lem:phase4-warmstart}, and \eqref{eq:chi_to_tv_helper}. Let $T' \le T$ denote the total number of calls made to $\calO$ (as guaranteed by Corollary~\ref{cor:part-1-iteration-count}). Then, iterating the above display gives that $\TV(\Law(\vx_{T'}), \pi) = \TV(\Law(\vx_{T'}), \pi_{t_{T'}}) \leq \frac{\delta}{2} + \frac{\delta}{2} = \delta$. 
\end{proof}

\subsection{Proof of Theorem~\ref{thm:main_part1}}

We conclude the section by proving Theorem~\ref{thm:main_part1}, the main result of Part~\ref{part:sqrtn}.

\restatepartI*
\begin{proof}
We first check that the choice of $\tau$ is compatible with the requirements of our $\chi^2$-to-TV oracle in Lemma~\ref{lem:warm-llt-prox-guarantee}. Indeed, Lemma~\ref{lem:warm-llt-prox-guarantee} requires for a constant $C$,
\[\tau \ge C\Par{\sqrt{n\log \frac {mKT}\delta} + \log \frac {mKT} \delta},\text{ where } K \defeq \left\lceil 32\min\Brace{m\tau, 1 + \frac \tau t}\log \frac 1 \delta\right\rceil,\]
and $T = O(m + \sqrt{n}\log \frac m \delta)$ (Corollary~\ref{cor:part-1-iteration-count}).
This induces a circular dependence between $\tau$ and $K$, but because $\tau$ scales logarithmically in $K$, by increasing $C$ by a constant factor, this requirement is met. 
The runtime now follows from Proposition~\ref{prop:ac} and Corollary~\ref{cor:part-1-iteration-count} (where we use that $n \le m$ so $\sqrt{m}n \le m\sqrt{n}$), and the correctness follows from Lemma~\ref{lem:annealing_correct}.
\end{proof}

\newpage

\part{Polytope Sampling in $\widetilde{O}(m)$ Iterations}\label{part:tilde1}
\section{Backwards Oracle Implementation via Dual Conditioning}\label{sec:part_2_backwards_step}

In this section, we provide a different approach to implementing a $(3, \delta)$-$\chi^2$-to-TV oracle (i.e., a replacement for Lemma~\ref{lem:warm-llt-prox-guarantee}), that succeeds at a lower step size $\tau = O(\log \frac m \delta)$. Our strategy gives up on proving a uniform bound over all backwards conditional distributions. Instead we show that a simple implementation of the backwards step, a na\"ive Metropolized random walk in dual-whitened coordinates, satisfies a weak $\chi^2$ contraction with a defect term proportional to the size of a certain ``dual-unstable'' region. The size of this region varies per backwards conditional distribution, but it is small on average and this suffices for the final TV guarantee.

Section~\ref{ssec:part_2_setup} sets up notation and states our main concentration fact (Lemma~\ref{lem:concentration_part2}), used to control the average defect. Section~\ref{ssec:mrw} proves an instancewise convergence guarantee on a whitened Metropolized random walk (Algorithm~\ref{alg:backwards_part2}), which is plugged into an approximate implementation of Algorithm~\ref{alg:prox-sample} in Section~\ref{ssec:comparison} as our oracle. Finally, Section~\ref{ssec:main_proof} proves Theorem~\ref{thm:main_part2}.

\subsection{Setup}\label{ssec:part_2_setup}

We begin by setting up notation specific to this section. Recall that Part~\ref{part:sqrtn} treated the cooled joint measures $\tpi_{t, \tau}$ in \eqref{eq:cooling-path}, \eqref{eq:cooling-path-joint} as being over $\calX \times \calY$, where the dual cone $\calY$ is defined in \eqref{eq:dualspace}.

In this section, we instead use $\vgamma \in \R^m_{ > 0}$ directly as the dual variable, which induces a $\vy \in \calY$ through $\vy = \ma^\top \vgamma$; the corresponding random variable is denoted $\Gamma$. This is consistent with how $\vy$ is generated in Proposition~\ref{prop:forward}. 
Our convention induces a  set of backwards conditional measures with respect to $\vgamma\in\R^m_{>0}$ instead of $\vy\in\calY$. We thus define, as measures over $\calX$,
\begin{equation}\label{eq:pi_gamma}
\pi_{t,\tau}^{X|\Gamma=\vgamma} (\vx)\propto \exp\Par{\inprod{\ma^\top \vgamma}{\vx} -\tau\psi(\vx)}\pi_t(\vx) \propto \exp\Par{\inprod{\ma^\top \vgamma}{\vx} - (\tau + t)\psi(\vx)},
\end{equation}
following \eqref{eq:cooling-path}. We denote the corresponding $\Gamma$-conditional distributions by $\pi_{t, \tau}^{\Gamma \mid X = \vx}$, which admit exact samplers via Proposition~\ref{prop:forward}. The total $\Gamma$-marginal law is denoted $\pi_{t,\tau}^\Gamma \defeq \E_{\vx \sim \pi_t}[\pi_{t, \tau}^{\Gamma \mid X = \vx}]$.

Next, following Section~\ref{ssec:logbarrier}, we define
\[\mh_{\vgamma} \defeq \ma^\top \diag{\vgamma}^2\ma\]
analogously to $\mh(\vx) \defeq \ma^\top \ms(\vx)^{-2}\ma$. Intuitively, $\mh_{\vgamma}$ simulates a Hessian matrix at a point with inverse slacks $\vgamma$. We also define normalized relative slacks with respect to $\vgamma$,
\begin{equation}\label{eq:gamma_slacks}\vr_\tau(\vx; \vgamma) \defeq \1_m - \frac{\vgamma \circ \vs(\vx)}{\tau},\end{equation}
where $\circ$ denotes entrywise multiplication. Up to scaling, this is consistent with \eqref{eq:relative_slacks} when $\vgamma$ are inverse slacks of a reference point. Finally, we define at each $\vgamma \in \R^m_{> 0}$ a conditional set
\begin{equation}\label{eq:gamma_box}\calB_\tau(\vgamma, \rho) \defeq  \Brace{\vx\in\calX  : \norm{\vr_\tau(\vx;\vgamma)}_\infty\le \rho}, \text{ for all } \rho\in (0,1).\end{equation}
The normalization used in our definitions \eqref{eq:gamma_slacks}, \eqref{eq:gamma_box} is justified through the following concentration fact (Lemma~\ref{lem:concentration_part2}), which shows that a draw from a backwards conditional distribution induced by $\vgamma$ falls in $\calB_\tau(\vgamma, \frac 1 4)$  with high probability, \emph{on average} over $\vgamma \sim \pi_{t,\tau}^{\Gamma}$.

\begin{lemma}\label{lem:concentration_part2}
For all $t \ge 0$ and $\tau \in \N$,
\[
\E_{\vgamma\sim \pi_{t,\tau}^{\Gamma}}\Brack{\pi_{t,\tau}^{X|\Gamma=\vgamma}\Par{\calB_\tau\Par{\vgamma,\frac14}}} \ge 1 -  2m\exp\Par{-\frac{\tau}{64}}.
\]
\end{lemma}
\begin{proof}
By the law of total expectation,
\[    \E_{\vgamma\sim\pi_{t,\tau}^\Gamma}\Brack{1 -\pi_{t,\tau}^{X|\Gamma = \vgamma}\Par{\calB_\tau\Par{\vgamma,\frac14}}} = \Pr\Brack{\max_{i\in[m]}\Abs{1 - \frac{\vgamma_i\vs_i(\vx)}{\tau}}>\frac14},
    \]
where the probability is taken with respect to the joint law of $(X, \Gamma)$.
Moreover, conditioned on $X = \vx$, Proposition~\ref{prop:forward} shows that $\Gamma \sim \pi_{t, \tau}$ induces $S_i \defeq \Gamma_i \vs_i(\vx)$, such that each $S_i$ is the sum of $\tau$ independent draws from $\Exp(1)$. Recall that for $E \sim \Exp(1)$,
\[\log \E\Brack{\exp\Par{\lam(E - 1)}} = -\lam - \log(1 - \lam) \le \lam^2, \text{ for all } |\lam| \le \half. \]
Therefore, taking $\lam = \frac 1 8$ and convolving $\tau$ i.i.d.\ draws $E$,
\[\E\Brack{\exp\Par{\frac{S_i - \tau}{8}} \mid X = \vx} \le \exp\Par{\frac \tau {64}},\]
so that Markov's inequality implies
\[\Pr\Brack{S_i - \tau > \frac \tau 4 \mid X = \vx} \le \exp\Par{-\frac \tau {64}}.\]
A similar bound holds for the lower tail. Because this holds for every $\vx$, taking expectations over $\vx \sim \pi_t$ and then taking a union bound over all $m$ facets gives the claim.
\end{proof}

\subsection{Metropolized random walk}\label{ssec:mrw}
% \textcolor{red}{Typos} 

In this section, we give an approximate backwards oracle implementation (i.e., sampler for~\eqref{eq:pi_gamma}), with an instancewise guarantee depending on the forwards draw $\vgamma$. Our sampler yields a weak form of $\chi^2$ contraction, with a defect term whose size is governed by the backwards measure of $\calB_\tau(\vgamma, \frac 1 4)$. To compare this defective contraction to our ideal outer loop iteration (the proximal sampler of Proposition~\ref{prop:prox-sample-guarantee}), we define for any bounded function $f: \Omega \to \R$,
\[
\osc(f) \defeq 
\esssup f - \essinf f.
\]

We next state our approximate backwards oracle implementation, a simple Metropolized random walk in whitened coordinates with respect to $\mh_{\vgamma}$, in Algorithm~\ref{alg:backwards_part2}. Observe that the algorithm is $\half$-lazy, meaning that its Markov kernel from any $\vx \in \calX$ places at least half its mass on $\vx$.

\begin{algorithm2e}[t]
\caption{$\mathsf{BackwardsMRW}(\vx_0, \vgamma, \tau, t, N)$}
\label{alg:backwards_part2}
\DontPrintSemicolon
\textbf{Input: }previous iterate $\vx_0\in \calX$, forward draw $\vgamma\in\R^m_{>0}$, regularization parameter $\tau\in \N$, cooling parameter $t\in \R_{\ge0}$, number of iterations $N \in \N$\;
\For{$i=0,\dots, N-1$}{\label{line:loop_start}
$W \sim  \Unif([2])$\;
$\vx_{i + 1} \gets \vx_i$\;
\If{$W = 2$}{
$\vx' \gets  \vx_i + \frac{\tau}{8\sqrt{n(\tau+t)}}\mh_{\vgamma}^{-1/2}\vxi$ for a fresh random draw $\vxi\sim \calN(\0_n,\mi_n)$\;\label{line:proposal}
$\vx_{i + 1} \gets \vx'$ with probability 
\[
\min\Brace{1, \exp\Par{\inprod{\ma^\top\vgamma}{\vx'-\vx_i}-(\tau+t)(\psi(\vx')-\psi(\vx_i))}} \ind_{\vx' \in \calX}
\]
}
}\label{line:loop_end}
\textbf{Return: }$\vx_N$
\end{algorithm2e}

To prove our instancewise convergence guarantee on Algorithm~\ref{alg:backwards_part2}, we use three helper results. 

The first result says that two Metropolized chains starting from nearby points have overlapping one-step distributions, provided their proposals are close in TV distance and their acceptance probabilities are sufficiently large. Formally, for a \emph{proposal kernel} $\sfP$ and acceptance probability function $a(\vz, \vz')$, we define the \emph{Metropolized kernel} $\sfM$ as follows: for a measurable set $B$,
\begin{equation}\label{eq:metropolis}
\sfM(\vz, B) \defeq \int_B a(\vz, \vz')\sfP(\vz, \dd \vz') + (1 - A(\vz))\ind_{\Brace{\vz\in B}},
\quad A(\vz) \defeq \E_{\vz' \sim \sfP(\vz, \cdot)}\Brack{a(\vz, \vz')}.
\end{equation}
The Metropolized kernel $\sfM(\vz, \cdot)$ places all rejected probability mass on $\vz$, and $A(\vz)$ is the total acceptance probability for a fixed $\vz$. 

We focus on the Metropolis-Hastings rule for symmetric $\sfP$, where for a target measure $\mu$,
\begin{equation}\label{eq:mh_rule}
a(\vz, \vz') \defeq \min\Brace{1, \frac{\mu(\vz')}{\mu(\vz)}}.
\end{equation}
It is straightforward that under \eqref{eq:mh_rule}, $\sfM$ is reversible with respect to $\mu$. Algorithm~\ref{alg:backwards_part2} is the lazy analog of the random walk \eqref{eq:metropolis}, \eqref{eq:mh_rule} with stationary distribution $\mu \gets \pi_{t, \tau}^{X \mid \Gamma = \vgamma}$. 

\begin{proposition}[Lemma 19, \cite{AndrieuLPW24}]\label{prop:alpw24}
    Let $\sfM: \calZ \times \calB(\calZ) \to [0, 1]$ be a Metropolized kernel \eqref{eq:metropolis}, \eqref{eq:mh_rule} with respect to a symmetric proposal kernel $\sfP: \calZ \times \calB(\calZ) \to [0, 1]$ and target measure $\mu \in \calP(\calZ)$, and define the total acceptance probabilities $A: \calZ \to [0, 1]$ as in \eqref{eq:metropolis}. Then for all $(\vz, \vw) \in \calZ \times \calZ$,\footnote{Lemma 19 in \cite{AndrieuLPW24} replaces $\min\{A(\vz), A(\vw)\}$ with $\inf_{\vz \in \calZ} A(\vz)$, but our statement is implicit in the proof.} \[
    \TV(\sfM(\vz,\cdot), \sfM(\vw,\cdot))\le \TV(\sfP(\vz, \cdot),\sfP(\vw, \cdot)) + 1-\min \Brace{A(\vz),A(\vw)}.
    \]
\end{proposition}

The second result states that TV overlap lower bounds the conductance of a Markov kernel, under a strongly log-concave target measure.

\begin{proposition}[Lemma 6, \cite{DwivediCWY19}]\label{prop:dcwy19}
Let $\Omega$ be a Borel subset of $\R^n$, and let $\sfK$ be a Markov kernel on $\Omega$ reversible with respect to $\mu \in \calP(\Omega)$. Let $\calK \subseteq \Omega$ be convex and Borel with $\mu(\calK) > 0$, and suppose that $\mu(\cdot\mid \calK)$ is $\alpha$-strongly log-concave. If for some $(\eta, \Delta) \in (0, 1) \times \R_{>0}$,  we have $\TV(\sfK(\vz,\cdot),\sfK(\vw,\cdot))\le 1-\eta$ for all $\vz,\vw\in \calK$ with $\norm{\vz-\vw}_2\le \Delta$, then for any Borel $B \subseteq \Omega$, \[
    \int_B \sfK(\vz,\Omega \setminus B)\mu(\dd \vz) \ge \frac{\eta}{4}\min\Brace{1,\frac{\log 2\cdot  \Delta \mu(\calK)^2 \sqrt{\alpha}}{8}}\min\Brace{\mu(B\cap \calK),\mu((\Omega \setminus B)\cap \calK)}.
    \]
\end{proposition}

Our third helper result is based on the recent weak Poincar\'e inequality framework of \cite{AndrieuLPW26}. In particular, we put together several ingredients from \cite{AndrieuLPW26} to prove a self-contained result on $\chi^2$ convergence with a defect term. Consistently with \cite{AndrieuLPW26}, we denote the Dirichlet form of a Markov kernel $\sfK : \Omega \times \calB(\Omega)\rightarrow [0, 1]$ that is reversible with respect to $\mu \in \calP(\Omega)$ by
\[\calE_{\sfK}(f, f) \defeq \half \iint_{\Omega \times \Omega} (f(\vx) - f(\vy))^2 \sfK(\vx, \dd\vy)\mu(\dd \vx) \text{ for all } f: \Omega \to \R.\]

\begin{lemma}\label{lem:consequence_of_alpw26}
    Let $\sfK$ be a $\frac{1}{2}$-lazy Markov kernel on $\Omega$ that is reversible with respect to $\mu \in \calP(\Omega)$, let $b\in(0,1]$, and let $\delta\in (0,\frac12)$. If for every measurable $B \subseteq \Omega$, \begin{equation}\label{eq:weak_conductance_profile}
    \int_B\sfK(\vz,\Omega \setminus B)\mu(\dd\vz)\ge b\max\Brace{\min\Brace{\mu(B),\mu(\Omega \setminus B)}-\delta ,0},
    \end{equation}    
    then for every bounded $f: \Omega \to \R$ and  $N\in\N$, \[
    \Var_\mu[\sfK^N f]\le \exp\Par{-\frac{b^2N}{2}}\Var_\mu[f] + 2\delta \osc(f)^2.
    \]
\end{lemma}
\begin{proof}
Let $\Pi(\vx,\cdot) \defeq \mu(\cdot)$, and for all $\eps \in (0, 1)$, set $\sfK_\eps \defeq (1-\eps)\sfK + \eps\Pi$. Following Definition 34, \cite{AndrieuLPW26}, we define the \emph{weak conductance profile}
\[\Phi_{\sfK_{\eps}}(s) \defeq \inf_{B \subseteq \Omega: s \le \mu(B) \le \half} \frac{\int_B \sfK_{\eps}(\vz, \Omega \setminus B)\mu(\dd \vz)}{\mu(B)} \text{ for all } 0 < s \le \half.\]
Then by using~\eqref{eq:weak_conductance_profile} and $\Pi(\vz, \Omega \setminus B) = 1 - \mu(B)$, it follows that
\[
\Phi_{\sfK_\eps}(s)\ge (1-\eps)b\max\Brace{1-\frac{\delta}{s},0} + \frac{\eps}{2},\quad 0 < s \le \frac12.
\]
By Lemma 40(a), \cite{AndrieuLPW26}, the $L^1$-weak Poincar\'e coefficient satisfies \[
\alpha_{1,\eps}(r) \defeq \sup_{r \leq s \leq \frac{1}{2}}\frac{s-r}{s\Phi_{\sfK_\eps}(s)}\le \frac{1}{(1-\eps)b}.
\]
Applying Lemma 42, \cite{AndrieuLPW26} with $r\gets 2\delta$ then yields \begin{align*}
    \Var_\mu[f]\le 2\alpha_{1,\eps}(\delta)^2\calE_{\sfK_\eps}(f,f) + 2\delta \osc(f)^2 \le \frac{2}{(1-\eps)^2b^2}\calE_{\sfK_\eps}(f,f) + 2\delta \osc(f)^2.
\end{align*}
Letting $\eps\to 0$ implies 
\begin{equation}\label{eq:wpi}
   \Var_\mu[f]\le \frac{2}{b^2}\calE_{\sfK}(f,f) + 2\delta \osc_\mu(f)^2.      
    \end{equation}
  
Next, by lazyness, suppose $\sfK$ is of the form $\sfK = \frac{\sfI + \widetilde \sfK}2$ for the lazy kernel $\sfI$ and some other kernel $\widetilde \sfK$ with stationary measure $\mu$. Set $g \defeq f - \mu(f)$ and $\norm{h}_2^2 \defeq \int h^2 \dd \mu$ for a square-integrable function $h: \Omega \to \R$. By invariance of $\widetilde \sfK = 2\sfK - \sfI$, Jensen's inequality yields $\norms{\widetilde \sfK g}_2^2\le \norms{g}_2^2$, so 
\[
\inprod{g}{\sfK g} - \norm{\sfK g}_2^2 = \frac{1}{4}\Par{\norm{g}_2^2 - \norm{\widetilde \sfK g}_2^2}\ge 0.
\]
Thus, \[
 \Var_\mu[f] - \Var_\mu[\sfK f] = \norm{g}_2^2 - \norm{\sfK g}_2^2 \ge \inprod{g}{(\sfI - \sfK)g} = \calE_{\sfK}(f,f).
\]

Rearranging~\eqref{eq:wpi} and substituting the above display yields \[
\Var_\mu[\sfK f] \le \Par{1-\frac{b^2}{2}}\Var_\mu[f] + b^2\delta \osc(f)^2.
\]
Iterating and using $\osc(\sfK f) \le \osc(f)$, and $1-x\le e^{-x}$ for $x \in (0,1)$, gives the claim: \begin{align*}
    \Var_\mu[\sfK^Nf] &\le \Par{1-\frac{b^2}{2}}^N\Var_{\mu}[f] + b^2\delta \sum_{j=0}^{N-1}\Par{1-\frac{b^2}{2}}^j \osc(f)^2\\
    &\le \exp\Par{-\frac{b^2N}{2}}\Var_\mu[f] + 2\delta \osc(f)^2.
\end{align*}
\end{proof}

We can now combine Propositions~\ref{prop:alpw24} and~\ref{prop:dcwy19} with Lemma~\ref{lem:consequence_of_alpw26} to prove a defective form of $\chi^2$ convergence for Algorithm~\ref{alg:backwards_part2}, with weak Poincar\'e constant $\Omega(\frac 1 n)$.

\begin{lemma}\label{lem:mrw}
Let $\tau, N \in \N$, $t \ge 0$, $\vgamma \in \R^m_{> 0}$, and let $\sfM$ denote the Markov kernel corresponding to Algorithm~\ref{alg:backwards_part2} with inputs $(\vgamma, \tau, t, N)$. There is a universal constant $C > 0$ such that, for every bounded $h: \calX \to \R$,
\[\Var_{\pi_{t, \tau}^{X \mid \Gamma = \vgamma}}\Brack{\sfM h} \le \exp\Par{-\frac{N}{Cn}}\Var_{\pi_{t, \tau}^{X \mid \Gamma = \vgamma}}[h] + 2\pi_{t, \tau}^{X \mid \Gamma = \vgamma}\Par{\calX \setminus \calB_\tau\Par{\vgamma, \frac 1 4}}\osc(h)^2.\]
\end{lemma}
 
\begin{proof}
For convenience, throughout this proof we use the shorthand notation
\[G \defeq \calB\Par{\vgamma, \frac 1 4},\quad H \defeq \calB\Par{\vgamma, \half},\quad \mu \defeq \pi_{t, \tau}^{X \mid \Gamma = \vgamma},\quad \delta \defeq \mu\Par{\calX \setminus G}.\]
Thus the claim is $\Var_\mu[\sfM h] \le \exp(-\frac{N}{Cn})\Var_\mu[h] + 2\delta \osc(h)^2$. If $\delta\ge\frac12$, then the claim follows from $\Var_\mu[\sfM h] \le \Var_\mu[h]\le \frac{\osc(h)^2}4\le 2\delta\osc(h)^2$. Henceforth assume $\delta<\frac12$. 

Next, we show that the potential of $\mu$, $F(\vx) \defeq (\tau + t)\psi(\vx) - \inprod{\ma^\top \vgamma}{\vx}$, is well-conditioned on a large region. For all $\vx \in H$, we have $\vs(\vx)^{-1} \in [\frac 2 {3\tau}\vgamma, \frac 2 \tau \vgamma]$ entrywise, and thus
\[\frac{4(\tau + t)}{9\tau^2} \mh_{\vgamma} \preceq \nabla^2 F(\vx) = (\tau + t)\ma^\top \ms(\vx)^{-2} \ma \preceq \frac{4(\tau + t)}{\tau^2} \mh_{\vgamma}.\]
Under the invertible change of variables
\begin{equation}\label{eq:whiten_part2}\vz \defeq \frac{\sqrt{\tau + t}}{\tau} \mh_{\vgamma}^{\half}\vx,\quad U(\vz) \defeq F\Par{\frac \tau {\sqrt{\tau + t}}\mh_{\vgamma}^{-\half}\vz},\end{equation}
and using tildes to denote pushforwards under this map (e.g., $\widetilde{H}$ is the set of $\vz$ that are the image of some $\vx \in H$), we have therefore shown
\begin{equation}
\frac14\mi_n\preceq \nabla^2U(\vz)\preceq 4\mi_n \text{ for all }\vz\in \widetilde H.\label{eq:smoothness_whitened}
\end{equation}
We next lower bound the acceptance probability $A(\vz)$ from a current point $\vz\in \widetilde G$ (corresponding to the image of an iterate $\vx = \vx_i$ of Algorithm~\ref{alg:backwards_part2} under the mapping \eqref{eq:whiten_part2}). In the whitened coordinates, the symmetric proposal $\widetilde{\sfP}$ implemented by Line~\ref{line:proposal} is simply $\widetilde{\sfP}(\vz, \cdot) = \calN(\vz, \sigma^2 \id_n)$, where $\sigma \defeq (8\sqrt{n})^{-1}$. Define an event on $\vxi \sim \calN(\0_n, \id_n)$:
\[\calE(\vz) \defeq \Brace{\vxi \in \R^n: \norm{\vxi}_2^2 \le 2n \text{ and } \inprod{\nabla U(\vz)}{\vxi} \le 0}.\]
Then by Markov's inequality on $\E[\norm{\vxi}_2^2] = n$ and symmetry of Gaussians, we have $\Pr[\calE(\vz)] \ge \frac 1 4$. For any $\vz \in \widetilde{G}$, on the event $\calE(\vz)$, the proposal $\vz + \sigma \vxi$ lies in $\widetilde H$, because
\[\norm{\sigma \vxi}_2 \le \frac{\sqrt{2n}}{8\sqrt{n}} = \frac 1 {\sqrt{32}} \le \frac 1 4,\]
and for any $i \in [m]$ and $\vx, \vx' \in \calX$ with images $\vz$, $\vz'$ under the map \eqref{eq:whiten_part2},
\[\Abs{\Brack{\vr_\tau(\vx; \vgamma) - \vr_{\tau}(\vx'; \vgamma)}_i} = \frac{\Abs{\inprod{\vgamma_i \mh_{\vgamma}^{-\half}\va^i}{\vz' - \vz}}}{\sqrt{\tau + t}} \le \frac{\norm{\vz' - \vz}_2}{\sqrt{\tau + t}} \le \norm{\vz' - \vz}_2,\]
because $\mh_{\vgamma}\succeq \vgamma_i^2 \va^i(\va^i)^\top$.
Now a second-order Taylor expansion with \eqref{eq:smoothness_whitened} shows that for $\vxi \in \calE(\vz)$,
\[U\Par{\vz + \sigma \vxi} - U(\vz) \le \sigma\inprod{\nabla U(\vz)}{\vxi} + 2 \norm{\sigma\vxi}_2^2 \le \frac 1 {16}.\]
Consequently, for any $\vz \in \widetilde G$, we have $A(\vz) \ge \frac 1 4 \exp(-\frac 1 {16}) \ge \frac 1 8$. Moreover, Pinsker's inequality and the KL divergence between identity-covariance Gaussians with different means show that
\[\TV\Par{\widetilde{\sfP}\Par{\vz, \cdot}, \widetilde{\sfP}\Par{\vw, \cdot}} \le \frac{\norm{\vz - \vw}_2}{2\sigma}.\]
For $\Delta \defeq \frac 1 {64\sqrt{n}}$ and $\norm{\vz - \vw}_2 \le \Delta$, combining the preceding inequalities and Proposition~\ref{prop:alpw24} gives
\begin{align*}\TV\Par{\widetilde{\sfM}(\vz, \cdot), \widetilde{\sfM}(\vw, \cdot)} &\le \half + \half \cdot \Par{1 + \TV\Par{\widetilde{\sfP}(\vz, \cdot), \widetilde{\sfP}(\vw, \cdot)}- \min\Brace{A(\vz), A(\vw)}} \\
&\le \half + \half \cdot\Par{1 + \frac{\Delta}{2\sigma} - \frac 1 8}= \frac{31}{32},\end{align*}
where the additive $\half$ comes from the lazy updates in Algorithm~\ref{alg:backwards_part2}.

Now, Proposition~\ref{prop:dcwy19} with $\Omega \gets \widetilde{\calX}$, $\eta \gets \frac1{32}$, $\calK \gets \widetilde{G}$, and $\alpha \gets \frac 1 4$,
yields for all Borel $B \subseteq \widetilde{\calX}$:
\begin{align*}
\int_B \widetilde{\sfM}(\vz, \widetilde{\calX} \setminus B) \widetilde{\mu}(\dd \vz)&\ge \frac{\log 2}{2^{19}\sqrt{n}} \min\Brace{\widetilde{\mu}\Par{B \cap \widetilde{G}}, \widetilde{\mu}\Par{\Par{\widetilde{\calX} \setminus B} \cap \widetilde{G}}} \\
&\ge \frac{\log 2}{2^{19}\sqrt{n}} \max\Brace{\min\Brace{\widetilde{\mu}\Par{B}, \widetilde{\mu}\Par{\widetilde{\calX} \setminus B }} - \delta, 0},
\end{align*}
where the second inequality used that $\delta = \widetilde{\mu}(\widetilde{\calX} \setminus \widetilde{G})$. The conclusion now follows from applying Lemma~\ref{lem:consequence_of_alpw26} and undoing the change of coordinates in \eqref{eq:whiten_part2}.
\end{proof}

We conclude the section with a runtime claim for Algorithm~\ref{alg:backwards_part2}.

\begin{lemma}\label{lem:part2_backward_runtime}
    Algorithm~\ref{alg:backwards_part2} can be implemented in time
    \[O\Par{mn^{\omega - 1}\Par{1 + \frac N n}}.\]
\end{lemma}
\begin{proof}
    Computing $\mh_{\vgamma}$ and $\mh_{\vgamma}^{-\half}$ takes $O(mn^{\omega-1})$ time. We claim that every $n$ iterations of the loop in Lines~\ref{line:loop_start} to~\ref{line:loop_end} can be implemented in time $O(mn^{\omega - 1})$, which gives the claim.

    For each batch of $n$ iterations, we generate $\mXi \in \R^{n \times n}$ with i.i.d.\ $\calN(0, 1)$ entries, and compute 
    \[\mDelta \defeq \frac \tau {8\sqrt{n(\tau + t)}}\mh_{\vgamma}^{-\half} \mXi,\]
    along with $\ma \mDelta$, in time $O(mn^{\omega - 1})$. We also maintain the current slack vector $\vs(\vx_i) = \vb - \ma \vx_i$ in each iteration. Now given $\ma \mDelta$, the difference $\vs(\vx') - \vs(\vx_i)$ for the proposal $\vx'$ for the $i^{\text{th}}$ iteration in the batch is just a negated column of $\ma\mDelta$, that can be evaluated in $O(m)$ time. Hence,
    \[\exp\Par{\inprod{\vgamma}{\ma(\vx' - \vx_i)} - (\tau + t)\Par{\psi(\vx') - \psi(\vx_i)}}\]
    can be evaluated in $O(m)$ time, because we computed $\ma(\vx' - \vx_i)$, and $\psi(\vx') - \psi(\vx_i)$ is an entrywise function of $\vs(\vx')$ and $\vs(\vx_i)$. We can also determine whether $\vx' \in \calX$ from $\vs(\vx')$. In sum, each loop takes $O(m)$ time after our stated precomputation, or $O(mn)$ total for the batch of $n$ steps.
\end{proof}

\subsection{Comparison to the ideal proximal sampler}\label{ssec:comparison}

In this section, we provide our improved $\chi^2$-to-TV oracle (Definition~\ref{def:cs_to_tv_oracle}), that functions even at $\tau = O(\log \frac m \delta)$. We present this oracle in Algorithm~\ref{alg:warm-start-LLT-MRW}.

\begin{algorithm2e}[ht!]
\caption{$\WarmLLTProxMRW(\vx_0, \tau, t, \delta)$}
\label{alg:warm-start-LLT-MRW}
\DontPrintSemicolon
\textbf{Input: }initial $\vx_0\in \calX$, regularization parameter $\tau\in \N$, cooling parameter $t \ge 0$,  target error $\delta \in (0, \half)$ \;
$N \gets \lceil C n \tau\rceil$\tcp*[r]{$C$ is the constant from Lemma~\ref{lem:mrw_to_ideal}.}
$K \gets \lceil 32 \min\{m\tau, 1 + \frac \tau t\}\log \frac 1 \delta\rceil $\;
\For {$i = 0, \ldots, K - 1$}{
$\vgamma_i\sim \pi_{t, \tau}^{\Gamma \mid X = \vx_i}$ \tcp*[r]{See Algorithm~\ref{alg:forwards}, retain only $\vgamma$.}\label{line:forwards}
$\vx_{i+1}\leftarrow \BackwardsMRW(\vx_i, \vgamma_i, \tau, t, N)$ \tcp*[r]{See Algorithm~\ref{alg:backwards_part2}.}
}
\textbf{Return: }$\vx_K$
\end{algorithm2e}

Algorithm~\ref{alg:warm-start-LLT-MRW} uses Algorithm~\ref{alg:backwards_part2} run for $O(n\tau)$ steps as its approximate backwards oracle. Its analysis compares Algorithm~\ref{alg:warm-start-LLT-MRW} to its ideal implementation in Algorithm~\ref{alg:prox-sample}. In particular, it combines Lemma~\ref{lem:mrw}, which gives an instancewise defectiveness guarantee for Algorithm~\ref{alg:backwards_part2}, with Lemma~\ref{lem:concentration_part2}, which proves that this defectiveness is small on average.

\begin{lemma}\label{lem:mrw_to_ideal}
Let $t \ge 0$, and let $\tau, N \in \N$. Let $\sfM_{\vgamma}$ denote the Markov operator corresponding to Algorithm~\ref{alg:backwards_part2} run with inputs $(\vgamma, \tau, t, N)$, and define for all bounded $h: \calX \to \R$,
\begin{equation}\label{eq:approx_prox}(\tsfP h)(\vx) \defeq \E_{\Gamma \sim \pi_{t,\tau}^{\Gamma \mid X = \vx}}\Brack{(\sfM_\Gamma h)(\vx)}.\end{equation}
There is a universal constant $C > 0$ such that, if $N \ge Cn\tau$, then for all bounded $h: \calX \to \R$,
\[ \Var_{\pi_t}\Brack{\tsfP h} \le \frac{1}{\Par{1 + \frac{\alpha_t}{\tau}}^2}\Var_{\pi_t}[h] + 5\sqrt{m}\exp\Par{-\frac{\tau}{128} }\osc(h)^2,\text{ for } \alpha_t \defeq \max\Brace{t, \frac 1 {16m}}.\]
\end{lemma} 
\begin{proof}
We take $C$ to be the same constant as in Lemma~\ref{lem:mrw}.
We next set up some notation. Let $\sfF$ denote the Markov operator of the exact forward step (Algorithm~\ref{alg:forwards} with output $\vgamma$), let $\sfB$ denote the corresponding operator for the exact backward step, and let $\sfP \defeq \sfF\sfB$ denote the overall Markov operator of Algorithm~\ref{alg:prox-sample}, all with step size $\tau$ and applied to $\pi_t$. Let $\tsfB_{\vgamma}$ denote the inexact backwards step Markov operator corresponding to Algorithm~\ref{alg:backwards_part2} when run with inputs $(\vgamma, \tau, t, N)$. 

The key observation is that the operator $\tsfP$ in \eqref{eq:approx_prox} is exactly the Markov operator of the overall inexact proximal sampler that uses $\tsfB$ in place of $\sfB$, i.e., $\tsfP = \sfF \tsfB$. To see this, we formally have
\[\sfF(\vx, \dd(\vx', \vgamma)) = \delta_{\vx}(\vx') \pi_{t,\tau}^{\Gamma \mid X = \vx} (\dd \vgamma),\quad \tsfB((\vx, \vgamma),\vx') = \sfM_{\vgamma}(\vx, \dd \vx'),\]
and hence
\[(\sfF\tsfB h)(\vx) = \int \sfF(\vx, \dd(\vx', \vgamma))(\tsfB h)(\vx', \vgamma) =  \E_{\vgamma \sim \pi_{t,\tau}^{\Gamma \mid X = \vx}}\Brack{(\sfM_\Gamma h)(\vx)}= (\tsfP h)(\vx).\]
We also record that $\sfP$ similarly satisfies
\[(\sfP h)(\vx) \defeq \E_{\Gamma \sim \pi_{t, \tau}^{\Gamma \mid X = \vx}} \Brack{\E_{X' \sim \pi_{t, \tau}^{X \mid \Gamma}}\Brack{h(X')}},\]
i.e., it replaces $\sfM_\Gamma$ in \eqref{eq:approx_prox} with a direct sample $X'$ from the backwards conditional distribution induced by $\Gamma$. Comparing the above display and \eqref{eq:approx_prox}, and applying Jensen's inequality, shows
\[\Par{(\tsfP h - \sfP h)(\vx)}^2 \le \E_{\Gamma \sim \pi_{t, \tau}^{\Gamma \mid X = \vx}} \Brack{\Par{(\sfM_\Gamma h)(\vx) - \E_{X' \sim \pi_{t, \tau}^{X \mid \Gamma}}\Brack{h(X')}}^2}.\]
Taking averages of both sides under $\vx \sim \pi_t$ and reversing the order of conditioning then gives
\begin{align*}\E_{\pi_t}\Brack{\Par{\tsfP h - \sfP h}^2} &\le \E_{\Gamma \sim \pi_{t,\tau}^\Gamma}\Brack{\E_{X \sim \pi_t^{X \mid \Gamma}}\Brack{\Par{(\sfM_\Gamma h)(X) - \E_{X' \sim \pi_{t, \tau}^{X \mid \Gamma}}\Brack{h(X')}}^2}} \\
&= \E_{\Gamma \sim \pi_{t,\tau}^\Gamma}\Brack{\Var_{\pi_{t, \tau}^{X \mid \Gamma}}\Brack{\sfM_\Gamma h}},
\end{align*}
where we used that $\pi_{t, \tau}^{X \mid \Gamma}$ is stationary for $\sfM_\Gamma$. 
Now continuing, we have from Lemma~\ref{lem:mrw} that
\begin{equation}\label{eq:comparison}
\begin{aligned}
\E_{\pi_t}\Brack{\Par{\tsfP h - \sfP h}^2} &\le \E_{\Gamma \sim \pi_{t,\tau}^\Gamma}\Brack{\exp\Par{-\frac{N}{Cn}}\Var_{\pi_{t,\tau}^{X \mid \Gamma}}[h] + 2\pi_{t,\tau}^{X \mid \Gamma}\Par{\calX \setminus \calB_\tau\Par{\Gamma, \frac 1 4}}\osc(h)^2}  \\
&\le \exp\Par{-\frac{N}{Cn}}\Var_{\pi_t}[h] + 4m\exp\Par{-\frac \tau {64}}\osc(h)^2 \\
&\le 5m\exp\Par{-\frac \tau {64}}\osc(h)^2,
\end{aligned}
\end{equation}
for a sufficiently large $C$, where the second line applied Lemma~\ref{lem:concentration_part2}, and the last line used $N \ge Cn\tau$, and $\Var_{\pi_t}[h] \le \osc(h)^2$. On the other hand, $\pi_t$ satisfies an $\alpha_t$-$\psi$-PI (see Fact~\ref{fact:bl} and Remark~\ref{rem:slc}), so Proposition~\ref{prop:prox-sample-guarantee} gives
\begin{equation}\label{eq:exact_contraction}\Var_{\pi_t}[\sfP h] \le \frac{1}{\Par{1 + \frac{\alpha_t}{\tau}}^2}\Var_{\pi_t}[h].\end{equation}
Finally, assume without loss of generality that $\E_{\pi_t}[h] = 0$, else shift $h$ by this amount (which does not affect the variance of $h$, $\sfP h$, or $\tsfP h$ because $\pi_t$ is stationary for $\sfP$, $\tsfP$). Then,
\begin{align*}
    \operatorname{Var}_{\pi_t}\Brack{\tsfP h}-\operatorname{Var}_{\pi_t}\Brack{\sfP h}
    &=\mathbb{E}_{\pi_t}\Brack{\Par{\tsfP h}^2-\Par{\sfP h}^2}\\
    &=\mathbb{E}_{\pi_t}\Brack{\Par{\tsfP h-\sfP h}\Par{\tsfP h+\sfP h}}\\
    &\leq \sqrt{\mathbb{E}_{\pi_t}\Brack{\Par{\tsfP h-\sfP h}^2}}\cdot \Par{\sqrt{\mathbb{E}_{\pi_t}\Brack{\Par{\tsfP h}^2}}+\sqrt{\mathbb{E}_{\pi_t}\Brack{\Par{\sfP h}^2}}} \\
    &\le \sqrt{5m}\exp\Par{-\frac \tau {128}}\osc(h) \cdot 2\osc(h) \le 5\sqrt{m}\exp\Par{-\frac \tau {128}}\osc(h)^2,
\end{align*}
where the third line used the Cauchy-Schwarz inequality, and the fourth used \eqref{eq:comparison} and $\E_{\pi_t}[(\tsfP h)^2] \le \osc(h)^2$, $\E_{\pi_t}[(\sfP h)^2] \le \osc(h)^2$. The claim follows from the above display and \eqref{eq:exact_contraction}.
\end{proof}

Patterning off of Lemma~\ref{lem:warm-llt-prox-guarantee}, we now prove that Algorithm~\ref{alg:warm-start-LLT-MRW} is a $\chi^2$-to-TV oracle.
\begin{proposition}\label{prop:backwards_part2}
Let $t \ge 0$, $\delta \in (0, \half)$, let $\pi_t$ be defined as in \eqref{eq:cooling-path}.
    Further, suppose $\tau \in \N$ satisfies, for an absolute constant $C > 0$, and $K \defeq \lceil 32\min\{m\tau, 1 + \frac \tau t\}\log \frac 1 \delta\rceil$,
    \begin{equation*}
    \tau \ge C \log \frac{m}{\delta}
    \end{equation*}
    Then
    Algorithm~\ref{alg:warm-start-LLT-MRW}, as a function of its input $\vx_0$, is a $(3,\delta)$-$\chi^2$-to-TV oracle for $\pi_t$ that runs in time
    \[O\Par{K\tau mn^{\omega - 1}}.\]
\end{proposition}
\begin{proof}
The runtime is immediate from the iteration count, 
Proposition~\ref{prop:forward}, and Lemma~\ref{lem:part2_backward_runtime}.

Correctness follows from a similar argument as the proof of Lemma~\ref{lem:warm-llt-prox-guarantee}. Let $\{\vx_i\}_{i\in[K]}$ be the iterates produced by Algorithm~\ref{alg:warm-start-LLT-MRW}, and let $\tsfP$ denote its one-step Markov kernel. Define \[
\mu_0 \defeq \Law(\vx_0), \quad \mu_i \defeq \mu_0\tsfP^i = \Law(\vx_i), \quad \alpha_t\defeq \max\Brace{t, \frac{1}{16m}}.
\]
Iterating Lemma~\ref{lem:mrw_to_ideal} gives, for every bounded $h$, 
\begin{equation}\label{eq:lossy_contraction}
\begin{aligned}
\Var_{\pi_t}[\tsfP^K h]&\le \frac{1}{\Par{1 + \frac{\alpha_t}{\tau}}^{2K}}\Var_{\pi_t}[h] + K\cdot 5\sqrt{m}\exp\Par{-\frac \tau {128}}\osc(h)^2 \\
&\le \frac{1}{\Par{1 + \frac{\alpha_t}{\tau}}^{2K}}\Var_{\pi_t}[h] + 165 m^{1.5}\tau \log \Par{\frac 1 \delta} \exp\Par{-\frac \tau {128}}\osc(h)^2,
\end{aligned}
\end{equation}
where we used the first term in the bound for $K$. Now for any measurable $B\subseteq \calX$, 
\begin{align*}
    |\mu_K(B) - \pi_t(B)| &= \Abs{(\mu_0 - \pi_t)(\tsfP^K \ind_{B})}\\
    &= \Abs{\E_{\pi_t}\Brack{\frac{\dd\mu_0}{\dd\pi_t} - 1}\tsfP^K \ind_B}\\
    &\le \sqrt{\chi^2(\mu_0\|\pi_t)}\sqrt{\Var_{\pi_t}[\tsfP^K \ind_B]}\\
    &\le \frac{\sqrt{3}}{2\Par{1+\frac{\alpha_t}{\tau}}^K} + \sqrt{495 m^{1.5} \log \Par{\frac 1 \delta} \exp\Par{-\frac \tau {128}}}\le \delta.
\end{align*}
Above, the first line used $\pi_t \tsfP^K = \pi_t$, the second line used the Cauchy-Schwarz inequality, and the last line used \eqref{eq:lossy_contraction} with $h \gets \ind_B$ and that $\sqrt{a + b} \le \sqrt{a} + \sqrt{b}$ and $\chi^2(\mu_0\|\pi_t)\le 3$. The last line also used that the variance of a function with range $[0, 1]$ is at most $\frac 1 4$. By adjusting the constant $C$, we ensure both terms in the final line are $\le \frac{\delta}{2}$, concluding the proof.
\end{proof}

\subsection{Proof of Theorem~\ref{thm:main_part2}}\label{ssec:main_proof}

\begin{algorithm2e}[ht!]
\caption{$\PTERRYS(\calX, \delta)$}
\label{alg:main}
\DontPrintSemicolon
\textbf{Input: }bounded, full-dimensional polytope $\calX \subset \R^n$ with $m$ facets, target error $\delta \in (0, \half)$ \;
$\tau \gets \lceil C\log \frac m \delta\rceil$\tcp*[r]{$C$ is the constant from Theorem~\ref{thm:main_part2}.}
$T \gets C' (m + \sqrt{n}\log(\frac m \delta))$\tcp*[r]{$C'$ is the constant from Corollary~\ref{cor:part-1-iteration-count}.}
\textbf{Return: }$\FunctionalCooling(\delta, \{\WarmLLTProxMRW(\cdot, \tau, t, \frac{\delta}{2T})\}_{t\ge 0})$
\tcp*[r]{See Algorithms~\ref{alg:annealing},~\ref{alg:warm-start-LLT-MRW}.}
\end{algorithm2e}

\restatepartII*

\begin{proof}
Let $\tau$ be as specified and $T = O(m + \sqrt{n}\log(\frac{m}{\delta}))$ be the number of cooling stages guaranteed by Corollary~\ref{cor:part-1-iteration-count}, which we note does not depend on what the value of $\tau$ is.

By Proposition~\ref{prop:backwards_part2}, Algorithm~\ref{alg:warm-start-LLT-MRW} implements a $(3,\frac{\delta}{2T})$ $\chi^2$-to-TV oracle for every $\pi_t$. Then the accuracy guarantee is immediate by Lemma~\ref{lem:annealing_correct}.

Now we check the runtime. Because the constants in the iteration count $K$ in Algorithm~\ref{alg:warm-start-LLT-MRW} are exactly the same as those of $K$ in Algorithm~\ref{alg:warm-start-LLT} as a function of $t$, the total number of iterations $K$ of $\WarmLLTProxMRW$ throughout Phase 2, 3, and 4 by~\eqref{eq:iterations_phase2},~\eqref{eq:iterations_phase3}, and Lemma~\ref{lem:phase-4-iterations}, are respectively\begin{align*}
    &O\Par{\Par{\sqrt n \log\frac m \delta + \frac{\tau m}{\sqrt n}}\log\frac{m}{\delta}},\quad O\Par{m\tau\log\frac m \delta},\quad
    O\Par{m\tau\log \frac m \delta},
\end{align*}
so the total number of iterations $K$ in total of $\WarmLLTProxMRW$ is \[
   O\Par{\Par{m\tau + \sqrt{n}\log \frac{m}{\delta}}\log \frac{T}{\delta}} = O\Par{m\log^2 \frac{m}{\delta}}.
   \]
   Combining this with the runtime guarantees in Propositions~\ref{prop:ac} and~\ref{prop:backwards_part2} yields the theorem.
\end{proof}

\section*{AI Disclosure}
\phantomsection
\addcontentsline{toc}{section}{AI Disclosure}

The key algorithmic ideas in Part~\ref{part:sqrtn} were the contributions of the authors. Specifically, the authors observed that the logarithmic barrier was an LLT (and hence amenable to the analysis of \cite{GuST26}), came up with the ideas to center the backwards oracles at the mode (motivated by \cite{FanYC23}), and to apply functional cooling (based on the framework of \cite{Kook2024Gaussian}). We used ChatGPT 5.6 Pro to explore strategies for implementing the backwards oracle, and it materially affected Sections~\ref{ssec:concentration} and~\ref{ssec:smoothing}, particularly the proofs of our mean relative slack bounds (Lemmas~\ref{lem:infty_expect} and~\ref{lem:mean_bound}) and the idea to use a globally smoothed measure to simplify the analysis.

Upon completing a draft of Part~\ref{part:sqrtn}, the authors asked GPT 6 Pro on whether an improved $\tau$ range was possible, and it produced a proof of (essentially) the ideas in Part~\ref{part:tilde1}, based on the observation that the analyses of \cite{AndrieuLPW24, AndrieuLPW26} are directly compatible with properties of our backwards oracle (Lemma~\ref{lem:concentration_part2}). The authors' contribution to Part~\ref{part:tilde1} is primarily in giving a streamlined presentation, and providing a unified framework with Part~\ref{part:sqrtn}. We chose to split the paper based on our two backwards oracle implementations, and to reflect the provenance of the ideas. 

The manuscript was written solely by the authors, who take full responsibility for the organization and presentation of all results.

\newpage

\bibliographystyle{alpha}
\bibliography{refs}

\end{document}